\documentclass[11pt]{article}
\usepackage[margin=1in]{geometry}
\usepackage[utf8]{inputenc}
\usepackage[absolute]{textpos}
\usepackage{microtype}
\usepackage[all=normal,bibliography=tight]{savetrees}
\usepackage{listings}
  \usepackage{mathrsfs}

  \usepackage{amstext,amsfonts,amsthm,amsmath,amssymb}

\usepackage{graphicx,tikz}
\usepackage{comment}
\usepackage{url}
\usepackage{xspace}
\usepackage{todonotes}
\usepackage{enumerate}

\def\cqedsymbol{\ifmmode$\lrcorner$\else{\unskip\nobreak\hfil
\penalty50\hskip1em\null\nobreak\hfil$\lrcorner$
\parfillskip=0pt\finalhyphendemerits=0\endgraf}\fi} 

\newcommand{\cqed}{\renewcommand{\qed}{\cqedsymbol}}

\newtheorem{theorem}{Theorem}
\newtheorem{lemma}{Lemma}
\newtheorem{proposition}[lemma]{Proposition}
\newtheorem{observation}[lemma]{Observation}
\newtheorem{corollary}[lemma]{Corollary}
\newtheorem{claim}{Claim}
\theoremstyle{definition}
\newtheorem{definition}{Definition}

\newcommand{\pmca}{(PMC1)}
\newcommand{\pmcb}{(PMC2)}
\newcommand{\Oh}{\mathcal{O}}

\newcommand{\cc}{\mathsf{cc}}
\newcommand{\Cc}{\mathcal{C}}
\newcommand{\Dd}{\mathcal{D}}
\newcommand{\Ss}{\mathcal{S}}
\newcommand{\Ff}{\mathcal{F}}
\newcommand{\Gg}{\mathcal{G}}
\newcommand{\Xx}{\mathcal{X}}
\newcommand{\Yy}{\mathcal{Y}}
\newcommand{\Om}{\Omega}
\newcommand{\bag}{\beta}
\newcommand{\adh}{\sigma}

\newcommand{\Mod}{\mathsf{Mod}}
\newcommand{\Quo}{\mathsf{Quo}}
\newcommand{\Reach}{\mathsf{Reach}}
\newcommand{\Proj}{\mathsf{Proj}}
\newcommand{\tricky}{\mathsf{Tricky}}

\newcommand{\hid}{\mathrm{hidden}}
\newcommand{\rec}{\mathrm{rec}}

\newcommand{\NP}{$\mathsf{NP}$}

\newcommand{\wD}{\widehat{D}}
\newcommand{\wOm}{\widehat{\Om}}
\newcommand{\wG}{\widehat{G}}

\graphicspath{{figures/}}

\title{Polynomial-time algorithm for Maximum Weight Independent Set on $P_6$-free graphs\thanks{
The majority of research for this work was done while Andrzej Grzesik held a post-doc position of Warsaw Centre of Mathematics and Computer Science (WCMCS).
WCMCS also supported a visit of Tereza Klimo\v{s}ov\'a in Warsaw, during which this research was initiated.
The first part of work of Ma. Pilipczuk is supported by the Polish National Science Centre grant UMO-2013/09/B/ST6/03136.
Mi.\ Pilipczuk was supported by the Foundation for Polish Science (FNP) via the START stipend programme.
Later stages of this research are parts of projects that have received funding from the European Research Council (ERC) under the European Union's Horizon 2020 research and innovation programme
Grant Agreements no.~714704 (Ma. Pilipczuk) and no.~677651 (Mi. Pilipczuk).
Klimo\v sov\'a was supported by ANR project Stint under reference ANR-13-BS02-0007
and by the LABEX MILYON (ANR-10-LABX-0070) of Universit\'e de Lyon, within the program “Investissements
d’Avenir” (ANR-11-IDEX-0007) operated by the French National Research Agency (ANR), by Center of Excellence -- ITI, project P202/12/G061 of GA \v{C}R and by Center for Foundations of Modern Computer Science (Charles Univ. project UNCE/SCI/004).}}

\author{ 
  Andrzej Grzesik\thanks{
    Faculty of Mathematics and Computer Science, Jagiellonian University, Krak\'ow, Poland, \texttt{andrzej.grzesik@tcs.uj.edu.pl}.
  }
  \and 
  Tereza Klimo\v{s}ov\'a\thanks{
    Department of Applied Mathematics, Charles University,  Prague, Czech Republic, \texttt{tereza@kam.mff.cuni.cz}.
  }
  \and 
  Marcin Pilipczuk\thanks{
    Institute of Informatics, University of Warsaw, Poland, \texttt{marcin.pilipczuk@mimuw.edu.pl}.
  }
  \and 
  Micha\l{} Pilipczuk\thanks{
    Institute of Informatics, University of Warsaw, Poland, \texttt{michal.pilipczuk@mimuw.edu.pl}.
  }
}

\date{}

\begin{document}
\pagenumbering{gobble}
\thispagestyle{empty}

\maketitle

\begin{textblock}{20}(0, 11.0)
\includegraphics[width=40px]{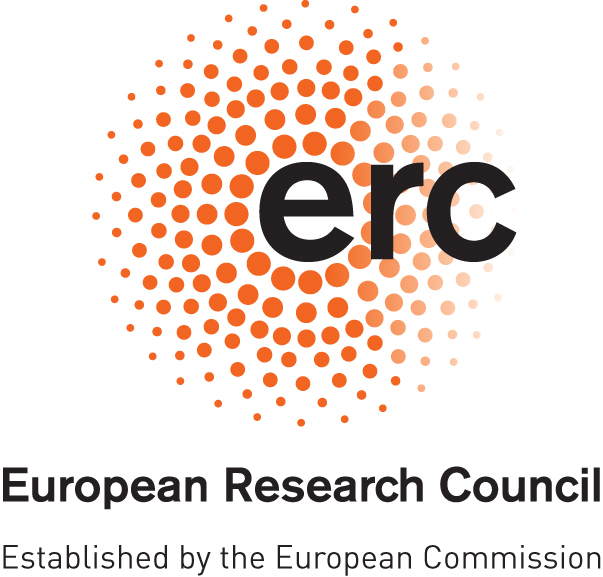}%
\end{textblock}
\begin{textblock}{20}(-0.25, 11.4)
\includegraphics[width=60px]{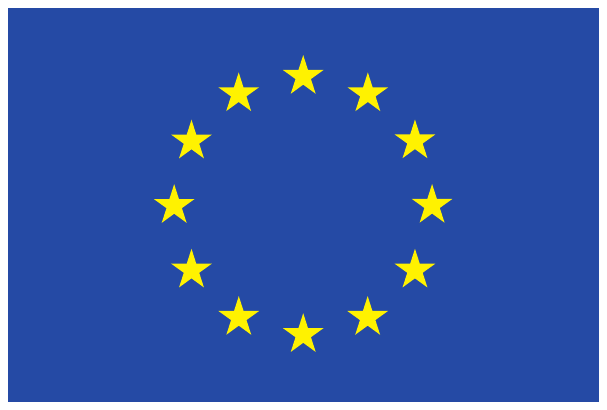}%
\end{textblock}

\begin{abstract}
In the classic {\sc{Maximum Weight Independent Set}} problem we are given a graph $G$ with a nonnegative weight function on its  vertices, 
and the goal is to find an independent set in $G$ of maximum possible weight.
While the problem is NP-hard in general, we give a polynomial-time algorithm working on any $P_6$-free graph, that is, a graph that has no path on $6$ vertices as an induced subgraph.
This improves the polynomial-time algorithm on $P_5$-free graphs of Lokshtanov et al.~\cite{LokshtanovVV14}, and the quasipolynomial-time algorithm on $P_6$-free graphs of Lokshtanov et al.~\cite{LokshtanovPL16}.
The main technical contribution leading to our main result is enumeration of a polynomial-size family $\Ff$ of vertex subsets with the following property: 
for every maximal independent set $I$ in the graph, $\Ff$ contains all maximal cliques of some minimal chordal completion of $G$ that does not add any edge incident to a vertex of $I$.

\end{abstract}

\clearpage

\pagenumbering{arabic}

\section{Introduction}\label{sec:intro}
All graphs considered in this paper are finite and simple, i.e., without multiedges or loops.
A subset $I$ of vertices of a graph $G$ is {\em{independent}} if the vertices of $I$ are pairwise non-adjacent.
The {\sc{Maximum Weight Independent Set}} ({\sc{MWIS}}) problem asks, for a given graph $G$ with nonnegative weights assigned to its vertices, for an independent set in $G$ that has the maximum possible total weight.
The problem is \NP-hard in general graphs~\cite{Karp72},
even in the case of uniform weights. Therefore, the study of {\sc{MWIS}} on restricted classes of inputs, like planar, sparse, or well-decomposable graphs,
is a recurring topic in the algorithm design.

In this work  we focus on restricting the input graph to a {\em{hereditary}} graph class, that is, a class closed under taking induced subgraphs.
A considerable amount of work has been devoted to this direction. Perhaps the most prominent result here is the polynomial-time solvability of the problem on the class of perfect graphs
using linear programming methods~\cite{GrotschelLS81}. We refer to the introductory sections of~\cite{LokshtanovVV14, LokshtanovPL16} for a wider discussion of the literature.

While a complete classification of the complexity of {\sc{MWIS}} on every hereditary graph class is most probably out of reach,
one can focus on classes of {\em{$H$-free graphs}}, that is, graphs that do not admit a fixed graph $H$ as an induced subgraph.
Alekseev~\cite{Alekseev82} proved that the {\sc{MWIS}} problem remains \NP-hard on $H$-free graphs unless every connected component of $H$ is a subdivision of a claw or a path;
this leaves only two simple families to consider.

Unfortunately, even in this restricted setting, only small progress has been achieved.
For excluding induced subdivisions of a claw, polynomial-time algorithms are known for claw-free graphs~\cite{Minty80,Sbihi80} and fork-free graphs~\cite{LozinM08}.
For excluding induced paths, there is a polynomial-time algorithm on $P_4$-free graphs (also known as {\em{cographs}}), observed by Corneil~\cite{CorneilLB81} in 1981.
Note that cographs have bounded cliquewidth, so a simple dynamic programming algorithm suffices.
Recently, Lokshtanov, Vatshelle, and Villanger~\cite{LokshtanovVV14} managed to give a polynomial-time algorithm for $P_5$-free graphs, thus making the first breakthrough in over 30 years.
Shortly later, Lokshtanov, Pilipczuk, and van Leeuwen~\cite{LokshtanovPL16} gave a quasipolynomial-time algorithm for $P_6$-free graphs, with running time $n^{\Oh(\log^2 n)}$.
However, the question whether the problem can be solved in polynomial time on $P_6$-free graphs remained open.
To the best of our knowledge, for all the other paths and subdivisions of a claw, it is still open whether the problem is polynomial-time solvable.

\paragraph*{Our contribution.} In this work we settle the complexity of {\sc{MWIS}} on $P_6$-free graphs by showing that it can be solved in polynomial time.

\begin{theorem}\label{thm:main-mwis}
The {\sc{Maximum Weight Independent Set}} problem can be solved in polynomial time on $P_6$-free graphs.
\end{theorem}

To prove Theorem~\ref{thm:main-mwis}, at high a level we employ the same methodology that led Lokshtanov et al.~\cite{LokshtanovVV14} to the polynomial-time algorithm on $P_5$-free graphs.
We now discuss the intuition behind this approach.

The idea is to find a maximum-weight independent set using dynamic programming on the structure of the input graph $G$.
Imagine that the graph admits some inclusion-wise maximal independent set $I$, unknown to the algorithm.
Since $I$ is independent, it is not hard to see that there exists a tree decomposition of the graph where every bag has at most one element of $I$.
If such a tree decomposition was given to us, for some maximum-weight independent set $I$, then we would be able to reconstruct $I$ (or some other independent set of the same weight) using dynamic
programming in polynomial time. 
Roughly, a state of the dynamic program would be formed by a bag of the decomposition together with at most one of its elements, interpreted as the intended intersection of the bag with 
the constructed independent set.

However, a priori we cannot assume that we are given such a useful tree decomposition.
Instead, the idea is to try to compute a rich enough family of candidates for bags, in hope that for some tree decomposition that can lead to the discovery of a maximum-weight solution,
its bags will be included among the enumerated candidates. Then we would be able to apply a similar dynamic programming procedure, which intuitively tries to
assemble all possible tree decompositions using given candidate bags, and compute a maximum-weight independent set along the way.
This family should be rich enough so that all bags of some useful decomposition are captured, implying that this decomposition yields a computation path leading to the discovery of an optimum solution,
but also small enough so that the whole algorithm runs in polynomial time.

Fortunately for us, there is a well-understood theory of so-called {\em{potential maximal cliques}} ({\em{PMCs}}), which are essentially candidates for bags of useful tree decompositions.
Observe that given a tree decomposition as described above, by turning each of its bags into a clique, we obtain a chordal supergraph $G+F$ of $G$. 
Here, $F$ is the {\em{chordal completion}}, or {\em{fill-in}}: the set of edges added to $G$ in order to turn it into a chordal graph. 
We can further require that this chordal completion $F$ is (a) $I$-free, meaning that it does not contain any edge incident to a vertex of $I$, and (b) 
minimal, meaning that there is no other chordal completion of $G$ that is a proper subset of $F$.

The following lemma formally summarizes the approach of Lokshtanov et al.~\cite{LokshtanovVV14}. 
As noted in~\cite{LokshtanovVV14}, the proof of this lemma is implicit in the earlier work of Fomin and Villanger~\cite{FominV10}, from which Lokshtanov et al. drew inspiration.

\begin{lemma}[Lokshtanov et al.~\cite{LokshtanovVV14}, based on Fomin and Villanger~\cite{FominV10}]\label{lem:dp}
There is an algorithm that given a graph $G$ with nonnegative weights assigned to its vertices, and a family of vertex subsets $\Ff$, works in time polynomial in the size of $G$ and $\Ff$
and finds a maximum-weight independent set in $G$ provided the following condition holds: for every inclusion-wise maximal independent set $I$ in $G$, there exists an $I$-free
minimal chordal completion $F$ of $G$ such that all maximal cliques of $G+F$ belong to $\Ff$.
\end{lemma}

Thus, Lemma~\ref{lem:dp} reduces solving {\sc{MWIS}} to enumerating candidates for maximal cliques in some $I$-free minimal chordal completion of $G$.
Candidates for such maximal cliques are called {\em{potential maximal cliques}}, or just {\em{PMCs}}.
Formally, a subset $\Om$ is a PMC if there is a minimal chordal completion of the graph in which $\Om$ is a maximal clique.
The theory of minimal separators and PMCs was pioneered by Bouchitt\'e and Todinca~\cite{BouchitteT01,BouchitteT02}, and provides us with many tools useful for finding the candidates.
In particular, it can be checked whether a set is a PMC in polynomial time by verifying a handful of combinatorial properties, and there are multiple techniques for finding and enumerating PMCs in graphs.

Unfortunately, it was already observed by Lokshtanov et al.~\cite{LokshtanovVV14} that the total number of PMCs even in a $P_5$-free graph can be exponential, so we cannot enumerate
all the PMCs as candidates. For this reason, Lokshtanov et al.~\cite{LokshtanovVV14} performed a structural analysis that revealed that there is essentially only one type of PMCs whose number can be exponential,
and these PMCs ``live'' in closed neighborhoods of pairs of vertices of $I$. The idea now is that we do not need to enumerate all possible such PMCs, but it suffices to add to the family of candidates
the maximal cliques of an arbitrarily chosen minimal chordal completion of the neighborhood of every pair of vertices. Then an exchange argument shows that every minimal chordal completion can be modified to
one that uses only the enumerated maximal cliques.

In our setting of $P_6$-free graphs, we follow the same high-level approach of reducing {\sc{MWIS}} to enumerating candidates for maximal cliques in some $I$-free minimal chordal completion of the input graph.
More precisely, throughout the whole paper we will focus on proving the following result.

\begin{theorem}\label{thm:main}
Given a $P_6$-free graph $G$, one can in polynomial time compute a polynomial-size family $\Ff$ of vertex subsets satisfying the following:
for every maximal independent set $I$ in $G$, there exists an $I$-free minimal chordal completion $F$ of $G$ such that $\Ff$ contains all maximal cliques of $G+F$.
\end{theorem}

Then Theorem~\ref{thm:main-mwis} follows by combining Theorem~\ref{thm:main} with Lemma~\ref{lem:dp}. 

One of main differences between the setting of $P_5$-free graphs and our setting is that in our case, the variety of situations where one cannot enumerate a polynomial number of candidate PMCs is far richer.
Such ``unguessable'' PMCs may occur not only near vertices of $I$, as was the case in~\cite{LokshtanovVV14}, but all over the clique tree of the completed chordal graph.
Intuitively, our main goal is to have polynomial-sized families of candidates for as many types of PMCs as possible. 
For those types of PMCs, for which such families cannot be found due to the potentially exponential number of candidates,
we would like to gain a very good combinatorial understanding of the situation around them.
This is in order to give appropriate exchange arguments in a similar, but more general spirit as the exchange argument for neighborhoods of vertices from~\cite{LokshtanovVV14}.

Unfortunately, the treatment of the $P_6$-free case is far more involved than the $P_5$-free case considered by Lokshtanov et al.~\cite{LokshtanovVV14}.
Not only we have to treat many more types of PMCs for which there are no polynomial-size families of candidates, but the structural analysis leading to their classification,
as well as enumeration of PMCs of those types that can be enumerated, requires far deeper and more complicated structural analysis.
One technique that we use here, and which was used in the quasipolynomial-time algorithm of Lokshtanov et al.~\cite{LokshtanovPL16}, is the {\em{modular decomposition}}.
This tool turns out to be invaluable for analysing the internal structure of components of the graph after removing a PMC.

Finally, we remark that our approach is far closer to lifting the polynomial-time algorithm for the $P_5$-free case of Lokshtanov et al.~\cite{LokshtanovVV14}, than
improving the running time of the quasipolynomial-time algorithm of Lokshtanov et al.~\cite{LokshtanovPL16}. This is because the latter algorithm is based on branching. The main goal there is
to find a vertex  such that the removal of its closed neighborhood shatters the graph into connected components of significantly smaller size.
Pursuing branching on whether to include such a vertex to a constructed independent set or not leads to quasipolynomial running time of the algorithm.
However, it seems very hard to use such strategy for designing a polynomial-time algorithm.
Therefore, from~\cite{LokshtanovPL16} we borrow only some technical tools related to the structural analysis of $P_6$-free graphs.

\paragraph*{Organization.}
First, in Section~\ref{sec:prelims} we establish the notation and recall definitions and known facts from the literature.
We start our proof with an overview in Section~\ref{sec:over}, where we present crucial technical tools and explain key steps of the reasoning.
In Section~\ref{sec:toolbox} we introduce some auxiliary tools of general usage.
Section~\ref{sec:capturing} describes types of PMCs for which we can find polynomial-size families of candidates.
In Section~\ref{sec:modifying} we treat the other PMCs via exchange arguments and prove the main result.

\section{Preliminaries}\label{sec:prelims}
%We provide now full preliminaries, repeating some content from the overview.

\paragraph*{Notation.} All graphs considered in this paper are finite, simple, and undirected, unless explicitly stated.
For a graph $G$, by $V(G)$ and $E(G)$ we denote the vertex and edge set of $G$, respectively.
An edge connecting vertices $u$ and $v$ will be denoted by $uv$; if $uv$ does not belong to $E(G)$, then we will say that $uv$ is a {\em{non-edge}} in $G$.
For a vertex $v$ or a vertex subset $X$, we write $v\in G$ and $X\subseteq G$ meaning $v\in V(G)$ and $X\subseteq V(G)$, respectively.
For a subset of vertices $X$, by $G[X]$ we denote the subgraph induced by $X$, and $G-X$ denotes the induced subgraph $G[V(G)\setminus X]$.
When $X=\{v\}$ for some vertex $v$, we use the shorthand $G-v$.
The set of connected components of a graph $G$ will be denoted by $\cc(G)$.
A {\em{clique}} in a graph is a set of pairwise adjacent vertices.
A clique is {\em{maximal}} if no its proper superset is also a clique.

The {\em{open neighborhood}} of a vertex $u$ in a graph $G$ comprises all neighbors of $u$ in $G$, and is denoted by $N_G(u)$.
The {\em{closed neighborhood}} of $u$ in $G$ is then defined as $N_G[u]=N_G(u)\cup \{u\}$.
This terminology is extended to open and closed neighborhoods of any vertex subset $X\subseteq V(G)$ as follows: $N_G[X]=\bigcup_{u\in X}N_G[u]$ and $N_G(X)=N_G[X]\setminus X$.
In case $X$ consists of vertices $u_1,\ldots,u_k$, we may write $N_G[u_1,\ldots,u_k]$ and $N_G(u_1,\ldots,u_k)$ instead of $N_G[X]$ and $N_G(X)$.
Whenever the graph $G$ is clear from the context, we may omit the subscript.

For a vertex subset $X\subseteq V(G)$ in a graph $G$, and a vertex $u\in V(G)\setminus X$, we will say that $u$ is {\em{complete to}} $X$ if $X\subseteq N_G(u)$.
Disjoint vertex subsets $X,Y\subseteq V(G)$ are called {\em{complete to each other}} if each vertex of $X$ is complete to $Y$, equivalently each vertex of $Y$ is complete to $X$.
Similarly, $X$ and $Y$ are {\em{anti-complete}} if there is no edge with one endpoint in $X$ and second in $Y$.
By convention, if $X=\emptyset$ then $X$ is both complete and anti-complete to $Y$.

For a positive integer $k$, a {\em{path}} $P_k$ is a graph with vertex set $\{v_1,\ldots,v_k\}$ and the edge set $\{v_1v_2,\ldots,v_{k-1}v_k\}$, for some $v_1,v_2,\ldots,v_k$.
Such $P_k$ will be denoted by $(v_1,\ldots,v_k)$.
An {\em{induced $P_k$}} in a graph $G$ is an induced subgraph in $G$ that is a $P_k$.
A graph $G$ is {\em{$P_k$-free}} if it does not contain any induced $P_k$.
To facilitate the proofs, we introduce the following notation. Suppose $G$ is a graph and $X_1,X_2,\ldots,X_k\subseteq G$ are vertex subsets. If $P=(v_1,v_2,\ldots,v_k)$ is a $P_k$,
then $P$ will be {\em{of the form $X_1X_2\ldots X_k$}} if $v_i\in X_i$ for all $i=1,2,\ldots,k$. If $X_i=\{v\}$ for some vertex $v$, we may put vertex $v$ instead of $X_i$ in the sequence denoting the form.
For instance, a $P_4$ of the form $vAAA$ is one that starts in vertex $v$, and all the other three vertices belong to the set $A$.

For a graph $G$ and disjoint vertex sets $X,Y\subseteq V(G)$, by $\Reach_G(X,Y)$ we denote the set of all vertices $u\in V(G)\setminus Y$
for which there is a path in $G-Y$
that starts in a vertex of $X$ and ends in $u$. That is, $\Reach_G(X,Y)$ is the union of the vertex sets of those connected components of $G-Y$ that contain a vertex of $X$.
By $\Proj_G(X,Y)$ we denote the set $N_G(\Reach_G(X,Y))$; note that $\Proj_G(X,Y)\subseteq Y$.
If $X=\{x\}$, we use $\Reach_G(x,Y)$ and $\Proj_G(x,Y)$ instead, and we drop the subscript whenever $G$ is clear from the context.

\paragraph*{Chordal graphs and chordal completions.}
A \emph{hole} in a graph $H$ is an induced cycle of length at least $4$.
A graph $H$ is {\em{chordal}} if it contains no holes.
In this work, we will also rely on an equivalent definition of chordal graphs via tree decompositions.
Recall that a tree decomposition of a graph $G$ is a pair $(T,\bag)$ where $T$ is a tree and $\bag\colon V(T)\to 2^{V(G)}$ is a function that associates each {\em{node}} $x$ of $T$ with its {\em{bag}} $\bag(x)$
so that the following conditions are satisfied: (1) for each $u\in V(G)$, the set $\{x\colon u\in \bag(x)\}$ induces a non-empty and connected subtree of $T$, which we shall denote by $T_u$, 
and (2) for each $uv\in E(G)$, there is a node $x$ of $T$
such that $\{u,v\}\subseteq \bag(x)$. For an edge $xy\in E(T)$, the set $\bag(x)\cap \bag(y)$ is called the {\em{adhesion of $xy$}}, and denoted $\adh(xy)$.

The following classic result gives a structural characterization of chordal graphs.
\begin{proposition}[Folklore]
A graph $H$ is chordal if and only if $H$ admits a tree decomposition whose bags are exactly all the maximal cliques in $H$.
\end{proposition}
\noindent Such a tree decomposition is often called a {\em{clique tree}} of $H$. Note that while its set of bags is unique, the tree structure may not be unique.

For a graph $G$ and a set of nonedges $F$, by $G+F$ we denote the graph obtained from $G$ by adding all elements of $F$ as edges.
If $G+F$ is chordal, then $F$ will be called a {\em{chordal completion}} (or {\em{fill-in}}) of $G$.
A chordal completion $F$ is {\em{minimal}} if there is no other chordal completion $F'$ that is a proper subset of $F$.

\paragraph*{Minimal separators and PMCs.} Suppose $G$ is a graph and $S\subseteq V(G)$ is a vertex subset.
Consider a connected component $D\in \cc(G-S)$. We say that $D$ is {\em{full}} to $S$ if every vertex of $S$ has a neighbor in $D$ (in the graph $G$); note that this is different than saying that $V(D)$ is complete
to $S$. 
The set $S$ is called a {\em{minimal separator}} if $\cc(G-S)$ contains at least two connected components that are full to $S$.

From a different perspective, for two nonadjacent vertices $u,v \in V(G)$, we say that $S \subseteq V(G)$ \emph{separates} $u$ and $v$ if $u,v \notin S$ and $u,v$ lie in distinct connected
components of $G-S$. A $u,v$-separator $S$ is \emph{minimal} if it is inclusion-wise minimal.
It is easy to see that a set is a minimal separator if it is a minimal $u,v$-separator for some choice of $u,v \in V(G)$.
Clearly, also for any nonadjacent $u,v \in V(G)$ there exists some minimal $u,v$-separator, e.g., the minimum vertex cut between $u$ and $v$.

Let $S_1$ and $S_2$ be two minimal separators in $G$. We say that $S_1$ \emph{crosses} $S_2$ if $S_2$ intersects at least two connected components of $G-S_1$. 
A standard observation is that the notion of crossing is symmetric:
\begin{lemma}[\cite{ParraS95}]\label{lem:minsep-cross}
Let $S_1$ and $S_2$ be two minimal separators in $G$. Then $S_1$ crosses $S_2$ if and only if $S_2$ crosses $S_1$.
\end{lemma}
%\begin{proof}
%We prove only the forward implication, as the opposite one is symmetric.
%Let $D_u$ and $D_v$ be two connected components of $G-S_1$ with vertices $u \in D_u \cap S_2$, $v \in D_v \cap S_2$. Let $D$ be a component of $G-S_2$ that is full to $S_2$.
%Since $N(D)$ contains $u \in D_u$, we have that $D \cap N[D_u] \neq \emptyset$. Similarly, $D \cap N[D_v] \neq \emptyset$. As $G[D]$ is connected and $D_u,D_v$ are two distinct connected components
%of $G-S_1$, we have that $D \cap S_1 \neq \emptyset$. The lemma follows as the choice of $D$ is arbitrary, and there are at least two components of $G-S_2$ that are full to $S_2$.
%\end{proof}
Lemma~\ref{lem:minsep-cross} allows us to use the phrases that two minimal separators $S_1$ and $S_2$ are \emph{crossing} or \emph{noncrossing}.

The following statement is a well-known characterization of chordal graphs in terms of minimal separators.
\begin{lemma}[\cite{Dirac61}]\label{lem:chordal-minsep}
A graph $G$ is chordal if and only if each of its minimal separators is a clique.
\end{lemma}
%\begin{proof}
%In one direction, if $S$ is a minimal separator with two full components $D_1,D_2$ and two nonadjacent vertices $u,v \in S$, then the shortest paths from $u$ to $v$
%via $D_1$ and $D_2$, respectively, form a hole in $G$. In the other direction, observe that a minimum vertex cut between two nonadjacent vertices of a hole in $G$
%is a minimum separator that contains a vertex from each of the two paths on the hole between the two chosen vertices, and these two vertices are nonadjacent.
%\end{proof}
\begin{corollary}[\cite{ParraS95}]\label{cor:chordal-crossing}
A graph $G$ is chordal if and only if every two of its minimal separators are noncrossing.
\end{corollary}
%\begin{proof}
%In one direction, note that no minimal separator $S_1$ can cross a minimal separator $S_2$ that is a clique, as in this case $G[S_2 \setminus S_1]$ is complete, and thus lies in a single connected
%component of $G-S_1$. In the other direction, if $u,v \in S_1$ are two nonadjacent vertices of a minimal separator $S_1$, then a minimum vertex cut $S_2$ between $u$ and $v$ is a minimal separator that
%crosses $S_1$.
%\end{proof}

For chordal completions, a crucial property is that a minimal chordal completion cannot create a new minimal separator; the following statement is one of the statements of~\cite[Theorem 2.9]{BouchitteT01}.
\begin{lemma}\label{lem:fillin-minsep}
If $G$ is a graph and $F$ is a minimal chordal completion in $G$, then every minimal separator of $G+F$ is a minimal separator of $G$ as well.
\end{lemma}
Thus, by the last two statements, a minimal chordal completion of $G$ corresponds to a choice of a pairwise noncrossing subset of minimal separators of $G$.
This correspondence can be made in both directions (cf.~\cite{BouchitteT01}), but we do not need the exact statements here.

A vertex subset $\Om\subseteq V(G)$ is called a {\em{potential maximal clique}} ({\em{PMC}}, for short) if the following conditions are satisfied:
\begin{description}
\item[\pmca]\label{pmca} none of the connected components of $\cc(G-\Om)$ is full to $\Om$; and
\item[\pmcb]\label{pmcb} whenever $uv$ is a non-edge with $\{u,v\}\subseteq \Om$, then there is a component $D\in \cc(G-\Om)$ such that $\{u,v\}\subseteq N(D)$.
\end{description}
In the second condition, we will say that the component $D$ {\em{covers}} the non-edge $uv$, thus this condition says that every non-edge within the PMC $\Omega$ must be covered. Observe  that it is possible to test whether a given subset of vertices is a PMC in polynomial time.
The following classic results link PMCs with chordal completions and their clique trees.

\begin{proposition}[Theorem 3.15 of~\cite{BouchitteT01}]
For a graph $G$, a vertex subset $\Omega\subseteq V(G)$ is a PMC if and only if there exists a minimal chordal completion $F$ of $G$ such that $\Omega$ is a maximal clique in $G+F$.
\end{proposition}

We will need a more refined understanding of the connection between minimal separators and minimal chordal completions. The following proposition essentially follows from the
toolbox introduced in~\cite{BouchitteT01}; we give a proof for the sake of completeness.

\begin{proposition}\label{prop:adh}
Suppose $F$ is a minimal chordal completion in a graph $G$, and suppose $(T,\beta)$ is a  clique tree of $G+F$. Then every adhesion $\adh(e)$ in $(T,\beta)$ is a minimal separator in $G$.
Moreover, if the removal of $e = t_1t_2$ from $T$ splits it into subtrees $T_1$ and $T_2$ with $t_i \in V(T_i)$, then there exist components $D_1,D_2\in \cc(G-\adh(e))$ that are 
full to $\adh(e)$ in $G$, $\beta(t_i) \setminus \adh(e) \subseteq D_i$, and such that $T_u\subseteq T_i$ for each $u\in D_i$ and $i=1,2$.
\end{proposition}
\begin{proof}
For $i=1,2$, let $\Om_i = \beta(t_i)$ be the maximal clique of $G+F$ associated with the endpoint $t_i$ of $e$ that is in $T_i$.
By~\cite[Proposition 2.2]{BouchitteT01}, $\adh(e)=\Om_1\cap \Om_2$ and $\adh(e)$ is a minimal separator in $G+F$.
Since $\Om_1$ and $\Om_2$ are different maximal cliques, there is a vertex $x\in \Om_1\setminus \Om_2=\Om_1\setminus \adh(e)$.
Let $D_1$ be the connected component of $x$ in $G-\adh(e)$. 

We claim $D_1$ is full to $\adh(e)$ and contains $\Om_1 \setminus \adh(e)$.
To show this, it suffices to show that every $v\in \Om_1$ satisfies $v \in N[D_1]$. If $xv\in E(G)$, then we already have $v\in N[D_1]$, so suppose $xv\notin E(G)$.
Since $\Om_1$ is a PMC, as witnessed by the minimal chordal completion $F$, there is a component $D$ of $G-\Om_1$ for which $\{x,v\}\subseteq N(D)$.
As $x\in N(D)$ and $\adh(e)\subseteq \Om_1$, we have that $D\subseteq D_1$, and hence $v\in N[D_1]$, as required.

We now claim that $T_u\subseteq T_1$ for each $u\in D_1$. This follows immediately from properties of tree decomposition, using the fact that 
$x\in \Om_1$, which is a bag associated with a bag of a node of $T_1$, and the fact that $D_1\cap \adh(e)=\emptyset$ by definition.

Hence, $D_1$ satisfies the required properties. A symmetric reasoning yields a component $D_2\in \cc(G-\adh(e))$ that is full to $\adh(e)$ in $G$ and such that $T_u\subseteq T_2$ for each $u\in D_2$.
In particular $D_1$ and $D_2$ have to be different, so $\adh(e)$ is a minimal separator in $G$.
\end{proof}

\begin{lemma}\label{lem:chordal-minsep-adh}
Let $G$ be a chordal graph and let $(T,\beta)$ be its clique tree.
A set $S \subseteq V(G)$ is a minimal separator in $G$ if and only if there exists an edge $e$ of $T$ with $\adh(e) = S$.
Furthermore, for every component $D$ of $G-S$ that is full to $S$ one can choose such an edge $e=st$ with $\beta(s) \subseteq N[D]$.
\end{lemma}
\begin{proof}
The backwards direction is asserted by Proposition~\ref{prop:adh}, as the empty set is the unique minimal chordal completion of a chordal graph.

In the forward direction, let $S$ be a minimal separator in $G$ and let $D_1,D_2$ be two components of $G-S$ that are full to $S$.
For $i=1,2$, let $T_i$ be the subgraph of $T$ induced by nodes $t$ with $\beta(t) \cap D_i \neq \emptyset$; since $D_i$ is connected and nonempty, $T_i$ is a nonempty subtree of $T$.
Since there are no edges between $D_1$ and $D_2$, $T_1$ and $T_2$ are node-disjoint.

Let $P_T$ be the minimal path in $T$ with one endpoint in $T_1$ and the second endpoint in $T_2$.
Note that, as $T_1$ and $T_2$ are disjoint, $P_T$ has at least one edge.
Due to its minimality, all internal vertices of $P_T$ lie neither in $T_1$ nor in $T_2$.
Let $e = st \in E(T)$ be the edge of $P_T$ incident with the endpoint in $T_1$
and let $s \in V(T_1)$ be this endpoint. 
We claim that $\adh(e) = S$.

First, note that $\beta(s) \subseteq N[D_1]$ (as in the second statement of the lemma),
as $\beta(s) \cap D_1 \neq \emptyset$ and $(T,\beta)$ is a clique tree. Second, by the definition of $T_1$, we have $\beta(t) \cap D_1 = \emptyset$.
Consequently, $\adh(e) \subseteq N(D_1) = S$.

In the other direction, let $u \in S$ and let $v_i \in N(u) \cap D_i$ for $i=1,2$. By the properties of a tree decomposition,
there exist nodes $s_1,s_2$ with $u,v_i \in \beta(s_i)$. Clearly, $s_i \in V(T_i)$.
Hence, the edge $e$ lies on the unique path in $T$ from $s_1$ to $s_2$.
Consequently, by the properties of a tree decomposition, we have $u \in \adh(e)$, as desired.
\end{proof}

Let us now focus on the relation between minimal separators and maximal cliques in chordal graphs. By Lemma~\ref{lem:chordal-minsep}, every minimal separator in a chordal graph
is a clique; however, Lemma~\ref{lem:chordal-minsep-adh} implies that it is never a maximal one.
\begin{lemma}\label{lem:chordal-minsep-clique}
Let $S$ be a minimal separator in a chordal graph $G$, and let $D$ be a component of $G-S$ that is full to $S$.
Then, there exists a maximal clique $\Om$ in $G$ with $S \subsetneq \Om \subseteq N[D]$.
Furthermore, if there exists more than one such maximal clique $\Om$, then there exists a minimal separator $S'$ in $G$ with $S \subsetneq S' \subseteq N[D]$.
\end{lemma}
\begin{proof}
The first statement follows directly from Lemma~\ref{lem:chordal-minsep-adh} by picking $\Om = \beta(s)$ for the promised edge $e=st$ with $\beta(s) \subseteq N[D]$.

For the second claim, assume there are two such maximal cliques $\Om_1$ and $\Om_2$, and let $v_i \in \Om_i \setminus \Om_{3-i}$; note that $v_1,v_2 \in D$.
Furthermore, since both $\Om_1$ and $\Om_2$ are maximal, we can choose $v_1$ and $v_2$ such that they are nonadjacent.
Then, we show that any minimal $v_1,v_2$-separator $S'$ satisfies the claim.
Clearly, such an $S'$ contains $N(v_1) \cap N(v_2) \supseteq \Om_1 \cap \Om_2 \supseteq S$. Consequently, $S' \cap N[D]$ also separates $v_1$ from $v_2$; from
the minimality of $S'$ we infer that $S' \subseteq N[D]$. However, since $v_1$ and $v_2$ are vertices of the same connected component of $G-S$, the set $S$
itself does not separate $v_1$ from $v_2$, and thus $S \subsetneq S'$.
\end{proof}

In the other direction, we have the following; see Theorem 3.14 of~\cite{BouchitteT01} for a statement with less details.
\begin{lemma}\label{lem:pmc-minseps}
Suppose $\Om$ is a PMC in a graph $G$, and let $D\in \cc(G-\Om)$. Denote $S=N(D)$. Then there is a component $D_\Om\in \cc(G-S)$ whose vertex set consists of $\Om\setminus S$ and the union 
of vertex sets of all components $D'\in \cc(G-\Om)$ for which $N(D')\not\subseteq N(D)$. Moreover, $S$ is a minimal separator in $G$, where $D$ and $D_\Omega$ are two connected components of $G-S$ that 
are full to $S$.
\end{lemma}
\begin{proof}
By the definition of $S$, $D$ is a connected component of $G-S$ and is full to $S$.
Let $D_{\Om}$ be the union of $\Om \setminus S$ and all connected components $D' \in \cc(G-\Om)$ with $N(D') \not\subseteq N(D)$.
We show that $D_{\Om}$ is indeed a connected component of $G-S$ that is full to $S$.

Observe that $N(D_{\Om}) \subseteq S$ by definition. Let us show that $D_{\Om}$ is connected in $G-S$. By \pmca, $\Om- S$ is nonempty and by definition of $D_{\Om}$, every $D'\in \cc(G-\Om)$ such that $D'\subseteq D_{\Om}$ has a neighbor in $\Om- S$. 
From \pmcb, it follows that $\Om- S$ lies in one connected component of $G- S$; every two vertices of $\Om- S$ are either connected by an edge or belong to $N(D')$ for some $D'\in \cc(G- \Om)$ (which therefore lies in $D_{\Om}$).
Thus, $D_{\Om}$ is indeed a connected component in $G-S$.

To show that $D_{\Om}$ is full to $S$, let $x$ be a vertex of $\Om- S$. By \pmcb, for every vertex $y\in S$, there is either $xy$ is an edge or $\{x,y\}\subseteq N(D')$ for $D'\in \cc(G- \Om)$. Such $D'$ is by definition a subset of $D_{\Om}$. It follows that $D_{\Om}$ is full to $S$ and thus, $S$ is a minimal separator.
\end{proof}

\paragraph*{Modules and modular decomposition.} In our argumentation we will use the basic properties of modules and the modular decomposition of a graph, introduced by Gallai~\cite{gallai67}.
We refer the reader to the survey of Habib and Paul~\cite{HabibP10} for a comprehensive review of modern approaches and algorithmic applications.

Suppose $G$ is a graph. A nonempty vertex subset $M\subseteq V(G)$ is a {\em{module}}
if for every vertex $u\notin M$, either $u$ is complete to $M$ or $u$ is anti-complete to $M$.
Note that if $M$ and $M'$ are two disjoint modules, then $M$ and $M'$ are either complete or anti-complete to each other; for brevity, we will just say that $M$ and $M'$ are adjacent or non-adjacent.

A module $M$ is {\em{proper}} if it is not equal to the whole vertex set, and $M$ is {\em{strong}} if for any other module $M'$, either $M\subseteq M'$, or $M\supseteq M'$, or $M\cap M'=\emptyset$.
A proper strong module is {\em{maximal}} if there is no other proper strong module that is its proper superset.
The following classic result explains the structure of maximal proper strong modules in a graph.

\begin{proposition}[cf. Lemma 2 in~\cite{HabibP10}]\label{prop:modules-partition}
For any graph $G$ with more than one vertex, maximal proper strong modules of $G$ form a partition of the vertex set of $G$.
\end{proposition}

This proposition naturally leads to the following definitions.
The {\em{modular partition}} of a graph $G$, denoted $\Mod(G)$ is simply the set of maximal proper strong modules in $G$; 
Proposition~\ref{prop:modules-partition} ensures that $\Mod(G)$ is a partition of $V(G)$ provided $G$ has more than one vertex.
The {\em{quotient graph}} of $G$, denoted $\Quo(G)$, has $\Mod(G)$ as the vertex set,
and two modules are connected by an edge in $\Quo(G)$ if and only if they are adjacent in $G$. A graph $G$ is \emph{prime} if its only modules are trivial, that is, singletons.
We have the following corollary of Proposition~\ref{prop:modules-partition}.

\begin{proposition}[cf. Theorem 2 in~\cite{HabibP10}]\label{prop:modules-quo}
For every graph $G$ with more than one vertex, its quotient graph $\Quo(G)$ is
an independent set if $G$ is not connected,
a clique if the complement of $G$ is not connected,
or a prime graph otherwise.
\end{proposition}

The next proposition explains how other modules behave with respect to modular partition.

\begin{proposition}[cf. Lemma 2 and Theorem 2 in~\cite{HabibP10}]\label{prop:modules-other}
Let $G$ be a graph. Then every proper module of $G$ is either contained in one of the maximal proper strong modules in $G$, or is the union of a collection of maximal proper strong modules in $G$.
Moreover, the latter case can happen only if the quotient graph $\Quo(G)$ is a clique or an independent set.
\end{proposition}

Propositions~\ref{prop:modules-partition} and~\ref{prop:modules-other} can be used recursively to form a hierarchical decomposition of the graph into smaller and smaller modules, as defined next.
\begin{definition}
For a graph $G$, the {\em{modular decomposition}} of $G$ is a rooted tree $T$ with every node $x$ labelled with a module $M_x$ such that the following conditions hold:
\begin{itemize}
\item The root of $T$ is labelled with the module $V(G)$.
\item Each leaf of $T$ is labelled by a module consisting of one vertex.
\item For each node $x$ of $T$, the set of labels of the children of $x$ is the modular partition of $G[M_x]$.
\end{itemize}
For a node $x$, if $\Quo(G[M_x])$ is a clique then $x$ is called a {\em{clique node}}, if $\Quo(G[M_x])$ is edgeless then $x$ is called an {\em{independent set node}}, and otherwise $x$ is called a {\em{prime node}}.
\end{definition}
As discussed, recursive application of Propositions~\ref{prop:modules-partition} and~\ref{prop:modules-other} yields the following.

\begin{proposition}[cf. discussion after Theorem 2 in~\cite{HabibP10}]\label{prop:modules-decomp}
For every graph $G$ there exists a unique modular decomposition $T$ of $G$ that can be computed in polynomial time.
Moreover, every module of $G$ is either the label of some node in $T$, or the union of labels of some collection of children of some clique or independent set node. 
\end{proposition}

\begin{corollary}\label{cor:enum-modules}
For every graph $G$ on $n$ vertices, there are at most $2n-1$ modules in $G$ that induce a graph that is both connected and its complement is connected.
Moreover, these modules can be enumerated in polynomial time.
\end{corollary}
\begin{proof}
Let $T$ be the modular decomposition of $G$.
From Proposition~\ref{prop:modules-decomp} it is easy to see that if a module $M$ induces a graph that is both connected and its complement is connected, then $M$ has to be the label of some node in $T$. 
Leaves of $T$ are in one-to-one correspondence with vertices of $G$, hence there is exactly $n$ of them. Since every internal node of $T$ has at least two children, it follows that $T$ has at most $2n-1$ nodes.
For enumeration, we can compute $T$ in polynomial time and output those labels of its nodes that satisfy the condition.
\end{proof}

In this paper we will often deal with graphs whose quotient graphs are cliques. This justifies introducing the following definition.

\begin{definition}
A graph $G$ is a {\em{mesh}} if $|V(G)|\geq 2$ and the quotient graph $\Quo(G)$ is a clique.
\end{definition}

\section{Overview}\label{sec:over}

\subsection{Replacement strategy}

Let $G$ be a $P_6$-free graph and let $I$ be a maximal independent set
in $G$. Recall that we are looking for a polynomially-sized
family $\Ff$ of PMCs in $G$
such that there exists an $I$-free minimal chordal completion $F$ of $G$
(i.e., one that does not add any edges incident with $I$)
with all maximal cliques of $G+F$ belonging to $\Ff$.
Such a family $\Ff$ can be then plugged into the dynamic programming
algorithm of Lemma~\ref{lem:dp}.

A useful point of view is the following. Fix an $I$-free minimal chordal
completion $F$ of $G$ and fix a maximal clique $\Om$ of $G+F$.
We would like to deduce $\Om$ by making a constant number of guesses,
each guess having one of polynomially many options. Typical guesses are:
one of the few structural options, a specific vertex or edge of $G$, 
or a strong module of $G$ or some induced subgraph of $G$.

As discussed earlier, the number of PMCs is exponential even in $P_5$-free
graphs, and thus we cannot hope to discover every such PMC $\Om$.
Thus, at some places we will need to fall back to some {\emph{greedy}} arguments, 
  where we will assume that $F$ behaves in some parts of the graph
  in a canonical way. 
  Correctness of such assumption will be ensured by a \emph{replacement}
  argument that is best illustrated on an example of a maximal clique
  $\Om$ intersecting $I$. The following argumentation is already present
  in~\cite{LokshtanovVV14}.

Assume $\Om \cap I \neq \emptyset$. Since $\Om$ is a clique in $G+F$ and
$F$ is $I$-free, we have that $|\Om \cap I| = 1$. 
Let $\Om \cap I = \{v\}$, $\mathcal{D} = \cc(G-N[v])$ and $D \in \mathcal{D}$.
Observe that $N(D)$
is a minimal separator in $G$ with $D$ as one full side and the connected component
of $G-N(D)$ containing $v$ as a second full side.
The critical insight is that $F$ needs to complete $N(D)$ into a clique
and $N(D)$ remains a minimal separator in $G+F$.
To see this, let $x,y \in N(D) \subseteq N(v)$ with $xy \notin E(G)$, and
consider a hole $H$ consisting of $x$, $v$, $y$, and a shortest path from $x$
to $y$ with internal vertices in $D$. Observe that
the only way to triangulate $H$ without
adding an edge incident with $v$ is to add the edge $xy$.

Thus, the transition (completion) from $G$ to $G+F$ can be seen as a two-stage
process. First, for every $D \in \mathcal{D}$
we complete $N(D)$ into a clique. Second, we complete $G[N[v]]$ into a chordal
graph (without adding any edge incident with $v$) 
  and, \emph{independently}, we complete every $G[N[D]]$ into a chordal
graph for every $D \in \mathcal{D}$ (without adding any edge incident with
    $I \setminus \{v\}$). 
Note that as $v \in I$, it does not really matter how in the second stage
we complete $G[N[v]]$ into a chordal graph.
That is, we can fix one chordal completion $F_v$ of $G[N[v]]$ that (i) does not add any edge incident with $v$, (ii) completes all $N(D)$ for $D \in \mathcal{D}$ into a clique, (iii) is minimal, respecting the previous requirements. 
The argumentation so far implies that if $F$ is an $I$-free minimal chordal completion with $v \in I$, then
$F' := (F \setminus \binom{N[v]}{2}) \cup F_v$, where $\binom{N[v]}{2}$ is the set of all pairs of vertices from $N[v]$, is an $I$-free minimal chordal completion as well. 
Hence, we can without loss of generality consider only completions $F$ such that $F_v \subseteq F$ for every $v \in I$.

This motivates the following step. For every $v \in V(G)$, we choose one $F_v$ as above, and add all maximal cliques of $G[N[v]] + F_v$ to $\Ff$. This results in at most $n^2$ elements of $\Ff$, while we obtain the property that every maximal clique $\Om$ of $G+F$ with $\Om \cap I \neq \emptyset$ is contained in $\Ff$.
In particular, in the rest of the argumentation, we care only about PMCs $\Om$ that are disjoint with $I$.

In our proof, we develop a number of other, more involved
\emph{replacement arguments}. However, they all share the same general structure:
we identify some part of the graph $\Gamma \subseteq V(G)$,
   argue that we can assume that $F$ turns every $N(D)$ for $D \in \cc(G-\Gamma)$
   into a clique, and then insert into $\Ff$ all cliques of 
   one arbitrarily fixed completion of $G[\Gamma]$
   that turns all $\{N(D) | D \in \cc(G-\Gamma)\}$ into cliques.
   
\subsection{PMC guessing}\label{o:ss:guess}

To successfully employ replacement arguments, our strategy is to discover (guess) 
a dense enough subset of the maximal cliques of $G+F$ such that the areas
between the guessed cliques are simple enough to admit replacement arguments.
Let us now take a look at our PMC discovery arguments.

The first important observation is that
instead of trying to guess a PMC $\Om$, we may try to find a small family of 
candidates for the components of $\cc(G-\Om)$.
This is expressed in the following lemma,
     which generalizes
an analogous result for $P_5$-free graphs from~\cite{LokshtanovVV14}.

\begin{lemma}\label{o:lem:recover1}
Suppose $G$ is a $P_6$-free graph on $n$ vertices.
Given a family $\Xx\subseteq 2^{V(G)}$, one can in time polynomial in $n$ and $|\Xx|$ compute a family $\Ff_{\rec,1}(\Xx)\subseteq 2^{V(G)}$ 
such that $|\Ff_{\rec,1}(\Xx)|\leq  3n^6\cdot |\Xx|^3$ and the following property
holds: for every PMC $\Om$ in $G$, if $\cc(G-\Om)\subseteq \Xx$ then $\Om\in \Ff_{\rec,1}(\Xx)$.
\end{lemma}

The proof of Lemma~\ref{o:lem:recover1}, presented in Section~\ref{sec:toolbox},
builds upon the corresponding proof of~\cite{LokshtanovVV14}.

Recall that if $\Om$ is a PMC, then for every $D \in \cc(G-\Om)$, $N(D)$ is a minimal separator in $G$.
Therefore, when guessing a PMC $\Om$ we will often try to guess such minimal separators $N(D)$.
For this, the next technical tool is very useful.

%Lemma~\ref{o:lem:recover1} essentially says that we can focus on guessing 
%components of $\cc(G-\Om)$ instead of the PMC $\Om$ itself.

\begin{lemma}[Separator Covering Lemma, simplified version]\label{o:lem:covering-simple}
Suppose $G$ is a $P_6$-free graph, and suppose $S\subseteq V(G)$ is a minimal separator in $G$ with $D_1,D_2\in \cc(G-S)$ being two components full to $S$.
Then there exist nonempty sets $A_1\subseteq D_1$ and $A_2\subseteq D_2$ with $|A_1|\leq 3$ and $|A_2|\leq 3$ such that $N[A_1\cup A_2]\supseteq S$.
\end{lemma}
The statement and the proof of Lemma~\ref{o:lem:covering-simple} (presented in Section~\ref{sec:toolbox}) is inspired by~\cite[Theorem~1.3]{LokshtanovPL16}, but requires much deeper analysis of the structure of $G$.

An immediate consequence of Lemma~\ref{o:lem:covering-simple} is our ability
to guess components $D$ with neighborhoods covered by other minimal separators.
The proofs of the next few lemmata can be found in Section~\ref{sec:capturing}.
\begin{lemma}\label{o:lem:hidden}
Given an $n$-vertex $P_6$-free graph $G$, one can in polynomial time
output a family $\Xx \subseteq 2^{V(G)}$ of size at most $n^7$
such that for every potential maximal clique $\Om$ in $G$ and for every
$D \in \cc(G-\Om)$ such that there exists another $D_1 \in \cc(G-\Om)$
with $N(D) \subseteq N(D_1)$, we have $D \in \Xx$.
\end{lemma}

We now illustrate our PMC discovery techniques with two important examples.
We start with noting a simple observation.
\begin{lemma}\label{o:lem:wings}
Let $G$ be a $P_6$-free graph and let $\Omega$ be a PMC in $G$. Suppose $D_1,D_2\in \cc(G-\Omega)$ are two components.
Then either $N(D_1)\setminus N(D_2)$ is complete to $D_1$, or $N(D_2)\setminus N(D_1)$ is complete to $D_2$.
Furthermore, for all $v_1\in N(D_1)\setminus N(D_2)$ and $v_2\in N(D_2)\setminus N(D_1)$ with $v_1v_2\notin E(G)$, 
both $v_1$ is complete to $D_1$ and $v_2$ is complete to $D_2$.
\end{lemma}

Our first example case is inspired by a similar reasoning in $P_5$-free graphs~\cite{LokshtanovVV14}.
\begin{lemma}\label{o:lem:two-not-whole}
Given an $n$-vertex $P_6$-free graph $G$, one can in polynomial time compute a family $\Ff_1$ of size at most $n^8$ such that the following holds: 
for any PMC $\Om$ in $G$ and any components $D_1,D_2\in \cc(G-\Om)$ with $N(D_1)\cup N(D_2)\subsetneq \Om$,
at least one of $D_1$ or $D_2$ belongs to $\Ff_1$.
\end{lemma}
In $P_5$-free graphs, the case covered by Lemma~\ref{o:lem:two-not-whole} and the case of a PMC $\Om$ contained in $N[u,v]$ for some $u,v\in I$
turn out to exhaustively cover all the cases. In $P_6$-free graphs, the range of possible PMCs turns out to be far richer.

Our second example lemma shows that out of any three components, we can recognize at least one.
\begin{lemma}\label{o:lem:one-in-three-see}
Given an $n$-vertex $P_6$-free graph $G$, one can in polynomial time compute a family $\Ff_2$ of size at most $2n^9$ such that the following holds: for any PMC $\Om$ in $G$ and any pairwise different components $D_1,D_2,D_3\in \cc(G-\Om)$,
at least one of $D_1$, $D_2$, and $D_3$ belongs to $\Ff_2$.
\end{lemma}
Thus, we can compute a family that for any PMC $\Om$, contains all but at most two components of $G-\Om$.
Even though this seems very powerful, the issue of recognizing the remaining two components remains. 

The arguments of Section~\ref{sec:capturing} culminate with the following
statement, which essentially says that we can ``almost'' guess all PMCs that give rise to at least three neighbor-maximal components. Here, a component $D \in \cc(G-\Om)$ is \emph{neighbor-maximal}
if there is no other component $D'\in \cc(G-\Om)$ with $N(D) \subseteq N(D')$.
A \emph{fuzzy version} of a mesh component $D \in \cc(G-\Om)$ is
a vertex set $D^+$ containing $D$ and possibly some vertices of $\Om$ that are full to $D$.
\begin{lemma}\label{o:lem:summary}
Given an $n$-vertex $P_6$-free graph $G$, one can in polynomial time compute families $\Ff^1_9$ and $\Ff^2_9$, each of size at most $10^{13}\cdot n^{159}$%TODO-maybe $10^{13}\cdot n^{156}$
, such that the following holds: for any $I$-free PMC $\Om$ in $G$ with at least three neighbor-maximal components,
either $\Ff^1_9$ contains $\Om$, or $\Ff^2_9$ contains a triple $(\Om\cup D_1\cup D_2,D_1^+,D_2^+)$ for some components $D_1,D_2\in \cc(G-\Om)$ that are meshes,
where $D_i^+$ is a fuzzy version of $D_i$, for $i=1,2$.
\end{lemma}

\subsection{Mesh components}\label{o:ss:mesh}

In the arguments discussed in the previous subsection, 
the main trouble is caused by components $D \in \cc(G-\Om)$ that are
meshes. 
In particular, consider the statement of Lemma~\ref{o:lem:summary}: we are not able to guess all PMCs with at least
three neighbor-maximal components, but for some of them --- when two of these components are meshes ---
we are able to discover only an ``approximate version'' of the form $(\Om \cup D_1 \cup D_2, D_1^+, D_2^+)$.
We now discuss our methodology of dealing with such troublesome mesh components.

Let $\Om$ be a maximal clique of $G+F$ with $\Om \cap I = \emptyset$
and let $D \in \cc(G-\Om)$ be such that $D$ is a mesh, that is, the
quotient graph of $G[D]$ is a clique or, equivalently, the complement
of $G[D]$ is disconnected. Let $S = N(D)$; recall that $S$
is a minimal separator.
Observe that from the maximality of $I$ we have $D \cap I \neq \emptyset$,
while from the fact that $D$ is a mesh it follows that there exists
a unique proper strong maximal module $M_{F,D}$ of $G[D]$ that
contains all vertices of $I \cap D$.

In Section~\ref{ss:modif-mesh},
we study the structure of minimal separators in $G+F$ ``between'' $S$ and
$M_{F,D}$; see Figure~\ref{fig:modif-mesh1} on Page~\pageref{fig:modif-mesh1}.
We show that these separators in $G+F$ are linearly ordered, in particular, there is a well-defined separator $S_{F,D}$ ``closest'' to $M_{F,D}$
and some vertex $q_{F,D} \in D \setminus M_{F,D}$ that lies in the same
connected component of $G-S_{F,D}$ as elements of $I \cap M_{F,D}$.
Furthermore, we show that there exists a good notion of
a canonical separator between
$S_{F,D}$ and $S$ that can be greedily used in the considered
$I$-free completion~$F$. This canonical separator can be guessed 
if we know $M_{F,D}$, which in turn can be guessed from a fuzzy version
of $D$. (Note that the quotient
graph of $G[D^+]$, where $D^+$ is a fuzzy version of a mesh~$D$, is a clique too and proper strong maximal modules of $G[D]$
are proper strong maximal modules of $G[D^+]$ too.)

Unfortunately, sometimes we do not even have an access to a small number
of candidates for a fuzzy version of $D$.
This may happen when there exists a minimal separator of $G+F$ with two full
sides being meshes. We are able to handle this case with the following lemma,
      proven in Section~\ref{ss:twomesh}.

\begin{lemma}\label{o:lem:rozrywanie}
One can in polynomial time compute a family $\Ff^Z$ of size at most $n^6$
such that the following holds.
Let $S$ be a minimal separator in $G$ and let $D_1$ and $D_2$ be two components
of $G-S$ that are full to $S$ and are meshes.
For $i=1,2$, let $M_i^p$ be an arbitrary proper strong module of $D_i$ 
and let $q_i \in D_i \setminus M_i^p$ be arbitrary.
Then there exists an element $(Z_1,Z_2) \in \Ff^Z$ such that
$Z_1,Z_2 \subseteq V(G)$ induce connected subgraphs of $G$,
there are no edges between $Z_1$ and $Z_2$ and
 for $i=1,2$, we have
$$M_i^p \cup \{q_i\} \subseteq Z_i \subseteq (D_i \cup S) \setminus N(V(G) \setminus (S \cup D_1 \cup D_2)).$$
\end{lemma}
The usage of Lemma~\ref{o:lem:rozrywanie} is as follows. Assume $S$, $D_1$, $D_2$ are as in the statement of the lemma with $S$ being a minimal separator of $G+F$ with $S \cap I = \emptyset$. Similarly as before,
    let $M_i^p$ be the unique proper strong module of $G[D_i]$
with $M_i^p \cap I \neq \emptyset$ and let $q_i$ be a vertex in $D_i \setminus M_i^p$ that lies in the same connected component of $G-S_{F,D_i}$ as $I \cap M_i^p$.
For these objects, let $(Z_1,Z_2) \in \Ff^Z$ be a pair
as in Lemma~\ref{o:lem:rozrywanie}.
Then, we can use a replacement argument
to argue that without loss of generality in $G+F$, every vertex set
of the form $N(D)$ for $D \in \cc(G-N[Z_i])$, $i=1,2$, is a minimal separator
(and thus a clique). The knowledge and size bound on $\Ff^Z$ gives a polynomial bound
on the number of such separators.

\subsection{Segments admitting replacement arguments}

Arguments sketched in Section~\ref{o:ss:guess}
allow us to arrive at the following situation in Section~\ref{ss:chopping}.
Let $(T,\beta)$ be a clique tree of $G+F$. Color red all nodes $t \in V(T)$
and all edges $e \in E(T)$ for which we can guess the PMC $\beta(t)$
or the minimal separator $\adh(e)$. Let $T'$ consist of vertices
and edges of $T$ that are not red. The aforementioned results, 
    in particular Lemma~\ref{o:lem:summary}, allow us to deduce
    that every connected component of $T'$ is a path.

    Fix such a path $P$.
    Let $t$ be an internal node on $P$
    with incident edges $e_1=tt_1$ and $e_2=tt_2$.
   Let $\Om = \beta(t)$ and let $D_1,D_2 \in \cc(G-\Om)$ be such that
   $\beta(t_i) \setminus \beta(t) \subseteq D_i$ for every $i=1,2$.
   From the results highlighted in Section~\ref{o:ss:guess}, we deduce that:
   \begin{itemize}
   \item $N(D_1) = \adh(e_1)$ and $N(D_2) = \adh(e_2)$ 
    are the only two inclusion-wise maximal separators of the form
    $N(D)$ for some $D \in \cc(G-\beta(t))$,
    \item $N(D_1) \cup N(D_2) = \Om$, and
    \item at least one of $D_1$ and $D_2$ is a mesh.
   \end{itemize}

The last property --- that $D_1$ or $D_2$ is a mesh --- allows us to argue
that the union of bags along the path $P$ is relatively simple and bounded by canonical separators
discussed in Section~\ref{o:ss:mesh}. This in turn allows us to enumerate
a polynomial number of candidates for $\Gamma := \bigcup_{t \in V(P)} \beta(t)$
and, for each such $\Gamma$, insert all maximal cliques of a single
completion of $G[\Gamma]$ that turns every $N(D)$ for $D \in \cc(G-\Gamma)$ into
a clique. This finishes the proof of Theorem~\ref{thm:main}.

\section{Toolbox}\label{sec:toolbox}
A pair $(X,\leq)$ is a {\em{quasi-order}} if $\leq$ is a reflexive and transitive relation on $X$.
If moreover $X$ is antisymmetric, then $(X,\leq)$ is a {\em{partial order}}.
The following statement is a basic combinatorial tool that we will use in our proofs.

\begin{lemma}[Bi-ranking Lemma]\label{lem:biranking}
Suppose $X$ is a non-empty finite set and $(X,\leq_1)$ and $(X,\leq_2)$ are two quasi-orders.
Suppose further that every pair of two different elements of $X$ is comparable either with respect to $\leq_1$ or with respect to $\leq_2$.
Then there exists an element $x\in X$ such that for every $y\in X$ we have either $x\leq_1 y$ or $x\leq_2 y$.
\end{lemma}
\begin{proof}
We first observe that it suffices to prove the statement for the case when $(X,\leq_1)$ and $(X,\leq_2)$ are partial orders.
Indeed, if $(X,\leq_1)$ is not a partial order, then let us modify $\leq_1$ by taking every maximal set $A$ of elements pairwise equivalent in $\leq_1$, 
and changing $\leq_1$ within $A$ so that $\leq_1$ becomes a linear order on $A$.
This does not change the set of pairs of elements comparable in $\leq_1$, while $\leq_1$ becomes a partial order on $X$ that is a subrelation of the original $\leq_1$.
The same transformation can be then applied to $(X,\leq_2)$ as well.
It follows that $\leq_1$ and $\leq_2$ still satisfy the conditions of the lemma, because the sets of pairs of comparable elements did not change, 
while any element $x$ satisfying the assertion claimed in the lemma statement after the modification, satisfies it also for the original quasi-orders $\leq_1$ and $\leq_2$.

We proceed with the proof assuming that $(X,\leq_1)$ and $(X,\leq_2)$ are partial orders. Let us construct an auxiliary directed graph $D$ with vertex set $X$ and arcs defined as follows.
For every two different elements $a,b\in X$, put
\begin{enumerate}[(a)]
\item\label{case:ab} an arc $(a,b)$ if $a\not\leq_1 b$ and $a\not\leq_2 b$;
\item\label{case:ba} an arc $(b,a)$ if $b\not\leq_1 a$ and $b\not\leq_2 a$;
\item\label{case:no} no arc otherwise.
\end{enumerate}
Note that since $a$ and $b$ are comparable in either $\leq_1$ or $\leq_2$, cases~\eqref{case:ab} and~\eqref{case:ba} are exclusive.
Moreover, case~\eqref{case:no} occurs only if $a\leq_1 b$ and $b\leq_2 a$, or $b\leq_1 a$ and $a\leq_2 b$; that is, $\leq_1$ and $\leq_2$ point in different directions on the pair $(a,b)$.

We claim that the digraph $D$ is acyclic. Let us first verify that this claim implies the statement of the lemma.
Supposing $D$ is acyclic, so let $x$ be any sink in $D$, that is, a vertex with no outgoing arcs.
Take any other vertex $y\in X$; then either there is an arc $(y,x)$, or there is no arc between $x$ and $y$ at all.
In the former case we have $y\not\leq_1 x$ and $y\not\leq_2 x$, so since $x$ and $y$ are comparable in at least one of the orders, we infer that indeed $x\leq_1 y$ or $x\leq_2 y$.
In the latter case we have that either $x\leq_1 y$ and $y\leq_2 x$, or $y\leq_1 x$ and $x\leq_2 y$, so again $x$ is smaller than $y$ in at least one of the orders.

It remains to prove the claim. For the sake of contradiction, suppose $D$ is not acyclic. 
Let $C=(a_1,a_2,\ldots,a_k)$ be the shortest cycle in $D$, where $(a_i,a_{i+1})$ is an arc for $i=1,\ldots,k$ (indices behave cyclically modulo $k$).
Every pair of vertices in $D$ is bound by at most one arc by construction, so we have $k\geq 3$.
Since $(a_i,a_{i+1})$ is an arc, we have either $a_i\geq_1 a_{i+1}$ or $a_i\geq_2 a_{i+1}$.
Color the arc $(a_i,a_{i+1})$ \textcolor{red}{red} if $a_i\geq_1 a_{i+1}$, and \textcolor{blue}{blue} if $a_i\geq_2 a_{i+1}$; in case both these assertions hold, choose any of these colors arbitrarily.
Since both $(X,\leq_1)$ and $(X,\leq_2)$ are partial orders, it cannot happen that $C$ is entirely blue or entirely red.
Then $C$ contains both a red and a blue arc. Without loss of generality, by swapping colors and shifting the cycle if necessary, we assume that $(a_1,a_2)$ is red and $(a_2,a_3)$ is blue.
In particular, $a_1\geq_1 a_2$ and $a_2\geq_2 a_3$.

Observe now that it cannot happen that $a_1\leq_1 a_3$, because then by the transitivity we would have $a_2\leq_1 a_3$, hence the arc $(a_2,a_3)$ would not be added to $D$ in the construction.
Symmetrically, it cannot happen that $a_1\leq_2 a_3$, because then we would have $a_1\leq_2 a_2$, which contradicts the existence of the arc $(a_1,a_2)$.
This means that $a_1\not\leq_1 a_3$ and $a_1\not\leq_2 a_3$, so the arc $(a_1,a_3)$ was added to $D$ in the construction.
However, then the cycle $C$ could be shortcut by omitting vertex $a_2$ and using the arc $(a_1,a_3)$ instead, a contradiction with the
minimality of $C$. This concludes the proof.
\end{proof}

We remark that the statement of the Bi-ranking Lemma (Lemma~\ref{lem:biranking}) for partial orders was used as Problem 5 in the second round of 68th Polish Mathematical Olympiad.
We refer to official solutions~\cite{om} for an alternative inductive proof. 
A problem similar to the above lemma, but for more quasi-orders, was stated by Sands, Sauer, and Woodrow in~\cite{SSW} and recently solved by Bousquet, Lochet, and Thomass\'e in~\cite{BLT}.

\paragraph*{Structure of minimal separators.} One of our key tools will be the following lemma describing the structure of interaction between a minimal separator and a component full to it.
Intuitively, it says that one can always pick two vertices $p,q$ in the component so that every vertex $u$ of the separator that is not adjacent to $p$ or $q$, behaves nicely in one of two possible ways:
either $u$ is an end-vertex of some $P_4$ sticking to the component, or the component is a mesh and the neighborhood of $u$ in the component is the union of a collection of proper strong modules of the component.
This analysis was already implicitly present in~\cite{LokshtanovPL16}, see the proof of Lemma 3.1 therein.

\begin{lemma}[Neighborhood Decomposition Lemma]\label{lem:nei-decomp}
Suppose $G$ is a graph and $D\subseteq G$ is a connected induced subgraph of $G$ with $|D|\geq 2$.
Suppose further that vertices $p,q\in D$ respectively belong to different elements $M^p,M^q$ of the modular partition $\Mod(D)$ such that $M^p$ and $M^q$ are adjacent in the 
quotient graph $\Quo(D)$.
Then, for each vertex $u\in N(D)$ at least one of the following conditions holds:
\begin{enumerate}[(a)]
\item\label{cnd:P4} $u$ is such that there exists a $P_4$ of the form $uDDD$ in $G$;
\item\label{cnd:pq} $u\in N[p,q]$;
\item\label{cnd:(} $D$ is a mesh and the neighborhood of $u$ in $D$ is the union of some collection of maximal proper strong modules in $D$.
\end{enumerate}
If $D$ is not a mesh, then the last condition cannot hold.
\end{lemma}
\begin{proof}
Let $u$ be a vertex not satisfying (\ref{cnd:P4}) and (\ref{cnd:pq}). If there is no such vertex, the lemma holds. Let $D'=D\setminus N(u)$. Since $u$ does not satisfy (\ref{cnd:pq}), $p,q\in V(D')$. Let $C$ be a connected component of $D'$ and $v$ a vertex in $D\cap N(u)$. We argue that either $C\subseteq N(v)$ or $C\cap N(v)=\emptyset$. Assume this is not the case. Then there exists an edge between a vertex $x$ in $C\cap N(v)$ and a vertex $y$ in $C\setminus N(v)$, because $C$ is connected. Then, $uvxy$ is an induced $P_4$ contradicting our assumption that $u$ does not satisfy (\ref{cnd:P4}).

Observe that 
\[\mathcal{N}=\{V(C)\,\colon\, C \mbox{ is a connected component of }D'\}\cup \{\{v\}|v\in N(u)\cap D\}\] is a modular partition of $D$ such that $p$ and $q$ belong to the same module $M$ (because $pq$ is an edge). Since $M^p$ and $M^q$ are strong modules, $M\supseteq M^p \cup M^q$. By Proposition~\ref{prop:modules-other}, it follows that $\Quo(D)$ is not a prime graph, and since it is connected, it must be a clique by Proposition~\ref{prop:modules-quo}. In particular, this implies that $D'$ is connected.
The existence of a strong module containing $D'$ would contradict maximality of $M^p$ or $M^q$, since they are distinct and both contain an element from $D'$. Thus,  every maximal strong module of $D$ is either contained in $D'$ or disjoint from $D'$. It follows that the vertex $u$ satisfies (\ref{cnd:(}).
\end{proof}

In our reasoning, it will be the vertices satisfying condition~\eqref{cnd:(} that will cause the most problems. This justifies the following definition. 

\begin{definition}\label{def:tricky}
In the setting of Lemma~\ref{lem:nei-decomp}, a vertex $u\in N(D)$ that satisfies neither condition~\eqref{cnd:P4} nor condition~\eqref{cnd:pq}---and hence satisfies condition~\eqref{cnd:(}---is called 
{\em{tricky toward $(D,p,q)$}}. By Lemma~\ref{lem:nei-decomp}, if $D$ is not a mesh, then there are no tricky vertices toward $(D,p,q)$. The set of vertices tricky toward $(D,p,q)$ will be denoted by 
$\tricky(D,p,q)$. 
\end{definition}

We can now use the Neighborhood Decomposition Lemma to show that in a $P_6$-free graph, every minimal separator can be covered by the union of neighborhoods of six vertices lying outside of it.
We will sometimes need to have a better understanding of how these six vertices can be chosen, which is encapsulated in the following general statement of the Separator Covering Lemma.

\begin{lemma}[Separator Covering Lemma, general version]\label{lem:covering-general}
Let $G$ be a $P_6$-free graph, and let $S\subseteq V(G)$ be a minimal separator in $G$ with $D_1,D_2\in \cc(G-S)$ being two components full to $S$.
Suppose that each of $D_1,D_2$ has more than one vertex. 
Let $p_1,q_1\in D_1$ be any two vertices of $D_1$ such that the modules of $\Mod(D_1)$ to which $p_1,q_1$ belong are different, but adjacent in $\Quo(D_1)$.
Similarly, let $p_2,q_2\in D_2$ be any two vertices of $D_2$ such that the modules of $\Mod(D_2)$ to which $p_2,q_2$ belong are different, but adjacent in $\Quo(D_2)$.
Then the following holds:
\begin{itemize}
\item If $\Quo(D_1)$ or $\Quo(D_2)$ is not a clique, then $N[p_1,q_1,p_2,q_2]\supseteq S$.
\item If both $\Quo(D_1)$ and $\Quo(D_2)$ are cliques, then there exist vertices $r_1\in D_1$ and $r_2\in D_2$ such that $N[p_1,q_1,r_1,p_2,q_2,r_2]\supseteq S$.
\end{itemize}
\end{lemma}
\begin{proof}

Assume that there is a vertex $u\in S\setminus N[p_1,q_1,p_2,q_2]$ and a $P_4$ of the form $uD_iD_iD_i$ for some $i=1,2$. Consider a shortest path from $u$ to $p_{3-i}$ or $q_{3-i}$ with internal vertices in $D_{3-i}$. Such a path exists, since $D_{3-i}$ is connected and $u$ has a neighbor in $D_{3-i}$ (because $S$ is full to $D_{3-i}$). Moreover, the path is induced and its length is at least $2$, because $u$ is not adjacent to $p_{3-i}$ and $q_{3-i}$. The union of such a path with the $P_4$ of the form $uD_iD_iD_i$ yields a path on at least $6$ vertices, contradicting the $P_6$-freeness of $G$.

Hence, by Lemma~\ref{lem:nei-decomp}, any vertex $u\in S\setminus N[p_1,q_1,p_2,q_2]$ satisfies condition \eqref{cnd:(} with respect to $D_i,p_i$ and $q_i$ for $i=1,2$. Thus, if $S\setminus N[p_1,q_1,p_2,q_2]\neq \emptyset$ both $\Quo(D_1)$ and $\Quo(D_2)$ are cliques.

%Assume that $\Quo(D_1)$ is not a clique (the case when $\Quo(D_2)$ is not a clique is symmetric) and that there is a vertex $u\in S\setminus N[p_1,q_1,p_2,q_2]$. Then, by Lemma~\ref{lem:nei-decomp} there is an induced $P_4$ of the form $uD_1D_1D_1$. 
%Also, by Lemma~\ref{lem:nei-decomp}, there is an induced $P_4$ of the form $uD_2D_2D_2$ or $\Quo(D_2)$ is a clique. The former contradicts $P_6$-freeness of $G$, since we the union of the $P_4$-s yields a $P_7$ (as there are no edges between $D_1$ and $D_2$). Thus, assume that $\Quo(D_2)$ is a clique. 

Now, assume that both $\Quo(D_1)$ and $\Quo(D_2)$ are cliques and  consider quasi-orders $\leq_1$ and $\leq_2$ on vertices of $S\setminus N[p_1,q_1,p_2,q_2]$ defined as $u\leq_i v$ if and only if $N[u]\cap D_i\subseteq N[v]\cap D_i$ for $i=1,2$. We show that these quasi-orders satisfy assumptions of Lemma~\ref{lem:biranking}, in particular, that every two vertices $u,v\in S\setminus N[p_1,q_1,p_2,q_2]$ are comparable either with respect to $\leq_1$ or with respect to $\leq_2$. Suppose that it is not the case. Then, there are vertices $u_1,v_1\in D_1$ and $u_2,v_2\in D_2$ such that $uu_i$ and $vv_i$ are edges and $uv_i$ and $vu_i$ are non-edges for $i=1,2$.

Observe that by Lemma~\ref{lem:nei-decomp}\eqref{cnd:(}, $u_i$ and $v_i$ are adjacent to $p_i$ and $q_i$ and moreover, $u_iv_i$ is an edge for $i=1,2$.
If $uv\in E(G)$, $p_1u_1uvv_2p_2$ forms an induced $P_6$. Otherwise, $p_1u_1uu_2v_2v$ forms an induced $P_6$. This contradicts the $P_6$-freeness of $G$. See Figure~\ref{fig:sep-cover} for illustration.
%\todo[inline]{Draw figures}

\begin{figure}
\centering
\includegraphics[width=\linewidth]{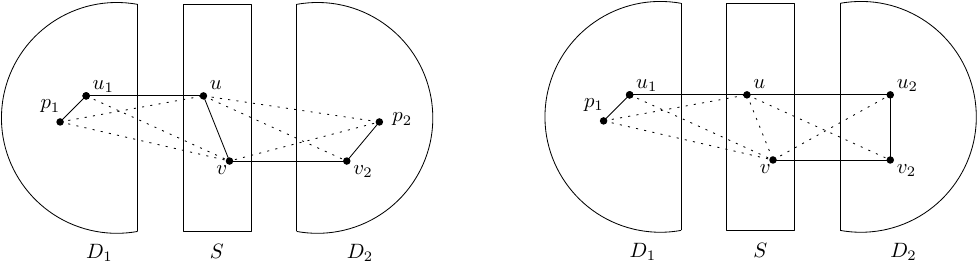}
\caption{Paths contradicting the $P_6$-freeness of $G$.}
\label{fig:sep-cover}
\end{figure}

Then, by Lemma~\ref{lem:biranking}, there exists a vertex $r$ of $S\setminus N[p_1,q_1,p_2,q_2]$ such that $N[r]\cap D_1\subseteq N[w]\cap D_1$ or $N[r]\cap D_2 \subseteq N[w]\cap D_2$ for every $w\in S\setminus N[p_1,q_1,p_2,q_2]$. Recall that $D_1$ and $D_2$ are full to $S$ and therefore $N[r]\cap D_1$ and $N[r]\cap D_2$ are nonempty. Finally, observe that any vertices $r_1\in N[r]\cap D_1$ and $r_2\in N[r]\cap D_2$ satisfy $N[p_1,q_1,r_1,p_2,q_2,r_2]\supseteq S$. 
\end{proof}

However, in most cases we can rely on the following simplified variant.

\begin{lemma}[Separator Covering Lemma, simplified version]\label{lem:covering-simple}
Suppose $G$ is a $P_6$-free graph, and suppose $S\subseteq V(G)$ is a minimal separator in $G$ with $D_1,D_2\in \cc(G-S)$ being two components full to $S$.
Then there exist nonempty sets $A_1\subseteq D_1$ and $A_2\subseteq D_2$ with $|A_1|\leq 3$ and $|A_2|\leq 3$ such that $N[A_1\cup A_2]\supseteq S$.
\end{lemma}
\begin{proof}
If $D_1$ has one vertex, say $V(D_1)=\{u\}$ for some vertex $u$, then $u$ is complete to $S$, since $D_1$ is complete to $S$. Then we can pick $A_1=\{u\}$ and $A_2=\{v\}$, where $v$ is any vertex of $D_2$.
Similarly if $D_2$ has one vertex.
Therefore, let us assume that both $D_1$ and $D_2$ have more than one vertex.
Then the modular partition $\Mod(D_1)$ has at least two modules and the quotient graph $\Quo(D_1)$ is connected, since $D_1$ is connected.
Let us then pick any $p_1\in M^p_1$ and $q_1\in M^q_1$, where $M^p_1$ and $M^q_1$ are modules of $\Mod(D_1)$ that are adjacent.
Symmetrically, we can pick $p_2\in M^p_2$ and $q_2\in M^q_2$, where $M^p_2$ and $M^q_2$ are modules of $\Mod(D_2)$ that are adjacent.
Then, by Lemma~\ref{lem:covering-general}, we either already have that $N[p_1,q_1,p_2,q_2]\supseteq S$, in which case we can set $A_1=\{p_1,q_1\}$ and $A_2=\{p_2,q_2\}$, or
there are vertices $r_1\in D_1$ and $r_2\in D_2$ such that $N[p_1,q_1,r_1,p_2,q_2,r_2]\supseteq S$, so we can set $A_1=\{p_1,q_1,r_1\}$ and $A_2=\{p_2,q_2,r_2\}$.
\end{proof}

\paragraph*{Basic tools on potential maximal cliques}

\paragraph*{Extending PMCs.} First, we recall the approach of lifting PMCs from induced subgraphs, due to Bouchitt\'e and Todinca~\cite{BouchitteT02}.
The following Lemma is implicit in~\cite{BouchitteT02}, we give our proof for the sake of completeness.

\begin{lemma}\label{lem:extending-PMC}
Let $G$ be a graph, $w$ be a vertex in $G$, and $G'=G-w$.
Suppose $\Om'$ is a PMC in $G'$. Then exactly one of the sets $\Om'$ and $\Om'\cup \{w\}$ is a PMC in $G$.
\end{lemma}
\begin{proof}
Observe that $\cc(G-\Om')$ may be constructed from $\cc(G'-\Om')$ by taking all components that are adjacent to $w$, and merging them (and $w$) into one connected component (if there are no such components
in $\cc(G'-\Om')$, then $w$ constitutes a new connected component). Denote by $D_w$ the connected component of $\cc(G-\Om')$ that contains $w$, and let $\Dd_w$ be the set of
connected components of $\cc(G'-\Om')$ that are contained in $D_w$; equivalently, $\Dd_w$ comprises the connected components of $G'-\Om'$ that are adjacent to $w$ in $G$. 
Then it follows that $\cc(G-\Om')=\cc(G'-\Om')\setminus \Dd_w\cup \{D_w\}$ and $N_G(D_w)\supseteq \bigcup_{D\in \Dd_w} N_{G'}(D)$.

Suppose first that $\Om'$ is not a PMC in $G$. Observe that if $uv$ is a nonedge in $\Om'$, then since $\Om'$ is a PMC in $G'$, there exists a component $D\in \cc(G'-\Om')$ that covers this nonedge, that is,
$\{u,v\}\subseteq N(D)$. If $D\notin \Dd_w$ then $D$ is also a component of $G-\Om'$ that covers $uv$, and otherwise, if $D\in \Dd_w$, then $D_w$ covers $uv$. In either case, every nonedge within $\Om'$
is covered by a component of $G-\Om'$, so $\Om'$ is not a PMC in $G$ for the reason of not satisfying the first property: there is a component of $G-\Om'$ whose neighborhood is equal to $\Om'$.
For each $D\in \cc(G-\Om')\setminus \{D_w\}$ we have that $N_{G}(D)=N_{G'}(D)$, and $\Om'$ was a PMC in $G'$, so this neighborhood cannot be equal to $\Om'$. We infer that $N(D_w)=\Om'$.

We now verify that $\Om=\Om'\cup \{w\}$ is a PMC in $G$. First, observe that since $\cc(G-\Om)=\cc(G'-\Om')$ and no component of $G'-\Om'$ was full to $\Om'$ in $G'$, 
it follows that no component of $G-\Om$ is full to $\Om$ in $G$.
Second, take any nonedge $uv$ in $\Om$. If $uv$ is a nonedge in $\Om'$, then $uv$ is covered by some component from $\cc(G-\Om)=\cc(G'-\Om')$, due to $\Om'$ being a PMC in $G'$.
Otherwise, one of endpoints of $uv$ is equal to $w$, say $v=w$; then $u\in \Om'$.
However, we argued that $N(D_w)=\Om'$, hence there exists a component $D\in \Dd_w$ such that $u\in N(D)$. Then it follows that $D$ is a component of $G-\Om$ with $\{u,w\}\subseteq N(D)$, hence this nonedge is
also covered. We conclude that $\Om$ is a PMC in $G$.

Suppose now that $\Om'$ is a PMC in $G$. As $D_w\in \cc(G-\Om')$, there must exist some vertex $u\in \Om'$ for which $u\notin N(D_w)$. As $w\in D_w$, we have $uw\notin E(G)$.
We now verify that $\Om=\Om'\cup \{w\}$ is not a PMC in $G$ by showing that the nonedge $uw$ is not covered by any component of $\cc(G-\Om)$.
Suppose, for the sake of contradiction, that there is some component $D\in \cc(G-\Om)=\cc(G'-\Om')$ for which $\{u,w\}\subseteq N_G(D)$.
Then $D\in \Dd_w$ due to $w\in N_G(D)$, hence $u\in N_{G'}(D)\subseteq N_G(D_w)$. This is a contradiction with the choice of $u$.
\end{proof}

Lemma~\ref{lem:extending-PMC} motivates the following definition. 
Let $\Om$ be a PMC in a graph $G$ and let $\Ss = (x_1, x_2, \ldots, x_k)$ be a sequence of pairwise distinct vertices of $G$.
If for every $0 \leq i \leq k$ the set $\Om \setminus \{x_1,x_2,\ldots,x_i\}$ is a PMC in the graph $G-\{x_1,x_2,\ldots,x_i\}$, then we say that the sequence
$\Ss$ is a \emph{survival sequence} for $\Om$. Furthermore, we denote $V(\Ss) = \{x_1,x_2,\ldots,x_k\}$ and we say that the survival sequence $\Ss$ for $\Om$ \emph{ends in}  $\Om \setminus V(\Ss)$
(which is a PMC in $G-V(\Ss)$).

 Lemma~\ref{lem:extending-PMC} immediately yields the following.

\begin{lemma}[PMC Lifting Lemma]\label{lem:lifting}
Let $G$ be a graph and $\Ss = (x_1,x_2,\ldots,x_k)$ be a sequence of pairwise distinct vertices of $G$.
Then for every $\Om'$ that is a PMC in $G-V(\Ss)$, there exists a unique $\Om$ that is a PMC in $G$ and $\Ss$ is a survival sequence for $\Om$ ending in $\Om'$ (in particular, $\Om \setminus V(\Ss)=\Om'$). Moreover, given $G$, $\Ss$, and $\Om'$, such $\Om$ can be computed in polynomial time.
\end{lemma}
\begin{proof}
Starting from $G'$, reintroduce the vertices of $\Ss$ one by one in the right-to-left order (i.e., starting from $x_k$), up to obtaining $G$ at the end.
For each $i=k,k-1,\ldots,0$, we maintain the unique PMC $\Om_i$ in the graph $G_i := G-\{x_1,x_2,\ldots,x_i\}$ with $(x_{i+1}, x_{i+2}, \ldots, x_k)$ being a survival sequence
for $\Om_i$ in the graph $G_i$ ending in $\Om'$. We start with $\Om_k = \Om'$, and then observe that,
by Lemma~\ref{lem:extending-PMC}, when reintroducing the next vertex $x_i$ there is a unique alternative---include it or not--- that leads to lifting the current PMC to a PMC after the reintroduction.
Observe that at each step we can test in polynomial time which alternative should be followed, because this boils down to testing which of the two sets is a PMC in the graph after the reintroduction.
\end{proof}

\paragraph*{Deducing PMCs.} Next, we provide two auxiliary lemmas that are helpful for reconstructing PMCs in certain situations.

\begin{lemma}\label{lem:deduce-pmc-social}
Suppose $G$ is a $P_6$-free graph and $\Om$ is a PMC in $G$.
Suppose further that $v_1v_2\in \Om$ are such that $v_1v_2\notin E(G)$ and there exists exactly one component $D_0\in \cc(G-\Om)$ for which $\{v_1,v_2\}\subseteq N(D_0)$.
Then for every $D_1,D_2\in \cc(G-\Om)$ different than $D_0$ such that $v_1\in N(D_1)$ and $v_2\in N(D_2)$, the following holds:
\begin{equation}\label{eq:def-om}
\Om=N(D_0)\cup N(D_1)\cup N(D_2)\cup ((N(v_1)\cap N(v_2))\setminus D_0).
\end{equation}
\end{lemma}
\begin{proof}
To see that the right hand side of~\eqref{eq:def-om} is contained in $\Om$, observe that $\Om\supseteq N(D)$ for each $D\in \cc(G-\Om)$, 
whereas all the common neighbors of $v_1$ and $v_2$ are contained in $\Om$ or in $D_0$.
We are left with verifying that every vertex $u\in \Om$ is contained in the right hand side of~\eqref{eq:def-om}.
This is obvious for $u=v_1$ or $u=v_2$, since $\{v_1,v_2\}\subseteq N(D_0)$, hence assume that $u$ is different from $v_1$ and $v_2$.

For the sake of contradiction suppose that $u$ is not contained in the right hand side of~\eqref{eq:def-om}.
Then $uv_1\notin E(G)$ or $uv_2\notin E(G)$, as otherwise $u$ would be contained in $(N(v_1)\cap N(v_2))\setminus D_0$.
By symmetry, suppose w.l.o.g. that $uv_1\notin E(G)$.
Hence, there exists a component $D_3$ with $\{u,v_1\}\subseteq N(D_3)$. Since $u\in N(D_3)$ and $u$ is not contained in the right hand side of~\eqref{eq:def-om},
it follows that $D_3\neq D_t$ for all $t\in \{0,1,2\}$. In particular $D_3\neq D_0$, hence $v_2\notin N(D_3)$ since $D_0$ is the unique component of $G-\Om$ that covers the nonedge $v_1v_2$.
Let $Q_1$ be an induced path with endpoints $v_1$ and $u$, whose all internal vertices belong to $D_3$. Then the length of $Q_1$ is at least $2$.

Define now an induced path $Q_2$ as follows. If $uv_2\in E(G)$, then $Q_2=(u,v_2)$. Otherwise, there is a connected component $D_4$ such that $\{u,v_2\}\subseteq N(D_4)$.
Again, observe that $D_4\neq D_t$ for $t\in \{0,1,2\}$ for the same reason as for $D_3$, and also $D_4\neq D_3$ due to $v_2\in N(D_4)$ and $v_2\notin N(D_3)$.
Similarly as for $D_3$, we have that $v_1\notin N(D_4)$, for $D_0$ is the unique component of $\cc(G-\Om)$ that covers $v_1v_2$.
In this case, we define $Q_2$ to be an induced path with endpoints $u$ and $v_2$ whose all internal vertices belong to $D_4$.
In any case, the path $Q_2$ is of length at least one and has endpoints $u$ and $v_2$.

Let $d_1$ be an arbitrary neighbor of $v_1$ in $D_1$ and $d_2$ be an arbitrary neighbor of $v_2$ in $D_2$.
Observe that $v_2\notin N(D_1)$ and $v_1\notin N(D_2)$, because $D_0$ is the unique component of $\cc(G-\Om)$ that covers $v_1v_2$. 
Construct a path $Q$ by concatenating $Q_1$ and $Q_2$, and appending $d_1$ and $d_2$ at the respective ends of the obtained path.
As $v_1v_2\notin E(G)$, from all the non-adjacencies checked above it follows that $Q$ is an induced $P_\ell$ in $G$ with $\ell\geq 6$. This is a contradiction.
\end{proof}

\begin{lemma}\label{lem:deduce-pmc-solitary}
Suppose $G$ is a graph and $\Om$ is a PMC in $G$.
Suppose further that $v_1v_2\in \Om$ are such that $v_1v_2\notin E(G)$ and there exists exactly one component $D_0\in \cc(G-\Om)$ for which $\{v_1,v_2\}\subseteq N(D_0)$.
Finally, suppose that $D_0$ is the only component of $\cc(G-\Om)$ whose neighborhood contains $v_1$.
Then:
\begin{equation}\label{eq:def-om2}
\Om=(N(v_1)\setminus D_0)\cup N(D_0).
\end{equation}
\end{lemma}
\begin{proof}
By the assumption on $D_0$ being the unique component of $G-\Om$ that has $v_1$ in its neighborhood, it is clear that the right hand side of~\eqref{eq:def-om2} is contained in $\Om$.
Hence, we are left with verifying that every vertex of $\Om$ is contained in the right hand side of~\eqref{eq:def-om2}.

Take any $u\in \Om$. If $u\in N(D_0)$, then we are already done, hence suppose otherwise.
It suffices to prove that $uv_1$ is an edge. Indeed, otherwise there would be a component $D_1\in \cc(G-\Om)$ with $\{u,v_1\}\subseteq N(D_1)$, so in particular $D_1\neq D_0$ because $u\notin N(D_0)$.
This is a contradiction with the assumption that $D_0$ is the unique component that has $v_1$ in its neighborhood.
\end{proof}

\paragraph*{Recovering PMCs.} We now provide two auxiliary lemmas about recovering a family of PMCs based on a rich enough family of structures describing them.

\begin{lemma}[cf. Lemma~\ref{o:lem:recover1}]\label{lem:recover1}
Suppose $G$ is a $P_6$-free graph on $n$ vertices.
Given a family $\Xx\subseteq 2^{V(G)}$, one can in time polynomial in $n$ and $|\Xx|$ compute a family $\Ff_{\rec,1}(\Xx)\subseteq 2^{V(G)}$ 
such that $|\Ff_{\rec,1}(\Xx)|\leq  3n^6\cdot |\Xx|^3$ and the following property
holds: for every PMC $\Om$ in $G$, if $\cc(G-\Om)\subseteq \Xx$ then $\Om\in \Ff_{\rec,1}(\Xx)$.
\end{lemma}
\begin{proof}
Fix an arbitrary enumeration $w_1,w_2,\ldots,w_n$ of $V(G)$, and consider removing vertices $w_i$ from the graph one by one.
Denote $G_i = G - \{w_1,w_2,\ldots,w_i\}$ and $\Om_i = \Om \setminus \{w_1,w_2,\ldots,w_i\}$ for $0 \leq i \leq n$.
Let $0 \leq k \leq n$ be the maximum integer such that $\Om_i$ is a PMC in $G_i$ for every $0 \leq i \leq k$; note that $\Om_0 = \Om$ is a PMC in $G_0 = G$
so such an integer exists. 
That is, $\Om_k$ is a PMC in $G_k$ and $(w_1,w_2,\ldots,w_k)$ is a survival sequence for $\Om$ ending in $\Om_k$, but either $k=n$ or
after the removal of $v = w_{k+1}$, $\Om_k \setminus \{v\}$ is no longer a PMC in $G_k - v$.

In what follows, we mostly analyze the graph $G_k$ with the PMC $\Om_k$ and the connected components of $G_k-\Om_k$.
By the assumption that $\cc(G-\Om) \subseteq \Xx$, 
for a fixed choice of $k$, every connected component $\wD$ of $G_k-\Om_k$ belongs to the following family:
\begin{equation}\label{eq:xxprec}
\widehat{\Xx}_k= \bigcup_{D\in \Xx} \cc(D-\{w_1,\ldots,w_k\}).
\end{equation}
Furthermore, note that we have $|\widehat{\Xx}_k| \leq (n-k) |\Xx|$.

We proceed by a case study depending on whether $k=n$ or not and, in the latter case, where $v = w_{k+1}$ lies.

\bigskip

\noindent {\bf{Case 0}}: $k = n$. Observe that it must be that $(w_1,w_2,\ldots,w_n)$ is a survival sequence for $\Om$ ending in $\emptyset$.
Then, Lemma~\ref{lem:lifting} asserts that there is a unique PMC $\Om$ in $G$ satisfying the above, and that it can be computed in polynomial time. 
We set $\Gg_0 = \{\Om\}$ for such $\Om$.

\bigskip

\noindent {\bf{Case 1}}: $k < n$ and $v\in \Om$. Observe that then it must be that $\Om_k\setminus \{v\}$ is not a PMC in $G_k-v$ because there is some component $\wD$ of $G_k-\Om_k$
such that $N_{G_k}(\wD)=\Om_k\setminus \{v\}$, so after the removal of $v$, $\Om_k\setminus \{v\}$ is equal to the neighborhood of one component. 
Consequently, $\Om_k=N_{G_k}(\wD)\cup \{v\}$.

Construct now a family $\Gg_1$ as follows. For each choice of $0 \leq k < n$ construct $G_k=G-\{w_1,\ldots,w_k\}$ and $\widehat{\Xx}_k$ using formula~\eqref{eq:xxprec}. 
Then, for each choice of $\wD\in \widehat{\Xx}_k$, compute $\Om_k=N_{G_k}(\wD)\cup \{v\}$, where $v=w_{k+1}$.
 If $\Om_k$ is not a PMC in $G_k$, we discard the choice.
Otherwise, by the PMC Lifting Lemma (Lemma~\ref{lem:lifting}) we conclude that $\Om$ is the unique PMC in $G$ for which $(w_1,w_2,\ldots,w_k)$ is a survival sequence
ending in $\Om_k$, and, moreover, $\Om$ can be computed in
polynomial time from $\Om_k$. Hence, we compute $\Om$ and include it in the constructed family $\Gg_1$. 
Observe that thus we obtain that $|\Gg_1|\leq \sum_{k=0}^{n-1} |\widehat{\Xx}_k| \leq  \binom{n+1}{2} \cdot |\Xx|$, and $\Gg_1$ contains $\Om$ provided the situation conforms to this case (i.e., $k < n$ and
    $v\in \Om_k$).

\bigskip

\noindent {\bf{Case 2}}: $k < n$ and $v\notin \Om$. Hence, $\Om_k\setminus \{v\}$ is not a PMC in $G_k-v$ due to the fact that some nonedge $t_1t_2$ with $\{t_1,t_2\}\subseteq \Om_k$
stops to be covered by the connected component $\wD_0$ of $G_k-\Om_k$ that contains $v$, because $\wD_0$ gets shattered by the removal of $v$.
In particular, $\wD_0$ is the unique component of $G_k-\Om_k$ that covers the nonedge $t_1t_2$.

Suppose first that there exist $\wD_1,\wD_2\in \cc(G_k-\Om_k)$ with $\wD_i\neq \wD_0$ for $i=1,2$ such that $t_i\in N_{G_k}(\wD_i)$
for $i=1,2$. Then, by Lemma~\ref{lem:deduce-pmc-social}, applied in $G_k$, we obtain that
\begin{equation}\label{eq:deduced1rec}
\Om_k=N_{G_k}(\wD_0)\cup N_{G_k}(\wD_1)\cup N_{G_k}(\wD_2)\cup ((N_{G_k}(t_1)\cap N_{G_k}(t_2))\setminus \wD_0).
\end{equation}
Obviously $\wD_0,\wD_1,\wD_2\in \widehat{\Xx}_k$, where $\widehat{\Xx}_k$ is defined in~\eqref{eq:xxprec}.

Suppose now that, for one of vertices $t_1,t_2$, say for $t_1$ by symmetry, $\wD_0$ is the unique component of $G_k-\Om_k$ that contains this vertex in its neighborhood.
Then by Lemma~\ref{lem:deduce-pmc-solitary} we have that 
\begin{equation}\label{eq:deduced2rec}
\wOm=(N_{G_k}(t_1)\setminus \wD_0)\cup N_{G_k}(\wD_0).
\end{equation}

Formulae~\eqref{eq:deduced1rec} and~\eqref{eq:deduced2rec} suggest the following construction of a family $\Gg_2$ that encompasses this case.
For each choice of $0 \leq k < n$ construct $G_k=G-\{w_1,\ldots,w_k\}$ and $\widehat{\Xx}_k$ using formula~\eqref{eq:xxprec}. 
Then, for each choice of $\wD_0,\wD_1,\wD_2\in \widehat{\Xx}_k$ and $t_1,t_2 \in V(\wG)$ perform the following.
Compute two candidates for $\Om_k$, using formulae~\eqref{eq:deduced1rec} and~\eqref{eq:deduced2rec}, for each of them 
obtain the unique PMC $\Om$ in $G$ using the PMC Lifting Lemma (Lemma~\ref{lem:lifting}) such that $(w_1,w_2,\ldots,w_k)$
is a survival sequence for $\Om$ ending in $\Om_k$, and include $\Om$ in $\Gg_2$.
Observe that $\Gg_2$ constructed in this manner satisfies $|\Gg_2|\leq \sum_{k=0}^{n-1} 2n^2 \cdot |\widehat{\Xx}_k|^3\leq 2n^2 \binom{n+1}{2}^2\cdot |\Xx|^3$, 
and from the discussion above it follows that $\Gg_2$ contains $\Om$ provided the situation conforms to this case.

\bigskip

To conclude, for $n = 1$ we output $\Ff_{\rec,1} = \{V(G)\}$, while for $n > 1$ we output the family
\begin{equation*}
\Ff_{\rec,1}(\Xx)=\Gg_0 \cup \Gg_1\cup \Gg_2.
\end{equation*}
It follows that $\Ff_{\rec,1}(\Xx)$ has the required property and it is easy to check that for $n > 1$ it holds that
$$|\Ff_{\rec,1}(\Xx)|\leq 1 + \binom{n+1}{2} + 2n^2\binom{n+1}{2}^2 \leq 3n^6\cdot |\Xx|^3.$$
\end{proof}

\begin{lemma}\label{lem:recover2}
Suppose $G$ is a $P_6$-free graph on $n$ vertices.
Given a family $\Xx\subseteq 2^{V(G)}$, one can in time polynomial in $n$ and $|\Xx|$ compute a family $\Ff_{\rec,2}(\Xx)\subseteq 2^{V(G)}$ such that $|\Ff_{\rec,2}(\Xx)|\leq n\cdot |\Xx|$ and the following property
holds: for every PMC $\Om$ in $G$ and each component $D\in \cc(G-\Om)$, if $\Om\cup D\in \Xx$ then $\Om\in \Ff_{\rec,2}(\Xx)$.
\end{lemma}
\begin{proof}
Let us fix a PMC $\Om$ and a component $D\in \cc(G-\Om)$; denote $X=\Om\cup D$.
Since $\Om$ is a PMC, we have that $N(D)\subsetneq \Om$.
Let us pick an arbitrary vertex $s\in \Om\setminus N(D)$.
We claim that
\begin{equation}\label{eq:filter-om}
\Om=(N[s]\cap X)\cup \bigcup_{C\in \cc(G-X)} N(C).
\end{equation}
To see this, observe first that $\cc(G-X)=\cc(G-\Om) \setminus \{D\}$.
Hence, $N(C)\subseteq \Om$ for each $C\in \cc(G-X)$.
Further, $s$ has no neighbors in $D$, hence $N[s]\cap X\subseteq \Om$.
This implies that the right hand side of~\eqref{eq:filter-om} is contained in $\Om$, and we are left with verifying the reverse inclusion.

To see that the right hand side of~\eqref{eq:filter-om} contains $\Om$, pick any vertex $u\in \Om$.
If $us\in E(G)$ then $u\in N[s]\cap X$, and hence $u$ belongs to the right hand side of~\eqref{eq:filter-om}.
Suppose then that $us\notin E(G)$. Since $\Om$ is a PMC, there exists a component $D'\in \cc(G-\Om)$ with $\{u,s\}\subseteq N(D')$.
As $s\in N(D')$, we have $D\neq D'$.
Hence, $D'\in \cc(G-X)$ and, consequently, $u\in N(D')\subseteq \bigcup_{C\in \cc(G-X)} N(C)$.

To conclude, consider the following procedure for constructing $\Ff_{\rec,2}(\Xx)$ based on $\Xx$: For each $X\in \Xx$ and each $s\in X$, add the set $(N[s]\cap X)\cup \bigcup_{C\in \cc(G-X)} N(C)$
to $\Ff_{\rec,2}$. Clearly $|\Ff_{\rec,2}|\leq n\cdot |\Xx|$ by the construction, and the analysis above shows that $\Ff_{\rec,2}$ satisfies the requires properties.
\end{proof}

\paragraph*{Analysis of neighborhoods.} Finally, we give a lemma helpful in understanding the structure imposed between neighborhoods of two components.
Note that the proof of the following lemma has already been presented in Section~\ref{sec:over} (where it bears the name of Lemma~\ref{o:lem:wings}), but we repeat the proof for convenience.

\begin{lemma}[cf. Lemma~\ref{o:lem:wings}]\label{lem:wings}
Let $G$ be a $P_6$-free graph and let $\Omega$ be a PMC in $G$. Suppose $D_1,D_2\in \cc(G-\Omega)$ are two components.
Then either $N(D_1)\setminus N(D_2)$ is complete to $D_1$, or $N(D_2)\setminus N(D_1)$ is complete to $D_2$.
Furthermore, for all $v_1\in N(D_1)\setminus N(D_2)$ and $v_2\in N(D_2)\setminus N(D_1)$ with $v_1v_2\notin E(G)$, 
both $v_1$ is complete to $D_1$ and $v_2$ is complete to $D_2$.
\end{lemma}
\begin{proof}
We start with the first part of the statement.
Towards contradiction suppose that neither $N(D_1)\setminus N(D_2)$ is complete $D_1$, nor $N(D_2)\setminus N(D_1)$ is complete $D_2$.
Then there are $u_1\in N(D_1)\setminus N(D_2)$ and $u_2\in N(D_2)\setminus N(D_1)$ such that $u_1$ is not complete to $D_1$ and $u_2$ is not complete to $D_2$.
Since $D_1$ and $D_2$ are connected, this means that there exist induced $P_3$s: $Q_1$ of the form $u_1D_1D_1$, and $Q_2$ of the form $u_2D_2D_2$.

If $u_1u_2\in E(G)$, then the concatenation of $Q_1$ and $Q_2$ would be an induced $P_6$, a contradiction. Hence suppose $u_1u_2\notin E(G)$.
Since $u_1,u_2\in \Om$ and $\Om$ is a PMC, there exists some component $D_3\in \cc(G-\Om)$ such that $\{u_1,u_2\}\subseteq N(D_3)$.
Since $u_1\notin N(D_2)$ and $u_2\notin N(D_1)$, we have that $D_3\neq D_1$ and $D_3\neq D_2$. Let $R$ be a shortest path connecting $u_1$ and $u_2$ with all internal vertices belonging to $D_3$;
by minimality, $R$ is an induced path. Then the concatenation of $Q_1$, $R$, and $Q_2$ is an induced $P_\ell$ for some $\ell>6$, a contradiction.

For the second part of the statement, suppose for the sake of contradiction that, say, $v_2$ is not complete to $D_2$.
Then there exists a induced paths: a $Q_1$ of the form $v_1D_1$ (which is a $P_2$), and $Q_2$ of the form $v_2D_2D_2$.
Since $v_1v_2\notin E(G)$, in the same manner as for the first statement we can find an induced path $R$, of length at least $2$, with endpoints $v_1$ and $v_2$ and all internal vertices
lying in some component $D_3\in \cc(G-\Om)$ with $D_3\neq D_1$, $D_3\neq D_2$. Then the concatenation of $Q_1$, $R$, and $Q_2$ would be an induced $P_\ell$ for some $\ell\geq 6$, a contradiction.
\end{proof}

\begin{comment}
\begin{lemma}\label{lem:wings}
Suppose $G$ is a $P_6$-free graph, $\Om$ is a PMC in $G$, and $D_1,D_2\in \cc(G-\Om)$ are two components.
Then either $N(D_1)\setminus N(D_2)$ is full to $D_1$ or $N(D_2)\setminus N(D_1)$ is full to $D_2$.
\end{lemma}
\begin{proof}
For the sake of contradiction, suppose that neither $N(D_1)\setminus N(D_2)$ is full to $D_1$ nor $N(D_2)\setminus N(D_1)$ is full to $D_2$.
Then there are $u_1\in N(D_1)\setminus N(D_2)$ and $u_2\in N(D_2)\setminus N(D_1)$ for which there exist $P_3$-s $L_1$ and $L_2$ respectively of the form $u_1D_1D_1$ and $u_2D_2D_2$.
If $u_1u_2\in E(G)$, then $V(L_1)\cup V(L_2)$ would induce a $P_6$, a contradiction.
Hence, as $u_1,u_2\in \Om$, there exist a component $D_3\in \cc(G-\Om)$ such that $u_1,u_2\in N(D_3)$.
Note that $D_3\neq D_1$ and $D_3\neq D_2$ since $u_1\notin N(D_2)$ and $u_2\notin N(D_1)$.
Since $u,v\in N(D_3)$, there is an induced path $Q$ with endpoints $u$ and $v$ whose internal vertices belong to $D_3$.
Then the concatenation of $L_1$, $Q$, and $L_2$ would be an induced path $P_\ell$ for some $\ell>6$, a contradiction.
\end{proof}
\end{comment}

\section{Capturing PMCs}\label{sec:capturing}

Throughout this section we fix a $P_6$-free graph $G$ and an inclusion-wise maximal independent set $I$ in $G$.
We denote $n=|V(G)|$.

\begin{definition}
A PMC $\Om$ in $G$ is called {\em{$I$-crossing}} if $\Om\cap I\neq \emptyset$, and {\em{$I$-free}} if $\Om\cap I=\emptyset$.
\end{definition}

\subsection{Covering easily recognizable components}

We first prove that for a PMC $\Om$, if the neighborhoods of two components of $G-\Om$ do not cover the whole $\Om$, then we can recognize at least one of them.

\begin{lemma}[cf. Lemma~\ref{o:lem:two-not-whole}]\label{lem:two-not-whole}
One can in polynomial time compute a family $\Ff_1$ of size at most $n^8$ such that the following holds: for any PMC $\Om$ in $G$ and any components $D_1,D_2\in \cc(G-\Om)$ with $N(D_1)\cup N(D_2)\subsetneq \Om$,
at least one of $D_1$ or $D_2$ belongs to $\Ff_1$.
Furthermore, if $N(D_1) \subseteq N(D_2)$, then $D_1 \in \Ff_1$.
\end{lemma}
\begin{proof}
By Lemma~\ref{lem:wings}, we have that either $N(D_1)\setminus N(D_2)$ is complete to $D_1$, or $N(D_2)\setminus N(D_1)$ is complete to $D_2$. By symmetry, assume the former.
Consider the graph $G'=G[D_1\cup D_2\cup (N(D_1)\cap N(D_2))]$. This graph is $P_6$-free and $N(D_1)\cap N(D_2)$ is a minimal separator in it, with full components $D_1$ and $D_2$.
Therefore, by Lemma~\ref{lem:covering-simple} there exist nonempty sets $A_1\subseteq D_1$ and $A_2\subseteq D_2$ with $|A_1|,|A_2|\leq 3$ such that $N(D_1)\cap N(D_2)\subseteq N[A_1\cup A_2]$.
Let $A=A_1\cup A_2$. Then $|A|\leq 6$.

Recall that $D_1$ is complete to $N(D_1) \setminus N(D_2)$. Since $A_1$ is nonempty,
we have $N(D_1) \setminus N(D_2) \subseteq N[A]$.
Together with $N(D_1) \cap N(D_2) \subseteq N[A]$, this implies the following:
\begin{align*}
N(D_1)\subseteq& N[A]\subseteq D_1\cup D_2\cup N(D_1)\cup N(D_2)\\
\subseteq& D_1\cup D_2\cup \Omega.
\end{align*}

Recall that by the assumptions of the lemma $\Om\setminus (N(D_1)\cup N(D_2)\neq \emptyset$. Let $s$ be any vertex of $\Om\setminus (N(D_1)\cup N(D_2))$. Denote $Z=\Proj_G(s,N[A])$. First, we claim that $N(D_1)\subseteq Z$.
Indeed, for every $u\in N(D_1)$ either we have $us\in E(G)$, which implies $u\in Z$, or there is component $D'\in \cc(G-\Om)$ with $\{u,s\}\subseteq N(D')$.
Component $D'$ has to be different from $D_1$ and $D_2$ due to $s\in N(D')$, and there is a path $P$ with endpoints in $u$ and $w$ whose all internal vertices belong to $D'$.
Since $N[A]\subseteq D_1\cup D_2\cup \Omega$, we infer that $D'$ is disjoint from $N[A]$, hence $P$ certifies that $u\in Z$.

Next, we observe that $Z\cap D_1=\emptyset$. Indeed, any path from $s$ to a vertex of $D_1$ has to intersect $N(D_1)$, which is contained in $N[A]$.

From these two observations---that $Z$ contains $N(D_1)$ but is disjoint from $D_1$--- it follows that $D_1=\Reach_G(v,Z)$ for any vertex $v\in D_1$.
Hence, consider the following procedure for computing $\Ff_1$: for every choice of $A\subseteq V(G)$ with $|A|\leq 6$ and $s\notin N[A]$, compute $Z=\Proj_G(s,N[A])$, then select any vertex $v\notin Z$,
and add the set $\Reach_G(v,Z)$ to $\Ff_1$. From the above analysis it follows that $\Ff_1$ constructed in this manner contains $D_1$ or $D_2$ for each choice of $\Om,D_1,D_2$ as in the lemma statement,
while $|\Ff_1|\leq n^8$ by the construction. Note that the case of containing $D_2$ arises in the second symmetric case, when $N(D_2)\setminus N(D_1)$ is complete to $D_2$.

If $N(D_1) \subseteq N(D_2)$, then in particular $D_1$ is complete to $N(D_1) \setminus N(D_2) = \emptyset$, so $D_1 \in \Ff_1$.
\end{proof}

The next lemma shows that out of any three components, we can recognize at least one. Thus, we can compute a family that for any PMC $\Om$, contains all but at most two components of $G-\Om$.
Even though this seems very powerful, the issue of recognizing the remaining two components remains.
Note that the proof of the lemma has already been presented in Section~\ref{sec:over} (where it bears the name of Lemma~\ref{o:lem:one-in-three-see}), but we repeat the proof for convenience.

\begin{lemma}[cf. Lemma~\ref{o:lem:one-in-three-see}]\label{lem:one-in-three-see}
One can in polynomial time compute a family $\Ff_2$ of size at most $2n^9$ such that the following holds: for any PMC $\Om$ in $G$ and any pairwise different components $D_1,D_2,D_3\in \cc(G-\Om)$,
at least one of $D_1$, $D_2$, and $D_3$ belongs to $\Ff_2$.
\end{lemma}
\begin{proof}
Fix a PMC $\Om$ and components $D_1,D_2,D_3\in \cc(G-\Om)$.

%MiPi: I think this is not necessary for the proof to work
\begin{comment}
Observe that if $N(D_1)\subseteq N(D_2)$, then by Observation~\ref{obs:hidden} we would have that $N(D_1)$ is a hidden separator; similarly if there is any other inclusion between $N(D_1)$, $N(D_2)$, and $N(D_3)$.
Consider the family $\Ff_\hid$  given by Lemma~\ref{lem:see-hidden}, and define $\Ff'_\hid=\bigcup_{S\in \Ff_\hid} \cc(G-S)$.
Clearly $|\Ff'_\hid|\leq n\cdot |\Ff_\hid|\cdot n\leq n^8$, while $\Ff'_\hid$ contains every vertex subset $D$ such that $N(D)$ is a hidden minimal separator.
Hence, by including the family $\Ff'_\hid$ in $\Ff_2$ we ensure that $\Ff_2$ contains $D_1$, $D_2$, or $D_3$ provided any inclusion between sets $N(D_1)$, $N(D_2)$, $N(D_3)$ holds.
Therefore, from now on we focus on the case when $N(D_1)$, $N(D_2)$, $N(D_3)$ are pairwise incomparable in the inclusion order.
\end{comment}

Observe that if $N(D_1)\cup N(D_2)\subsetneq \Om$, then either $D_1$ or $D_2$ is contained in the family $\Ff_1$ given by Lemma~\ref{lem:two-not-whole};
similarly if $N(D_2)\cup N(D_3)\subsetneq \Om$ or $N(D_3)\cup N(D_1)\subsetneq \Om$. Hence, by including the family $\Ff_1$ given by Lemma~\ref{lem:two-not-whole}
in $\Ff_2$ we ensure that $\Ff_2$ contains $D_1$, $D_2$, or $D_3$ unless the following holds:
\begin{align*}
N(D_1)\cup N(D_2)=&N(D_2)\cup N(D_3)\\
=&N(D_3)\cup N(D_1)=\Om.
\end{align*}
Therefore, from now on we focus on the case when the above holds. 
Summarizing, the PMC $\Om$ can be partitioned into sets $W_1$, $W_2$, $W_3$, and $U$, where 
\begin{align*}
W_i &= \Omega \setminus N(D_i), \\
U &= N(D_1) \cap N(D_2) \cap N(D_3).
\end{align*}
We have (see also Figure~\ref{fig:threecomps}):
\begin{align*}
N(D_1) & = W_2\cup W_3\cup U,\\
N(D_2) & = W_3\cup W_1\cup U,\\
N(D_3) & = W_1\cup W_2\cup U.
\end{align*}

\begin{figure}[htbp]
\begin{center}
\includegraphics{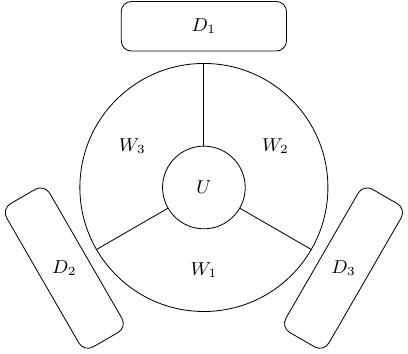}
\end{center}
\caption{Sets in the proof of Lemma~\ref{lem:one-in-three-see}.}\label{fig:threecomps}
\end{figure}
%Moreover, since sets $N(D_1)$, $N(D_2)$, $N(D_3)$ are pairwise incomparable in the inclusion order, we have that none of the sets $W_1$, $W_2$, $W_3$ is empty.\todo{Figure}

Apply Lemma~\ref{lem:wings} to the components $D_2$ and $D_3$ of $G-\Om$, thus inferring that either $W_3=N(D_2)\setminus N(D_3)$ is complete to $D_2$ or $W_2=N(D_3)\setminus N(D_2)$ is complete to $D_3$.
By symmetry, without loss of generality assume the former: $W_3$ is complete to $D_2$.
Next, apply Lemma~\ref{lem:wings} to the components $D_1$ and $D_2$, thus inferring that either $W_1$ is complete to $D_2$, or $W_2$ is complete toward $D_1$. Observe that now these cases are no longer symmetric,
so we consider them one by one.

We examine first the case when $W_2$ is complete to $D_1$.
Consider the graph $G'=G[D_1\cup D_2\cup W_3\cup U]$. Observe that $G'$ is $P_6$-free and that $W_3\cup U$ is a minimal separator in $G'$, with full sides $D_1$ and $D_2$.
By Lemma~\ref{lem:covering-simple} there exist nonempty sets $A_1\subseteq D_1$ and $A_2\subseteq D_2$ with $|A_1|,|A_2|\leq 3$ such that $W_3\cup U\subseteq N[A_1\cup A_2]$.
Since $A_1$ is nonempty and $W_2$ is complete to $D_1$, while $A_1\cup A_2\subseteq D_1\cup D_2$, we have the following:
$$W_2\cup W_3\cup U\subseteq N[A_1\cup A_2]\subseteq D_1\cup D_2\cup \Om.$$
Recall that $W_2\cup W_3\cup U=N(D_1)$, thus $N(D_1)\subseteq N[A_1\cup A_2]$.

Pick any vertices $v_2\in D_2$ and $v_3\in D_3$. Observe that since $N[A_1\cup A_2]$ is disjoint from $D_3$ and contains $N(D_1)$, we have that $\Proj_G(v_3,N[A_1\cup A_2])$ contains $N(D_3)\cap N(D_1)=U\cup W_2$,
and is disjoint from $D_1$. On the other hand, the set $N[v_2]$ contains $W_3$ due to $W_3$ being complete to $D_2$, and is also disjoint from $D_1$.
Consequently, the set $Z=\Proj_G(v_3,N[A_1\cup A_2])\cup N[v_2]$ contains $N(D_1)$ and is disjoint from $D_1$. This implies that $D_1$ is one of the connected components of $G-Z$.

Consider now the family $\Ff'_2$ constructed as follows: For every choice of a set $A\subseteq V(G)$ with $|A|\leq 6$ and $v_2,v_3\in V(G)$ such that $v_3\notin N[A]$,
compute $Z=\Proj_G(v_3,N[A])\cup N[v_2]$ and add $\cc(G-Z)$ to $\Ff'_2$.
It is clear from the construction that $|\Ff'_2|\leq n^9$, and from the above analysis it follows that it contains $D_1$ for any choice of
$\Om,D_1,D_2,D_3$ that conforms to the considered case.

We next examine the case when $W_1$ is complete to $D_2$.
Apply Lemma~\ref{lem:wings} again, this time to the components $D_1$ and $D_3$, implying that either $W_3$ is complete to $D_1$, or $W_1$ is complete to $D_3$.
In the former subcase, when $W_3$ is complete to $D_1$, observe that we can perform exactly the same reasoning as in the first case, but with the roles of $D_1$ and $D_2$ replaced respectively with $D_2$ and $D_1$.
This implies that the family $\Ff'_2$ constructed for the first case contains $D_2$.
In the latter subcase, when $W_1$ is complete to $D_3$, we again can perform exactly the same reasoning as in the first case, but with the roles of $D_1$ and $D_2$ replaces respectively with $D_3$ and $D_2$.
Again, this implies that the family $\Ff'_2$ constructed for the first case contains $D_3$.

To conclude, we set $\Ff_2=\Ff'_2\cup \Ff_1$, where $\Ff_1$ is the family given by Lemma~\ref{lem:two-not-whole}.
Then $|\Ff_2|\leq n^9+n^8\leq 2n^9$, and the above analysis shows that $\Ff_2$ contains at least one of $D_1$, $D_2$, $D_3$ in each of the cases.
\end{proof}

Next, we try to recognize components that are not meshes, as their neighborhoods have a simpler structure in the light of the Neighborhood Decomposition Lemma (Lemma~\ref{lem:nei-decomp}).
Namely, we show that they can be recognized efficiently provided some neighborhood that ``sticks out'' has already been recognized.

\begin{lemma}\label{lem:non-clique-see}
Given a family $\Xx\subseteq 2^{V(G)}$, one can in time polynomial in $n$ and $|\Xx|$ compute a family $\Ff_3(\Xx)\subseteq 2^{V(G)}$ such that $|\Ff_3(\Xx)|\leq n+2n^4\cdot |\Xx|$ and the following property
holds: for every PMC $\Om$ in $G$ and each component $D\in \cc(G-\Om)$ that is not a mesh, if there is another component $D_0\in \cc(G-\Om)$ with $N(D_0)\nsubseteq N(D)$ and 
$D_0\in \Xx$, then $D\in \Ff_3(\Xx)$.
\end{lemma}
\begin{proof}
Let us fix a PMC $\Om$ and components $D,D_0\in \cc(G-\Om)$ satisfying the assumptions. We assume that $|D|\geq 2$, for at the end we will add all singleton sets to the constructed family $\Ff_3(\Xx)$,
thus resolving also the case $|D|=1$.

Since $|D|\geq 2$, the modular partition $\Mod(D)$, which is also the vertex set of $\Quo(D)$, consists of at least two elements. Since $D$ is connected, so is its quotient graph $\Quo(D)$. 
Take any $M^p,M^q\in \Mod(D)$ that are different and adjacent in $\Quo(D)$, and pick any $p\in M^p$ and $q\in M^q$. By Lemma~\ref{lem:nei-decomp}, since $D$ is not a mesh,
we have that for each $u\in N(D)$ either there is an induced $P_4$ of type $uDDD$, or $u\in N[p,q]$.

Denote $W=N(D)\setminus N[p,q]$. Then for each $u\in W$ we have a $P_4$ of type $uDDD$; denote it by $Q_u$.
By Lemma~\ref{lem:pmc-minseps}, there is one connected component $D_\Om$ of $G-N(D)$ that contains all of $\Om\setminus N(D)$ and the vertex sets of all components $D'\in \cc(G-\Om)$ such that $N(D')\nsubseteq N(D)$.
Observe that $W$ has to be complete to $D_\Om$, for otherwise there would be a vertex $u\in W$ for which there is an induced $P_3$ of the form $uD_\Om D_\Om$. By concatenating such an induced $P_3$ with $Q_u$
we would obtain an induced $P_6$, a contradiction. Hence, in particular we have that $W$ is complete to every component $D'\in \cc(G-\Om)$ such that $N(D')\nsubseteq N(D)$, in particular to $D_0$.

We now consider two cases. Suppose first that $N(D)\cup N(D_0)\subsetneq \Om$.
Let $s$ be an arbitrary vertex of $\Om\setminus (N(D)\cup N(D_0))$, and let $Z=N[p,q]\cup N(D_0)$; then $s\notin Z$. Observe that $N(D)\subseteq Z$, because $N(D)\cap N[p,q]\subseteq N[p,q]$ and
$W=N(D)\setminus N[p,q]$ is complete to $D_0$, hence contained in $N(D_0)$.
Similarly as in the proof of Lemma~\ref{lem:two-not-whole}, we observe that $N(D)\subseteq \Proj_G(s,Z)$ and $D\cap \Proj_G(s,Z)=\emptyset$.
Indeed, for the first claim observe that every vertex $u\in N(D)$ is either directly adjacent to $s$, implying $u\in \Proj_G(s,Z)$, or there is a component $D'\in \cc(G-\Om)$ with $\{u,s\}\subseteq N(D')$.
In the latter case, there is a path from $s$ to $u$ whose all internal vertices belong to $D'$. This path certifies that $u\in \Proj_G(s,Z)$, since component $D'$ is different from $D$ and $D_0$ due to $s\in N(D')$. 
For the second claim, that is, $D\cap \Proj_G(s,Z)=\emptyset$, observe that every path from $s$ to a vertex of $D$ necessarily intersects $N(D)$, which is contained in $Z$.

Since $N(D)\subseteq \Proj_G(s,Z)$ and $D\cap \Proj_G(s,Z)=\emptyset$, we conclude that $D$ is one of the connected components of $G-\Proj_G(s,Z)$.
Consider then the following procedure constructing a family $\Ff'_3$: For each $D_0\in \Xx$ and each choice of different vertices $p,q\notin D_0$, construct $Z=N[p,q]\cup N(D_0)$
and, for each choice of $s\notin Z$, add all the components of $\cc(G-\Proj_G(s,Z))$ to $\Ff'_3$. Then we have that $|\Ff'_3|\leq n^4\cdot |\Xx|$, and the analysis above shows that $\Ff'_3$ 
contains $D$ provided the choice of $\Om$ and $D$ conforms to the current case (which is $N(D)\cup N(D_0)\subsetneq \Om$).

We now move to analyzing the second case, namely that $N(D)\cup N(D_0)=\Om$.
Consider again the set $Z=N[p,q]\cup N(D_0)$. Then on one hand $Z\subseteq D\cup \Om$, and on the other hand we have $\Om\subseteq Z$, 
since $\Om\setminus N(D)\subseteq N(D_0)$ and $W=N(D)\setminus N[p,q]$ is complete to $D_0$.
Thus, every connected component of $G-Z$ is either contained in $D$, or is a connected component of $G-\Om$ other than $D$.

Since $\Om$ is a PMC, we have $N(D)\subsetneq \Om$. Let us pick an arbitrary vertex $s\in \Om\setminus N(D)$.
Let $$\Cc=\{C\in \cc(G-Z)\ \colon\ s\in N(C)\}$$ be the set of those connected components of $G-Z$ whose neighborhoods contain $s$.
We claim that
\begin{equation}\label{eq:recover-om}
\Om = \bigcup_{C\in \Cc} N(C) \cup (N[s]\setminus \bigcup_{C\in \Cc} C).
\end{equation}
Indeed, to see that the right hand side of~\eqref{eq:recover-om} is contained in $\Om$, observe that each component of $\Cc$ has to be a component of $G-\Om$ due to containing $s$ in its neighborhood,
and the only neighbors of $s$ outside $\Om$ belong to the components of $\Cc$. For the second inclusion, take any vertex $u\in \Om$.
If $u=s$ or $us\in E(G)$, then $u\in N[s]\setminus \bigcup_{C\in \Cc} C$, thus $u$ is contained in the right hand side of~\eqref{eq:recover-om}.
On the other hand, if $u\notin N[s]$, then since $u\in \Om$ we have that there exists a component $D'\in \cc(G-\Om)$ with $\{u,s\}\subseteq N(D')$.
Since $s\in N(D')$, we have that $D'\neq D$, hence $D'\cap Z=\emptyset$ and $N(D')\subseteq \Om\subseteq Z$. This implies that $D'$ is a connected component of $G-Z$ that contains $s$ in its neighborhood,
so $D'\in \Cc$. Consequently, $u\in N(D')\subseteq \bigcup_{C\in \Cc} N(C)$, which is contained in the right hand side of~\eqref{eq:recover-om}.

Consider then the following procedure constructing a family $\Ff''_3$: For each $D_0\in \Xx$ and choice of different vertices $p,q\notin D_0$, construct $Z=N[p,q]\cup N(D_0)$. Then, for each choice
of $s\in Z$, define $\Cc=\{C\in \cc(G-Z)\ \colon\ s\in N(C)\}$, and construct a candidate for $\Om$ defined as $X=\bigcup_{C\in \Cc} N(C) \cup (N[s]\setminus \bigcup_{C\in \Cc} C)$.
Finally, add $\cc(G-X)$ to $\Ff''_3$. It is clear from the construction that $|\Ff''_3|\leq n^4\cdot |\Xx|$, and the analysis above shows that $\Ff''_3$ contains $D$ 
provided the choice of $\Om$ and $D$ conforms to the current case (which is $N(D)\cup N(D_0)=\Om$).

To conclude, construct the output family $\Ff_3=\Ff_3(\Xx)=\{\{u\}\colon u\in V(G)\}\cup \Ff'_3\cup \Ff''_3$. Then $|\Ff_3|\leq n+2n^4\cdot |\Xx|$ and the discussion above shows that $\Ff_3$ satisfies
the required properties.
\end{proof}

\subsection{Capturing separators and fuzzy components}

We now move to recognizing components that are meshes.
Unfortunately, we cannot do this in full exactness, but we can recognize a set consisting of the component and some elements of the PMC complete to it.
It will later appear that such a weaker structure is good enough for our purposes.

More formally, for a minimal separator $S$ and a component $D$ of $G-S$ that is
full towards $S$, we say
that a set $D^+$ is a \emph{fuzzy version} of $D$ if $D \subseteq D^+ \subseteq D \cup S$
and every vertex $v \in D^+ \setminus D$ is complete to $D$.

The most important insight from knowing a fuzzy version of a component $D$ is that we
can learn a lot about the modular decomposition of $D$ by the following straightforward
observation.
\begin{observation}\label{obs:fuzzy-module}
Let $S$ be a minimal separator and $D$ be a component of $G-S$ that is full towards $S$.
Let $X \subseteq V(G)$ be such that $X \cap D = \emptyset$ but $X$ contains all vertices
of $S$ that are not complete to $D$. 
Then $D$ is a module of $G-X$. In particular, if $D^+$ is a fuzzy version of $D$, then
$D$ is a module of $G[D^+]$.
\end{observation}

We first use the above observation to capture all minimal separators that have two
non-mesh full components.

\begin{lemma}\label{lem:capture-nonmesh}
One can in polynomial time compute a family $\Ff_4$ of size at most $2n^3$ such that the
following holds: for every minimal separator $S$ in $G$ such that $G-S$ has at least
two non-mesh components that are full to $S$, every connected component of $G-S$ that is not
a mesh and is full to $S$ is contained in $\Ff_4$.
\end{lemma}
\begin{proof}
Let $S$ be a separator as in the statement of the lemma,
and let $D_1$ be a component of $G-S$ that is not a mesh and that is full to $S$.
By the assumptions on $S$, there exists another non-mesh component $D_2$
that is full to $S$.

Let $p_2,q_2 \in D_2$ be any two vertices of $D_2$ such that the modules of
$\Mod(D_2)$ to which $p_2,q_2$ belong are different, but adjacent in $\Quo(D_2)$.
If $|D_2| = 1$, we pick $p_2=q_2$ to be the unique vertex of $D_2$; note that then $N(p_2) = S$.
Let $X = N[p_2,q_2]$.
By the Neighborhood Decomposition Lemma (Lemma~\ref{lem:nei-decomp}),
since $D_2$ is not a mesh, for every vertex $u \in S \setminus X$
there exists a $P_4$ of the form $uD_2D_2D_2$; note that this is also true
in the case $|D_2|=1$ as then $X = S \cup D_2$.
Consequently, since $G$ is $P_6$-free, every vertex $u \in S \setminus X$ is complete to $D_1$,
as otherwise the fact that $D_1$ is full to $S$ implies that 
there would exist a $P_3$ of the form $uD_1D_1$, which together with the $P_4$ of the form
$uD_2D_2D_2$ yields a $P_6$ in $G$.

By Observation~\ref{obs:fuzzy-module}, $D_1$ is a module of $G-X$.
Furthermore, note that $D_1$ is connected but is not a mesh.
Consequently, Corollary~\ref{cor:enum-modules} yields a family of at most $2n-1$
candidates for $D_1$, given the graph $G-X$.

To sum up, we can define the family $\Ff_4$ as follows: for every two (not necessarily distinct)
vertices $p,q \in V(G)$,
we insert into $\Ff_4$ all modules yielded by Corollary~\ref{cor:enum-modules}
applied to $G-N[p,q]$.
Since there are at most  $2n-1$ such modules for a fixed choice of $p$ and $q$,
   the bound $|\Ff_4| \leq 2n^3$ follows.
\end{proof}

For mesh components, the argumentation of Lemma~\ref{lem:capture-nonmesh} breaks down
as, although $D_1$ is still a module of $G-N[p_2,q_2]$, it may be a union of an arbitrary
subset of modules of a clique node in the modular decomposition of $G-N[p_2,q_2]$, giving
us too many choices. In the following lemma we resort to capturing a fuzzy version of
a mesh component, where a second full side of the separator is not a mesh.

\begin{lemma}\label{lem:capture-mesh-nomesh}
One can in polynomial time compute a family $\Ff_5$ of size at most $n^4$ such that the
following holds: for every minimal separator $S$ in $G$ and every component $D_1$ of $G-S$
that is a mesh and is full to $S$, if there exists a different component $D_2$ of $G-S$
that is full to $S$ and not a mesh, then some fuzzy version of $D_1$ belongs to $\Ff_5$.
\end{lemma}
\begin{proof}
Let $S$, $D_1$, and $D_2$ be as in the lemma statement.
For $i=1,2$, we pick $p_i,q_i \in D_i$ as in Lemma~\ref{lem:capture-nonmesh}:
let them be any two vertices of $D_i$ such that the modules of
$\Mod(D_i)$ to which $p_i,q_i$ belong are different, but adjacent in $\Quo(D_i)$.
If $|D_i| = 1$, we pick $p_i=q_i$ to be the unique vertex of $D_i$.
By the Neighborhood Decomposition Lemma (Lemma~\ref{lem:nei-decomp}),
   since $D_2$ is not a mesh,
for every vertex $u \in S \setminus N[p_2,q_2]$
there exists a $P_4$ of the form $uD_2D_2D_2$.
Consequently, since $G$ is $P_6$-free, every vertex $u \in S \setminus N[p_2,q_2]$ is complete to $D_1$.

Since $D_1$ is a mesh, we have $D_1 \subseteq N[p_1,q_1] \subseteq D_1 \cup S$.
Consequently, the set $D_1^+ := N[p_1,q_1] \setminus N[p_2,q_2]$ is a fuzzy version of $D_1$.
Thus, a family $\Ff_5$ consisting of sets $N[p_1,q_1] \setminus N[p_2,q_2]$
for every choice $p_1,q_1,p_2,q_2 \in V(G)$ satisfies the requirements of the lemma.
\end{proof}

We now provide a lemma about capturing a fuzzy version of a mesh component
near a well-structured PMC.

\begin{lemma}\label{lem:mesh-see-sticking}
Given a family $\Xx\subseteq 2^{V(G)}$, one can in time polynomial in $n$ and $|\Xx|$ compute a family $\Ff_6(\Xx)\subseteq 2^{V(G)}$ such that $|\Ff_6(\Xx)|\leq n^3\cdot |\Xx|$ and the following property
holds: for every PMC $\Om$ in $G$ and each component $D\in \cc(G-\Om)$ that is a mesh, 
if there is another component $D_0\in \cc(G-\Om)$ with $N(D_0)\nsubseteq N(D)$ and 
$D_0\in \Xx$, then $\Ff_6(\Xx)$ contains some fuzzy version of $D$.
\end{lemma}
\begin{proof}
Let us fix a PMC $\Om$ and components $D$ and $D_0$ as in the lemma statement. 
Pick any vertices $p,q\in D$ vertices respectively belonging to different proper strong modules $M^p$ and $M^q$ in $D$. 
Since $D$ is a mesh, $\Quo(D)$ is a clique, and hence every vertex of $D$ is complete either to $M^p$ (unless it is contained in $M^p$) or to $M^q$ (unless it is contained in $M^q$).
As a result, we have that every vertex of $D$ is adjacent either to $p$ or to $q$, hence
$$D\subseteq N[p,q]\subseteq D\cup N(D).$$
Since $N(D_0)\nsubseteq N(D)$, let us pick an arbitrary vertex $s\in N(D_0)\setminus N(D)$.
Define
$$D^+=N[p,q]\setminus (N(D_0)\cup N(s)).$$
Since $N(D_0)\subseteq \Om$ and $s\notin N(D)$, we have that $D\subseteq D^+\subseteq D\cup N(D)$.
Take any vertex $u\in D^+\setminus D$. Then $u\in N(D)\setminus N(D_0)$ and $us\notin E(G)$. Since $s\in N(D_0)\setminus N(D)$, by Lemma~\ref{lem:wings}, the second statement,
we infer that $u$ is complete to $D$. Hence, $D^+$ consists of $D$ and some vertices of $\Om$ that are complete to $D$.

Concluding, consider the following procedure for computing $\Ff_6=\Ff_6(\Xx)$. For each $D_0\in \Xx$ and each choice of $p,q,s\notin D_0$, add the set $N[p,q]\setminus (N(D_0)\cup N(s))$
to $\Ff_6$. Thus, $|\Ff_6|\leq n^3\cdot |\Xx|$, and the above analysis shows that $\Ff_6$ satisfies the required property.
\end{proof}

\subsection{Recognizing PMCs with many non-hidden components}

For a PMC $\Om$ in $G$, we say that a component $D \in \cc(G-\Om)$ is 
\emph{neighbor-maximal} if there is no component $D'\in \cc(G-\Om)$ with $D'\neq D$ and $N(D)\subseteq N(D')$.
Note that in the above definition, if there are two different components $D_1,D_2\in \cc(G-\Om)$ with $N(D_1)=N(D_2)$, then neither of them is neighbor-maximal.

Our goal in this section is to focus on $I$-free PMCs $\Om$ that have at least three neighbor-maximal components.
Intuitively, such PMCs are of crucial importance as they constitute ``branching points'' in any clique tree of the corresponding $I$-free chordal completion.
Unfortunately, we will not achieve this goal exactly, as the number of such PMCs in a $P_6$-free graph can be exponential.%
\footnote{One possible example here is as follows. Consider an $n$-prism, that is, a graph consisting of $2n$ vertices labelled $v_1,v_2,\ldots,v_n,u_1,u_2,\ldots,u_n$
  with $\{u_i~|~1 \leq i \leq n\}$ inducing a clique, $\{v_i~|~1 \leq i \leq n\}$ inducing a clique, and edges $u_iv_i$ for $1 \leq i \leq n$. 
    Furthermore, assume $n \geq 4$ and add a vertex $s$ adjacent to $u_1$, $u_2$, $v_1$, and $v_2$. Then
    it is easy to verify that the constructed graph is $P_6$-free and for every $I \subseteq \{3,4,\ldots,n\}$, $I \neq \emptyset, \{3,4,\ldots,n\}$ the set
    $\{u_1,u_2,v_1,v_2\} \cup \{u_i~|~i \in I\} \cup \{v_i~|~i \in \{3,4,\ldots,n\} \setminus I\}$ is a PMC
    with three neighbor-maximal components: $\{s\}$, $\{u_i~|~i \in \{3,4,\ldots,n\} \setminus I\}$, and $\{v_i~|~i \in I\}$.}
However, situations where such a PMC $\Om$ cannot be captured are restricted to some very special cases, where we can recover $\Om$ ``mixed'' with two mesh components $D_1,D_2\in \cc(G-\Om)$
for which we can obtain their fuzzy versions. Formally, in this section we prove the following lemma.

\begin{lemma}[cf.~Lemma~\ref{o:lem:summary}]\label{lem:summary}
One can in polynomial time compute families $\Ff^1_9$ and $\Ff^2_9$, each of size at most $10^{13}\cdot n^{159}$%TODO-maybe $10^{13}\cdot n^{156}$
, such that the following holds: for any $I$-free PMC $\Om$ in $G$ with at least three neighbor-maximal components,
either $\Ff^1_9$ contains $\Om$, or $\Ff^2_9$ contains triple $(\Om\cup D_1\cup D_2,D_1^+,D_2^+)$ for some components $D_1,D_2\in \cc(G-\Om)$ that are meshes,
where $D_i^+$ is a fuzzy version of $D_i$, for $i=1,2$.
\end{lemma}

The proof of Lemma~\ref{lem:summary} builds on two technical results. 
The most technical one, Lemma~\ref{lem:monster}, allows us to deduce a small family of candidates
for a PMC $\Om$ if we are given a family containing candidates for \emph{all but one} connected components of $G-\Om$; compare it with Lemma~\ref{lem:recover1} where we have a family containing
\emph{all} components of $G-\Om$.
The second one, Lemma~\ref{lem:Omplus}, allows us to deduce a small family of candidates
for $\Om \cup D_1 \cup D_2$ for a PMC $\Om$ with two mesh neighbor-maximal components
$D_1$ and $D_2$ that are hard to deduce by other means.

In the next section, it will be more convenient for us to use the technical statements
of Lemmata~\ref{lem:monster} and~\ref{lem:Omplus} directly,
rather than refer the general statement of Lemma~\ref{lem:summary}. 
However, since the statement of Lemma~\ref{lem:summary} sets a clear milestone in our paper, 
  while its proof is simple and direct (given Lemmata~\ref{lem:monster} and~\ref{lem:Omplus}), we decide to include it here for clarity of presentation.

We start with the following simple observation on how $I$ interacts with modules of a mesh.

\begin{observation}\label{obs:mesh-one-with-I}
Let $D\subseteq G$ be a connected induced subgraph of $G$ such that $D$ is a mesh.
Then there is at most one module $M\in \Mod(D)$ such that $I\cap M$ is nonempty.
Furthermore, if $D\in \cc(G-\Om)$ for some $I$-free PMC $\Om$, then $I\cap D$ is nonempty and, consequently, there exists exactly one module $M\in \Mod(D)$ such that $I\cap M$ is nonempty. 
\end{observation}
\begin{proof}
For the first claim, suppose $I\cap M\neq \emptyset$ and $I\cap M'\neq \emptyset$ for some different $M,M'\in \Mod(D)$.
However, since $D$ is a mesh, $\Quo(D)$ is a clique, and $M$ and $M'$ are complete to each other.
This contradicts the assumption that $I$ is an independent set.

For the second claim, suppose that $D\in \cc(G-\Om)$ for some $I$-free PMC $\Om$, however $I\cap D=\emptyset$.
Take any $u\in D$, and observe that $N(u)\subseteq D\cup \Om$ is disjoint from $I$, since $\Om$ is $I$-free.
This means that $I\cup \{u\}$ would be still an independent set, which contradicts the maximality of $I$.
\end{proof}

Observation~\ref{obs:mesh-one-with-I} justifies the setting of the next lemma, where we use the Bi-ranking Lemma (Lemma~\ref{lem:biranking}) to cover an independent set among tricky vertices (Definition~\ref{def:tricky}).
This tool will be used later.

\begin{lemma}\label{lem:cover-tricky}
Suppose $\Om$ is an $I$-free PMC in $G$ and $D\in \cc(G-\Om)$ is a mesh.
Let $p,q\in D$ belong to different proper strong modules in $D$, where the proper strong module to which $p$ belongs is the unique element of $\Mod(D)$ that intersects $I$. 
Suppose further that $J\subseteq \tricky(D,p,q)$ is a non-empty independent set.
Then there exist a vertex $w\in D$ and a component $D'\in \cc(G-\Om)$ with $D'\neq D$ such that $J\subseteq N(w)\cup N(D')$.
\end{lemma}
\begin{proof}
Let $A=\Mod(D)$ and $B=\cc(G-\Om)\setminus \{D\}$.
Recall that since $J\subseteq \tricky(D,p,q)$, the neighborhood of each vertex of $J$ within $D$ is formed by the union of a collection of proper strong modules of $D$, none of which contains $p$ or $q$.
For each $u\in J$, construct sets $X_u\subseteq A$ and $Y_u\subseteq B$ as follows: $X_u$ comprises all proper strong modules of $D$ adjacent to $u$, whereas $Y_u$ comprises all components of
$\cc(G-\Om)$ that contain $u$ in their neighborhoods, apart from $D$ itself.

Consider quasi-orders $\leq_1$ and $\leq_2$ on $J$ defined as follows: $u\leq_1 v$ if and only if $X_u\subseteq X_v$ and $u\leq_2 v$ if and only if $Y_u\subseteq Y_v$.
We now verify that $\leq_1$ and $\leq_2$ satisfy the prerequisites of the Bi-ranking Lemma (Lemma~\ref{lem:biranking}), that is, for all $u,v\in J$, $u$ and $v$ are either comparable in $\leq_1$
or in $\leq_2$. Suppose the contrary: all the four sets $X_u\setminus X_v$, $X_v\setminus X_u$, $Y_u\setminus Y_v$, and $Y_v\setminus Y_u$ are nonempty.
Let us pick arbitrary $M_1\in X_u\setminus X_v$, $M_2\in X_v\setminus X_u$, $D_1\in Y_u\setminus Y_v$, and $D_2\in Y_v\setminus Y_u$.
Further, pick arbitrary vertices $m_1\in M_1$, $m_2\in M_2$, $d_1\in D_1\cap N(u)$, and $d_2\in D_2\cap N(v)$.
Since $M_1\notin X_v$, we have $m_1v\notin E(G)$, and symmetrically $m_2u\notin E(G)$.
Since $D_1\notin Y_v$, we have $d_1v\notin E(G)$, and symmetrically $d_2u\notin E(G)$.
Since $D,D_1,D_2$ are pairwise different components of $G-\Om$, we have that there is no edge between $\{m_1,m_2\}$ and $\{d_1,d_2\}$, and that $d_1d_2\notin E(G)$. 
Since $M_1$ and $M_2$ are two different modules of $\Mod(D)$ and $D$ is a mesh, we have $m_1m_2\in E(G)$.
Finally, since $J$ is independent, we have $uv\notin E(G)$. 
We conclude that $(d_1,u,m_1,m_2,v,d_2)$ would be an induced $P_6$ in $G$, a contradiction.

Therefore, we may apply the Bi-ranking Lemma to quasi-orders $(J,\leq_1)$ and $(J,\leq_2)$. This yields a vertex $u\in J$ such that for every $v\in J$, we have $X_u\subseteq X_v$ or $Y_u\subseteq Y_v$.
Since $u\in N(D)$, obviously $X_u$ is nonempty.
We claim that also $Y_u$ is nonempty.
Indeed, otherwise we would have that $N[u]$ is contained in $D\cup \Om$. Observe that $N[u]$ is disjoint from the unique module of $\Mod(D)$ that intersects $I$, since this is the module
of $\Mod(D)$ to which $p$ belongs.
As $\Om$ is $I$-free, we conclude that in this case $N[u]$ would be disjoint from $I$, which means that $I\cup \{u\}$ would be still an independent set. This is a contradiction with the maximality of $I$.

Since $X_u\neq \emptyset$, let us pick any neighbor $w$ of $u$ in $D$. Observe that the condition $X_u\subseteq X_v$ for some $v\in J$ means $N(u)\cap D\subseteq N(v)\cap D$, which in particular implies
that $v\in N(w)$.
Next, since $Y_u\neq \emptyset$, let us pick an arbitrary component $D'\in Y_u$; obviously $D'\neq D$.
Observe that the condition $Y_u\subseteq Y_v$ for some $v\in J$ implies that $D'\in Y_v$, equivalently $v\in N(D')$.
Since for each $v\in J$ we have $X_u\subseteq X_v$ or $Y_u\subseteq Y_v$, we conclude that $J\subseteq N(w)\cup N(D')$, as claimed.
\end{proof}

We now proceed to the most technical lemma. Intuitively, it says that provided we have already recognized all but one connected components of $G-\Om$, for an $I$-free PMC $\Om$, 
the last component can be also recognized at the cost of taking a larger family of candidates.

\begin{lemma}\label{lem:monster}
Given a family $\Xx\subseteq 2^{V(G)}$, one can in time polynomial in $n$ and $|\Xx|$ compute a family $\Ff_7(\Xx)\subseteq 2^{V(G)}$ such that $|\Ff_7(\Xx)|\leq 11n^{12}\cdot |\Xx|^3$
%TODO-maybe 10n^{11}\cdot |\Xx|^3
and the following property
holds: for every $I$-free PMC $\Om$ in $G$ and each component $D\in \cc(G-\Om)$, if $\cc(G-\Om)\setminus \{D\}\subseteq \Xx$, then $\cc(G-\Om)\subseteq \Ff_7(\Xx)$.
\end{lemma}
\begin{proof}
Let us fix any vertex $s\in \Om\setminus N(D)$.
Consider first the case when $N(D)\subseteq N(s)$. 
Then it follows that $D$ is contained in $\cc(G-N(s))$. Hence, if we make sure that the family $\Gg_1=\bigcup_{u\in V(G)} \cc(G-N(u))$ is contained in the output family 
$\Ff_7(\Xx)$, then we ensure that components $D$ conforming to this case are contained in $\Ff_7(\Xx)$. Note that $|\Gg_1|\leq n^2$, so for the rest of the proof we may assume that 
there exists some vertex $u\in N(D)$ such that $us$ is a nonedge.

Next, consider the case when $D$ is not a mesh. Since $us$ is a nonedge in $\Om$ and $s\notin N(D)$, there exists some other component $D_0\in \cc(G-\Om)$ such that $\{u,s\}\subseteq N(D_0)$.
In particular $N(D_0)\nsubseteq N(D)$.
Since $D_0\in \Xx$ by assumption, it follows that the family $\Gg_2=\Ff_3(\Xx)$ provided by Lemma~\ref{lem:non-clique-see} contains $D$.
Hence, if we include this family in the construction of $\Ff_7(\Xx)$, we ensure that every component $D$ conforming to this case is included in the construction.
Since $|\Gg_2|\leq n+2n^4\cdot |\Xx|$, we may proceed with the assumption that $D$ is a mesh.

By applying the procedure of Lemma~\ref{lem:mesh-see-sticking}, due to the existence of $D_0$, we can compute a family $\Gg_{\mathrm{fuzzy}}=\Ff_6(\Xx)$ of size at most $n^3\cdot |\Xx|$ such 
that $\Gg_{\mathrm{fuzzy}}$ contains a fuzzy version $D^+$ of $D$. Our goal for the rest of the proof is to filter out those additional vertices of $D^+\setminus D$.

Fix an arbitrary enumeration $\{w_1,w_2,\ldots,w_p\}$ of $N(s)$, where $p=|N(s)|$. Consider removing vertices $w_1,w_2,\ldots,w_p$ one by one up to the moment when $\Om$ intersected with the current graph
stops to be a PMC in this graph; this happens at some point due to the nonedge $us$. 
More formally, for every $0 \leq i \leq p$, let $G_i = G - \{w_1,w_2,\ldots,w_i\}$ and $\Om_i = \Om \setminus \{w_1,w_2,\ldots,w_i\}$.
Let $0 \leq k < p$ be a maximum integer such that $\Om_i$ is a PMC in $G_i$ for every $0 \leq i \leq k$.
Note that such $k$ exists as $\Om = \Om_0$ is a PMC in $G = G_0$,
while the nonedge $us$ (with $u,s \in \Om_p$) is not covered in $\Om_p$ in $G_p$, because $s$ is an isolated vertex in $G_p$.
We denote $v = v_{k+1}$.
In the following, whenever we write $N(\cdot)$ we mean $N_G(\cdot)$, and for neighborhoods in $G_k$ we write $N_{G_k}(\cdot)$.
Observe that $D$ is among the components of $G_k-\Om_k$, since we remove only neighbors of $s$, which do not belong to $D$.

Consider a component $\wD \in \cc(G_k - \Om_k)$. 
If $\wD \neq D$, then there exists some component $D' \in \cc(G-\Om)$ such that $\wD \subseteq D' \setminus \{w_1,w_2,\ldots,w_k\}$ and $D' \in \Xx$ by assumption.
Let 
\begin{equation}\label{eq:xxp}
\widehat{\Xx}_k=\bigcup_{D'\in \Xx} \cc(D'-\{w_1,\ldots,w_k\}).
\end{equation}
Then observe that $|\widehat{\Xx}_k|\leq (n-k)\cdot |\Xx|$, $\widehat{\Xx}$ can be computed based on $\Xx$ in polynomial time, and it holds that $\wD\in \widehat{\Xx}$
for every $\wD \in \cc(G_k - \Om_k) \setminus \{D\}$.

We now perform a case study depending on where $v$ lies and how $\Om_k$ stops to be a PMC after removing $v$.
This case study is similar (with parts of reasoning copied verbatim), but far more complicated than the one we performed in the proof of Lemma~\ref{lem:recover1}.

\bigskip

\noindent {\bf{Case 1}}: $v\in \Om$. Observe that then it must be that $\Om_k\setminus \{v\}$ is not a PMC in $G_k-v$ because there is some component $\wD_0$ of $G_k-\Om_k$
such that $N_{G_k}(\wD_0)=\Om_k\setminus \{v\}$, so after the removal of $v$, $\Om_k\setminus \{v\}$ is equal to the neighborhood of one component. 
Consequently, $\Om_k=N_{G_k}(\wD_0)\cup \{v\}$.
Observe that since $s\notin N(s)$, we do not remove $s$ in the removal process, and hence $s\in \Om_k\setminus\{v\}$, so also $s\in N_{G_k}(\wD_0)$.

Construct now a family $\Gg_3$ as follows. For each choice of 
\begin{itemize}
\item $s\in V(G)$, 
\item $0 \leq k < p$, with an implicit choice of $v = w_{k+1}$, and
\item $\wD_0\in \widehat{\Xx}_k$, where $\widehat{\Xx}_k$ is computed using formula~\eqref{eq:xxp}, 
\end{itemize}
perform the following. 
Compute $G_k=G-\{w_1,\ldots,w_k\}$ and $\Om_k=N_{G_k}(\wD_0)\cup \{v\}$. 
If $\Om_k$ is not a PMC in $G_k$, we discard the choice.
Otherwise, by the PMC Lifting Lemma (Lemma~\ref{lem:lifting}) we conclude that $\Om$ is the unique PMC in $G$ 
for which $(w_1,w_2,\ldots,w_k)$ is a survival sequence ending in $\Om_k$, and moreover $\Om$ can be computed in
polynomial time from $\Om_k$. Hence, we compute $\Om$ and include $\cc(G-\Om)$ in the constructed family $\Gg_3$. 
Observe that thus we obtain that 
$$|\Gg_3|\leq \sum_{k=0}^{p-1} n^2\cdot |\widehat{\Xx}_k|\leq n^4\cdot |\Xx|,$$ and $\Gg_3$ contains $D$ provided the situation conforms to this case (i.e., $v\in \Om$).

\bigskip

In the remaining cases we have $v\notin \Om$. Hence, $\Om_k\setminus \{v\}$ is not a PMC in $\wG-v$ due to the fact that some nonedge $t_1t_2$ with $\{t_1,t_2\}\subseteq \Om_k$
stops to be covered by the connected component $\wD_0$ of $G_k-\Om_k$ that contains $v$, because $\wD_0$ gets shattered by the removal of $v$.
In particular, $\wD_0$ is the unique component of $G_k-\Om_k$ that covers the nonedge $t_1t_2$.
Observe that since $v$ is a neighbor of $s$ and is contained in $\wD_0$, we have that the connected component of $G-\Om$ that contains $\wD_0$ is different from $D$; this is because $s\notin N(D)$.
Consequently, we have that $\wD_0$ and $D$ are different connected components of $G_k-\Om_k$.

We proceed by analyzing to which sets the vertices $t_1,t_2$ may belong.

\bigskip

\noindent {\bf{Case 2}}: $\{t_1,t_2\}\subseteq \Om_k-N(D)$.
Suppose first that there exist $\wD_1,\wD_2\in \cc(G_k-\Om_k)$ with $\wD_i\neq \wD_0$ for $i=1,2$ such that $t_i\in N_{G_k}(\wD_i)$
for $i=1,2$. Then, by Lemma~\ref{lem:deduce-pmc-social}, applied in $G_k$, we obtain that
\begin{equation}\label{eq:deduced1}
\begin{split}
\Om_k=&N_{G_k}(\wD_0)\cup N_{G_k}(\wD_1)\cup N_{G_k}(\wD_2)\cup\\
&\cup ((N_{G_k}(t_1)\cap N_{G_k}(t_2))\setminus \wD_0).
\end{split}
\end{equation}
Let $D_1,D_2\in \cc(G-\Om)$ be such that $\wD_1\subseteq D_1$ and $\wD_2\subseteq D_2$.
Since both $N_{G_k}(\wD_1)$ and $N_{G_k}(\wD_2)$ contain vertices outside of $N(D)$, being $t_1$ and $t_2$ respectively, we infer that $D_1\neq D$ and $D_2\neq D$.
This implies that $D_1,D_2\in \Xx$, so also $\wD_1,\wD_2\in \widehat{\Xx}_k$.
Recall that also $\wD_0\in \widehat{\Xx}_k$.

Suppose now that, for one of vertices $t_1,t_2$, say for $t_1$ by symmetry, $\wD_0$ is the unique component of $G_k-\Om_k$ that contains this vertex in its neighborhood.
Then by Lemma~\ref{lem:deduce-pmc-solitary} we have that 
\begin{equation}\label{eq:deduced2}
\Om_k=(N_{G_k}(t_1)\setminus \wD_0)\cup N_{G_k}(\wD_0).
\end{equation}

Formulas~\eqref{eq:deduced1} and~\eqref{eq:deduced2} suggest the following construction of a family $\Gg_4$ that encompasses this case.
For each choice of
\begin{itemize}
\item $s\in V(G)$, 
\item $0 \leq k < p$ with an implicit choice of $v = w_{k+1}$,
\item $\wD_0,\wD_1,\wD_2\in \widehat{\Xx}$,
\item $t_1,t_2 \in V(\wG)$,
\end{itemize}
perform the following.
Compute two candidates for $\Om_k$, using formulae~\eqref{eq:deduced1} and~\eqref{eq:deduced2}, for each such $\Om_k$
apply Lemma~\ref{lem:lifting} to compute a unique PMC $\Om$ in $G$ for which $(w_1,w_2,\ldots,w_k)$ is a survival sequence ending in $\Om_k$,
and include $\cc(G-\Om)$ in $\Gg_4$.
Observe that $\Gg_4$ constructed in this manner satisfies 
$$|\Gg_4|\leq \sum_{k=0}^{p-1} 2n^4\cdot |\widehat{\Xx}_k|^3\leq 2n^8\cdot |\Xx|^3,$$
and from the discussion above it follows that $\Gg_4$ contains $D$ provided the situation conforms to this case.

\bigskip

In the remaining case, at least one of the vertices $t_1,t_2$ belongs to $N(D)$. 
Observe that it cannot happen that both of them belong to $N(D)$, because $D$ is also a connected component of $G_k-\Om_k$, and we assumed that $\wD_0$, which is different from $D$,
is the unique connected component of $G_k-\Om_k$ that covers the nonedge $t_1t_2$.
Hence, by symmetry we have one remaining case: $t_1\in \Om\setminus N(D)$ and $t_2\in N(D)$.

\bigskip 

\noindent {\bf{Case 3}}: $t_1\in \Om\setminus N(D)$ and $t_2\in N(D)$.
First, consider the set $A=N_{G_k}(t_1)\cup N_{G_k}(\wD_0)$, and note that $A$ is disjoint from $D$ due to $t_1\notin N(D)$.
Suppose first that $A$ contains the whole set $N(D)\setminus \{w_1,\ldots,w_k\}$. Then it follows that $D$ is one of the connected components of
$G_k-A$, so this suggest the following procedure for obtaining a family $\Gg_5$ encompassing this case.
For each choice of $s\in V(G)$, $0 \leq k < p$ implying $v=w_{k+1}$, $t_1 \in V(G)$, and $\wD_0\in \widehat{\Xx}_k$, compute $A=N_{G_k}(t_1)\cup N_{G_k}(\wD_0)$ and include $\cc(G_k-A)$ in $\Gg_5$.
It follows that 
$$|\Gg_5|\leq \sum_{k=0}^{p-1} n^3\cdot |\widehat{\Xx}_k|\leq n^5\cdot |\Xx|,$$
$\Gg_5$ can be computed in polynomial time, and $\Gg_5$ contains $D$ provided that the situation conforms to this case 
(i.e., $A$ contains $N(D)\setminus \{w_1,\ldots,w_k\}$).
Hence, by including $\Gg_5$ in the constructed family $\Ff_7(\Xx)$ we may proceed with the assumption that there exists some vertex $x\in N(D)\setminus \{w_1,\ldots,w_k\}$ that is contained 
neither in $N_{G_k}(\wD_0)$ nor is a neighbor of $t_1$ in $G_k$.
Since $x\in \Om_k$, there must exist some component $\wD_1\in \cc(G_k-\Om_k)$ such that $\{t_1,x\}\subseteq N_{G_k}(\wD_1)$.
Since $t_1\notin N(D)$ and $x\notin N_{G_k}(\wD_0)$, we have that $\wD_1\neq D$ and $\wD_1\neq \wD_0$.

Suppose for a moment that there was another component $\wD_2\in \cc(G_k-\Om_k)$, different from $D$ and $\wD_0$, such that $t_2\in N_{G_k}(\wD_2)$.
Then, by Lemma~\ref{lem:deduce-pmc-social} we infer that $\Om_k$ can be recovered from $\wD_0$, $\wD_1$, and $\wD_2$ using formula~\eqref{eq:deduced1}.
Since all these sets belong to $\widehat{\Xx}_k$, we infer that if the situation conforms to this case (i.e., there exists such $\wD_2$), then $D$ is already included in the family $\Gg_4$ computed in Case 2.
Hence, from now on we may proceed with the assumption that $D$ and $\wD_0$ are the only components of $G_k-\Om_k$ such that $t_2$ belongs to their neighborhoods in $G_k$.
This assumption yields the following.

\begin{claim}\label{cl:dom-t2}
It holds that $\Om_k\setminus (N(D)\cup N_{G_k}(\wD_0)) \subseteq N_{G_k}[t_2]$.
\end{claim}
\begin{proof}
For the sake of contradiction, suppose there is some $u\in \Om_k\setminus (N(D)\cup N_{G_k}(\wD_0))$ that is non-adjacent to $t_2$.
Then there is also some component $\wD'\in \cc(G_k-\Om_k)$ such that $\{u,t_2\}\subseteq N_{G_k}(\wD')$.
Since $u\in N_{G_k}(\wD')$ and $u\notin N(D)\cup N_{G_k}(\wD_0)$, we have that $\wD'\neq D$ and $\wD'\neq \wD_0$.
This is a contradiction with the assumption that $D$ and $\wD_0$ are the only components of $G_k-\Om_k$ that have $t_2$ in their neighborhoods in $G_k$.
\cqed\end{proof}

Recall that we work under the assumption that $D$ is a mesh.
Since $\Om$ is an $I$-free PMC in $G$, by Observation~\ref{obs:mesh-one-with-I} we infer that in $D$ there is a unique proper strong module $M^p\in \Mod(D)$ such that $I\cap D\subseteq M^p$ and $I\cap M^p$ is nonempty.
Let then $p$ be an arbitrary vertex of $I\cap D=I\cap M^p$.

Let $q\in D$ be an arbitrary vertex belonging to any proper strong module $M^q\in \Mod(D)$ such that $M^p$ and $M^q$ are different.
In particular, $p$ and $q$ are adjacent, since $D$ is a mesh.
Further, let $r\in D$ be an arbitrary neighbor of $t_2$ in $D$; such a vertex exists due to $t_2\in N(D)$.
Consider now a set $Y$ defined as follows:
\begin{equation*}
Y = \Om_k\setminus (N_{G_k}[p,q,r]\cup N_{G_k}[t_2]\cup N_{G_k}(\wD_0)).
\end{equation*}
We now analyze $Y$ through the following two claims.

\begin{claim}\label{cl:in-wD}
It holds that $Y\subseteq N(D)\setminus \{w_1,\ldots,w_k\}$.
\end{claim}
\begin{proof}
Observe that it suffices to show that every vertex $z\in \Om_k\setminus (N_{G_k}(D)\cup N_{G_k}(\wD_0))$ is a neighbor of $t_2$.
This is, however, asserted by Claim~\ref{cl:dom-t2} as $N_{G_k}(D) = N_G(D) \cap V(G_k)$ and $\Om_k \subseteq V(G_k)$. 
\cqed\end{proof}

\begin{claim}\label{cl:sad-face}
It holds that $Y\subseteq \tricky(D,p,q)$, where the tricky vertices are defined w.r.t. the graph $G_k$.
\end{claim}
\begin{proof}
Consider the conclusions of application of the Neighborhood Decomposition Lemma (Lemma~\ref{lem:nei-decomp}) to $D$ in $G_k$.
Since we explicitly excluded $N_{G_k}[p,q]$ from $Y$, no vertex of $Y$ satisfies condition~\eqref{cnd:pq}.
Hence, by Claim~\ref{cl:in-wD} it remains to show that also no vertex of $Y$ satisfies condition~\eqref{cnd:P4}.
For otherwise, there is some $y\in Y$ such that there exists an induced $P_4$ of the form $yDDD$, say $Q$.
Let $D_\Om$ be the unique connected component of $G-N(D)$ that contains $\Om\setminus N(D)$, given by Lemma~\ref{lem:pmc-minseps}.
By Lemma~\ref{lem:pmc-minseps} we have that $y$ is adjacent to $D_\Om$ in $G$, so in fact
$y$ has to be complete to $D_\Om$, because otherwise we could extend $Q$ using two vertices of $D_\Om$ to an induced $P_6$ in $G$.
Since $t_1\in \Om_k\setminus N(D)$ is adjacent to $\wD_0$, we have that $\wD_0\subseteq D_\Om$, so we obtain that $y$ is complete to $\wD_0$.
This means in particular that $y\in N_{G_k}(\wD_0)$, which contradicts the supposition $y\in Y$.
\cqed\end{proof}

%The case when $Y=\emptyset$ will be simple and we will resolve it at the end, so from now on we proceed with the analysis of $Y$ under the assumption that it is non-empty.
Let us define
\begin{equation*}
\begin{split}
Z =& \{ y\in Y\, \colon\, \textrm{there exists some } \wD'\in \cc(G_k-\Om_k)\\
&\textrm{  with $\wD'\neq D$ such that $y\in N_{G_k}(\wD')$}\}.
\end{split}
\end{equation*}
Recall that since $Z\subseteq Y\subseteq \tricky(D,p,q)$ (by Claim~\ref{cl:sad-face}), the neighborhood of every vertex from $Z$ in $D$ is
a collection of proper strong modules from the modular partition $\Mod(D)$, none of which is $M^p$ or $M^q$.
Define the following two quasi-orders $\leq_1$ and $\leq_2$ on $Z$:
\begin{itemize}
\item For $z_1,z_2\in Z$ we put $z_1\leq_1 z_2$ if and only if $N(z_1)\cap D\subseteq N(z_2)\cap D$.
\item For $z_1,z_2\in Z$ we put $z_1\leq_2 z_2$ if and only if the following holds: for every component $\wD'\in \cc(G_k-\Om_k)\setminus \{D\}$, if $z_1\in N_{G_k}(\wD')$ then also $z_2\in N_{G_k}(\wD')$.
\end{itemize}
We now verify that $\leq_1$ and $\leq_2$ satisfy the prerequisites of the Bi-ranking Lemma.
For the proof below,
recall that we are working with the assumption that $D$ and $\wD_0$ are the only components of $G_k-\Om_k$ such that $t_2$ belongs to their neighborhoods in $G_k$, that is,
$$N_{G_k}[t_2] \subseteq \Om_k \cup D \cup \wD_0.$$

\begin{claim}
Every two vertices from $Z$ are either comparable w.r.t. $\leq_1$ or w.r.t. $\leq_2$.
\end{claim}
\begin{proof}
For the sake of contradiction, suppose there are some $z_1,z_2\in Z$ that are neither comparable w.r.t. $\leq_1$ nor w.r.t. $\leq_2$.
This means that there are:
\begin{itemize}
\item modules $M_1,M_2\in \Mod(D)$ such that $M_t\subseteq N(z_t)$ and $M_t\cap N(z_{3-t})=\emptyset$ for $t=1,2$; and
\item components $\wD'_1,\wD'_2\in \cc(G_k-\Om_k)\setminus \{D\}$ such that $z_t\in N_{G_k}(\wD'_t)$ and $z_t\notin N_{G_k}(\wD'_{3-t})$ for $t=1,2$.
\end{itemize}
In particular, $M_1\neq M_2$ and $r \notin M_1 \cup M_2$ as $r$ is anti-adjacent to $Z$. Since $\Mod(D)$ is a clique, $r$ is fully adjacent to $M_1 \cup M_2$.

First, consider the case when $z_1z_2\notin E(G)$. Then pick any vertices $m_1\in M_1$, $m_2\in M_2$, $d_1\in \wD'_1$ that is a neighbor of $z_1$,
and $d_2\in \wD'_2$ that is a neighbor of $z_2$. Since $M_1\neq M_2$ and $D$ is a mesh, it follows that $m_1m_2\in E(G)$.
Then it can be easily seen that $(d_1,z_1,m_1,m_2,z_2,d_2)$ would be an induced $P_6$ in $G$, a contradiction.

Consider now the case when $z_1z_2\in E(G)$.
Pick any $d_1\in \wD'_1$ that is a neighbor of $z_1$. 
Also, pick any $d_0\in \wD_0$ that is a neighbor of $t_2$.
Finally, pick any $m_2\in M_2$.
Consider path $Q$ defined as follows: if $m_2$ is a neighbor of $t_2$ then we put $Q=(d_1,z_1,z_2,m_2,t_2,d_0)$, and otherwise we put $Q=(d_1,z_1,z_2,m_2,r,t_2,d_0)$.
We claim that in either of these cases $Q$ would be an induced $P_6$ or an induced $P_7$ in $G$, which would be a contradiction.

We first observe that since $z_1,z_2\subseteq Z\subseteq Y$ and no vertex of $Y$ is in the neighborhood of $N_{G_k}(\wD_0)$, we have that $\wD'_1\neq \wD_0$ and
$d_0$ is not adjacent to any of the vertices $d_1,z_1,z_2,m_2,r$.
Next, $d_1$ is not adjacent to $z_2$ due to $z_2\notin N_{G_k}(\wD'_1)$, and observe that also $d_1$ is not adjacent to $t_2$, since we assumed that $\wD_0$ is the only
component of $\cc(G_k-\Om_k)\setminus \{D\}$ that contains $t_2$ in its neighborhood. It follows that $d_1$ is not adjacent to any of the vertices $z_2,m_2,r,t_2,d_0$.
Next, $m_2$ is non-adjacent to $z_1$ by the choice of $M_2$. Also $m_2$ is non-adjacent to $t_2$ if we chose to make a detour through $r$, and adjacent to $t_2$ otherwise.
If we indeed made this detour, then observe that $r$ is adjacent to $t_2$ by the choice of $r$, and is not adjacent to $z_1,z_2$, because we excluded $N_{G_k}[p,q,r]$ in the construction of $Y$, and $z_1,z_2\in Y$.
This implies that the neighborhoods of $m_2$ and (if used) $r$ are as required.
Finally, observe that $t_2$ is non-adjacent to $z_1$ and $z_2$, because we excluded $N_{G_k}[t_2]$ in the construction of $Y$.
This concludes the verification that $Q$ would be an induced $P_6$ or an induced $P_7$ in $G$, a contradiction.
\cqed\end{proof}

From the Bi-ranking Lemma (Lemma~\ref{lem:biranking}) we infer that, provided $Z$ is non-empty, there exists a vertex $z_0\in Z$ such that for each $z\in Z$ we have either 
$z_0\leq_1 z$ or $z_0\leq_2 z$.
Assume for a moment that $Z\neq \emptyset$ and thus such $z_0\in Z$ is picked.
Pick any $m_0\in N(z_0)\cap D$ and any component $\wD'_0\in \cc(G_k-\Om_k)\setminus \{D\}$ such that $z_0\in N_{G_k}(\wD'_0)$; such $\wD'_0$ exists due to $z_0\in Z$.
Observe here that $\wD'_0\neq \wD_0$, since no vertex of $Y$ is in $N_{G_k}(\wD_0)$.
Then the conclusion that $z_0\leq_1 z$ or $z_0\leq_2 z$ for all $z\in Z$ implies that
\begin{equation}\label{eq:Zdone}
Z\subseteq N_{G_k}[m_0]\cup N_{G_k}(\wD'_0)\subseteq \Om_k\cup D.
\end{equation}
Let us now define
\begin{equation}\label{eq:defW}
W = Y\setminus (N_{G_k}[m_0]\cup N_{G_k}(\wD'_0)).
\end{equation}
By~\eqref{eq:Zdone}, we have $W\subseteq Y\setminus Z$.
Now define
\begin{equation}\label{eq:defX}
X=N_{G_k}[p,q,r,t_2,m_0]\cup N_{G_k}(\wD_0)\cup N_{G_k}(\wD'_0).
\end{equation}
Observe that by the construction and Claim~\ref{cl:in-wD}, we have that $W=(N(D)\setminus \{w_1,\ldots,w_k\})\setminus X$, and in particular $W$ and $X$ are disjoint.
Turning back to case when $Z=\emptyset$, in this case we simply define
\begin{equation}\label{eq:defXempty}
W = Y \qquad\textrm{and}\qquad X=N_{G_k}[p,q,r,t_2]\cup N_{G_k}(\wD_0).
\end{equation}
Note that then we still have that $W=(N(D)\setminus \{w_1,\ldots,w_k\})\setminus X$.

We now analyze the structure of $X$ through a series of claims. All of the following hold regardless whether $X$ and $W$ are defined using~\eqref{eq:defW} and~\eqref{eq:defX}, or using~\eqref{eq:defXempty}.

\begin{claim}\label{cl:rest-Om}
It holds that $\Om_k\setminus N(D)\subseteq X$ and $D\subseteq X$.
\end{claim}
\begin{proof}
For the first claim, observe that 
elements of $(\Om_k\setminus N(D))\cap N_{G_k}(\wD_0)$ are included in $N_{G_k}(\wD_0)$, while elements of $\Om_k\setminus (N(D)\cup N_{G_k}(\wD_0))$ are included in $N_{G_k}[t_2]$ by Claim~\ref{cl:dom-t2}.
For the second claim, observe that since $D$ is a mesh and $p,q$ belong to different modules of $\Mod(D)$, every vertex of $D$ is adjacent either to $p$ or $q$, and hence it is included in $N_{G_k}[p,q]$.
\cqed\end{proof}

Hence, by Claim~\ref{cl:rest-Om} and the fact that $W=(N(D)\setminus \{w_1,\ldots,w_k\})\setminus X$, we have that every vertex of $\Om_k\cup D$ is either in $X$ or in $W$.
Since $W\subseteq \Om_k$ and $t_2$ has only neighbors in $\Om_k$, $D$, and $\wD_0$, from~\eqref{eq:defX} or~\eqref{eq:defXempty}, respectively,  we immediately obtain the following.

\begin{claim}\label{cl:partition}
It holds that $\Om_k\cup D\subseteq X\cup W\subseteq \Om_k\cup D\cup \wD_0$.
\end{claim}

This yields the following structural claim about $W$.

\begin{claim}\label{cl:Wunion}
The graph $G_k[W]$ is the union of some collection of connected components of the graph $G_k-X$.
\end{claim}
\begin{proof}
It suffices to show that for every vertex $w\in W$,
every neighbor of $w$ in $G_k$ belongs either to $X$ or to $W$. Observe that since $w\in W\subseteq Y\setminus Z$, the only component of $G_k-\Om_k$ in which $w$ may have neighbors is $D$.
Hence $N_{G_k}(w)\subseteq \Om_k\cup D$, however Claim~\ref{cl:partition} asserts that $\Om_k\cup D\subseteq W\cup X$.
\cqed\end{proof}

From now on we assume that $W$ is non-empty, since the other case will be easy and we will resolve it at the end.
Construct a set $J$ by including one, arbitrarily chosen, vertex from each connected component of $G_k[W]$.
Clearly, $J$ is a non-empty independent set in $G$ by construction.
Since $D$ is assumed to be a mesh, whereas $\Om$ is $I$-free, we may apply Lemma~\ref{lem:cover-tricky} to $D$ and $J$ in $G$ and conclude that
there exist a vertex $w\in D$ and a component $D'\in \cc(G-\Om)$ with $D\neq D'$ such that $J\subseteq N(w)\cup N(D')$.
The following claim gives us a way to recognize those components of $G_k-X$ that are contained in $W$.

\begin{claim}\label{cl:recognize-W}
The set $W$ is the union of vertex sets of those connected components of $G_k-X$ that are not disjoint with $N(w)\cup N(D')$.
\end{claim}
\begin{proof}
Let $C$ be any connected component of $G_k-X$; by Claim~\ref{cl:Wunion} we have that either $C\subseteq W$ or $C\cap W=\emptyset$.
If $C\subseteq W$, then $C$ contains some vertex of $J$ by the definition of $J$, which is contained in $N(w)\cup N(D')$ by the choice of $w$ and $D'$.
On the other hand, suppose that $C\cap W=\emptyset$. 
Since $C\cap X=\emptyset$ by definition, and by Claim~\ref{cl:partition} we have $\Om_k\cup D\subseteq X\cup W$, we infer
that $C$ is entirely contained in some connected component of $G_k-(\Om\cup D)$. However, $N(w)\cup N(D')\subseteq D\cup \Om$, so $C$ contains no vertex of $N(w)\cup N(D')$.
\cqed\end{proof}

Claim~\ref{cl:recognize-W} suggests the following procedure for finding candidates for $D$ conforming to this case.
We shall compute an auxiliary family $\Yy$ defined as follows.
First, iterate through all choices of $s\in V(G)$, $0 \leq k < p$ implying $v=w_{k+1}$, $p,q,r,t_2\in V(G)$, and $\wD_0\in \widehat{\Xx}_k$.
For each choice, construct a collection of candidates for $X$:
one using formula~\eqref{eq:defXempty}, and one for each choice of $m_0\in V(G)$ and $\wD'_0\in \widehat{\Xx}_k$ using formula~\eqref{eq:defX}.
For each candidate $X$, include the set
\begin{equation*}
X\setminus \wD_0
\end{equation*}
in $\Yy$. Observe that this set is equal to $\Om_k\cup D$ in case $W=\emptyset$ by Claim~\ref{cl:partition}.
Finally, for each candidate $X$, 
iterate through all choices of $w\in V(G)$ and $D'\in \Xx$, and define $W$ to be the union of those connected components of $G_k-X$ that contain some vertices of $N(w)\cup N(D')$.
Include the set
\begin{equation*}
(X\cup W)\setminus \wD_0
\end{equation*}
in $\Yy$ and observe that this set is equal to $\Om_k\cup D$ in the remaining case, due to Claims~\ref{cl:recognize-W} and~\ref{cl:partition}.
Apply the algorithm of Lemma~\ref{lem:recover2} to $\Yy$ in the $P_6$-free graph $G_k$, thus obtaining a family $\Ff_{\rec,2}(\Yy)$.
Finally, construct a family $\Gg_6$ by including, for each $\Om_k\in \Ff_{\rec,2}(\Yy)$, the whole family $\cc(G_k-\Om_k)$ in $\Gg_6$.
It follows that 
\begin{align*}
|\Gg_6|\leq& n\cdot |\Ff_{\rec,2}(\Yy)|\leq n^2\cdot |\Yy|\\
\leq& n^2\cdot \sum_{k=0}^{p-1} 2n^7\cdot |\widehat{\Xx}_k|^2\cdot |\Xx|\leq 2n^{12}\cdot |\Xx|^3,
\end{align*}
%TODO-maybe n^2\cdot 2n^7\cdot |\widehat{\Xx}|^2\cdot |\Xx|\leq 2n^{11}\cdot |\Xx|^3
and $\Gg_6$ contains $D$ provided $t_1,t_2$ conform to this case.

%Compute the set
%\begin{equation*}
%(N_{\wG}[p,q,r]\cup N_{\wG}[t_2]\cup N_{\wG}(\wD_0))\setminus \wD_0
%\end{equation*}
%and include it in $\Yy$. Observe that this set is equal to $\wOm\cup D$ in case $Y=\emptyset$.

\bigskip

To wrap up, we output the family
\begin{equation*}
\Ff_7(\Xx)=\Xx\cup \Gg_1\cup \Gg_2\cup \Gg_3\cup \Gg_4\cup \Gg_5 \cup \Gg_6.
\end{equation*}
From the obtained bounds on the sizes of families $\Gg_t$ it follows that $|\Ff_7(\Xx)|\leq 11n^{12}\cdot |\Xx|^3$. 
%TODO-maybe 10n^{11}\cdot |\Xx|^3
Moreover, due to inclusion of $\Xx$ we have that $\cc(G-\Om)\setminus \{D\}\subseteq \Ff_7(\Xx)$, while the reasoning above also shows that $D\in \Ff_7(\Xx)$.
\end{proof}

Next, we show an auxiliary tool: for any two mesh components we can either recognize any of them, or their union with the PMC.
Note that this lemma does not assume any bounds on the number of neighbor-maximal components of $\Om$, a feature that we will use in the next section.

\begin{lemma}\label{lem:Omplus}
One can in polynomial time compute a family $\Ff_8$ of size at most $5n^9$ such that the following holds. 
Take any $I$-free PMC $\Om$ and suppose there are different components $D_1,D_2\in \cc(G-\Om)$ such that both $D_1$ and $D_2$ are meshes.
Then either $\Ff_8$ contains $D_1$, or $\Ff_8$ contains $D_2$, or $\Ff_8$ contains $\Om\cup D_1\cup D_2$.
\end{lemma}
\begin{proof}
Observe that if $N(D_1)\cup N(D_2)\subsetneq \Om$, then either $D_1$ or $D_2$ is included in the family $\Ff_1$, given by Lemma~\ref{lem:two-not-whole}.
Hence, by including $\Ff_1$ in the constructed family $\Ff_8$ we may proceed with the assumption that $N(D_1)\cup N(D_2)=\Om$.
Recall that $|\Ff_1|\leq n^8$.

Our goal is to try to identify the set $\Om\cup D_1\cup D_2$ using the general version of the Separator Covering Lemma (Lemma~\ref{lem:covering-general}).
By Observation~\ref{obs:mesh-one-with-I}, in each $D_i$, for $i=1,2$, there exists a unique proper strong module $M^p_i\in \Mod(D_i)$ such that $I\cap D_i\subseteq M^p_i$ and $I\cap D_i=I\cap M^p_i$ is non-empty.
Pick any $p_1\in I\cap M^p_1$ and $p_2\in I\cap M^p_2$.
Let $H=G[D_1\cup D_2\cup (N(D_1)\cap N(D_2))]$, and observe that $N(D_1)\cap N(D_2)$ is a minimal separator in $H$, with $D_1$ and $D_2$ being full sides.
By applying the Separator Covering Lemma to $H$ we may obtain vertices $q_i,r_i$, for $i=1,2$, 
such that $\{p_i,q_i,r_i\}\subseteq D_i$ and $N(D_1)\cap N(D_2)\subseteq N[p_1,q_1,r_1,p_2,q_2,r_2]$.
Let $M^p_i$, $M^q_i$, $M^r_i$ be the modules of $\Mod(D_i)$ that contain $p_i$, $q_i$, and $r_i$, respectively; note that in the Separator Covering Lemma we choose
$q_i$ so that $M^p_i\neq M^q_i$.
Since $D_i$ is a mesh, for $i=1,2$, we infer that every vertex of $D_i$ is adjacent either to $p_i$ or to $q_i$.
Consequently, if we define
\begin{equation}\label{eq:defXZ}
X=N[p_1,q_1,r_1,p_2,q_2,r_2]\qquad\textrm{and}\qquad Z=\Om\setminus X,
\end{equation}
then we have
\begin{equation*}
\begin{split}
&D_1\cup D_2\cup(N(D_1)\cap N(D_2)) \subseteq X\\
&\subseteq D_1\cup D_2\cup N(D_1)\cup N(D_2)
=D_1\cup D_2\cup \Om.
\end{split}
\end{equation*}
Assume for a moment that $Z$ is empty.
Then we have that $X=D_1\cup D_2\cup \Om$.
In the constructed family $\Ff_8$ we will include the family $\Gg_1$ defined as follows: for each choice of $p_1, q_1, r_1, p_2, q_2, r_2\in V(G)$, include 
$N[p_1,q_1,r_1,p_2,q_2,r_2]$ in $\Gg_1$; 
note that $|\Gg_1|\leq n^6$.
This ensures that in the case when $Z$ is empty, the set $D_1\cup D_2\cup \Om$ is included in $\Ff_8$.
Hence, we may proceed with the assumption that $Z$ is non-empty.

Since $N(D_1) \cup N(D_2) = \Om$, we have that $N(D_1)\setminus N(D_2)$ and $N(D_2)\setminus N(D_1)$ are non-empty.
By Lemma~\ref{lem:wings}, we have that either $N(D_1)\setminus N(D_2)$ is complete to $D_1$, or $N(D_2)\setminus N(D_1)$ is complete to $D_2$.
By symmetry, without loss of generality assume the former. This implies that $N(D_1)\setminus N(D_2)\subseteq X$, by the definition of $X$, and hence
\begin{equation*}
D_1\cup D_2\cup N(D_1)\subseteq X\subseteq D_1\cup D_2\cup \Om.
\end{equation*}
Equivalently, $Z\subseteq N(D_2)\setminus N(D_1)$. 
Note that in particular no vertex of $D_1$ has any neighbor outside of $X$.
We now take a closer look at $Z$.

\begin{claim}\label{cl:Ztricky}
It holds that $Z\subseteq \tricky(D_2,p_2,q_2)$.
\end{claim}
\begin{proof}
Take any $z\in Z$; as we have already argued, $z\in N(D_2)\setminus N(D_1)$.
By the Neighborhood Decomposition Lemma (Lemma~\ref{lem:nei-decomp}) we have that $z$ satisfies one of the conditions: either $z\in N[p_2,q_2]$, or there exists an induced $P_4$ of the form $zD_2D_2D_2$,
or $z\in \tricky(D_2,p_2,q_2)$. Note that neighbors of $p_2$ and $q_2$ are excluded in the construction of $Z$, hence the first alternative cannot happen. Therefore, it remains to exclude the second alternative.

For the sake of contradiction, suppose there is an induced $P_4$, say $Q$, of the form $zD_2D_2D_2$. 
Let $D_\Om$ be the unique connected component of $G-N(D_2)$ that contains $\Om\setminus N(D_2)$, which exists by Lemma~\ref{lem:pmc-minseps}.
Recall that $D_\Om$ is full to $N(D_2)$, so in particular $z$ has neighbors in $D_\Om$.
Since $N(D_1)\setminus N(D_2)$ is non-empty and complete to $D_1$, we have that $D_1\subseteq D_\Om$. However, $z$ has no neighbors in $D_1$ due to $Z\subseteq N(D_2)\setminus N(D_1)$, so we conclude that
$z$ is adjacent, but not complete to $D_\Om$. It follows that $Q$ can be extended by two vertices from $D_\Om$ to an induced $P_6$ in $G$, a contradiction.
\cqed\end{proof}

\begin{claim}\label{cl:Zcomplete}
It holds that $Z$ is complete to $N(D_1)\setminus N(D_2)$.
\end{claim}
\begin{proof}
Take any $z\in Z$.
By Claim~\ref{cl:Ztricky} we have $z\in \tricky(D_2,p_2,q_2)$, so in particular $z$ is not complete to $D_2$.
Then from the second claim of Lemma~\ref{lem:wings} it follows that $z$ has to be complete to $N(D_1)\setminus N(D_2)$. 
\cqed\end{proof}

Let us now define
\begin{equation*}
Y = N(D_1)\setminus (N(V(G)\setminus X)\cup N[p_2,q_2,r_2]).
\end{equation*}
Observe that since $Z\subseteq V(G)\setminus X$ and $Z$ is complete to $N(D_1)\setminus N(D_2)$ by Claim~\ref{cl:Zcomplete}, we have that $Y\subseteq N(D_1)\cap N(D_2)$.
The next claim is crucial: the neighborhoods of vertices of $Y$ and of $Z$ within $D_2$ are always comparable.

%Since no vertex of $Y$ is a neighbor of $p_2$ or $q_2$, 
%by applying the Neighborhood Decomposition Lemma (Lemma~\ref{lem:nei-decomp}) to the elements of $Y$ we infer that for every vertex $y\in Y$ either there is an induced $P_4$, say $Q_y$, of
%the form $yD_2D_2D_2$, or $y\in \tricky(D_2,p_2,q_2)$.

\begin{claim}\label{cl:comparable-nei}
For each $y\in Y$ and each $z\in Z$, either $N(y)\cap D_2\subseteq N(z)\cap D_2$ or $N(y)\cap D_2\supseteq N(z)\cap D_2$.
\end{claim}
\begin{proof}
Observe first that since $z\in Z\subseteq V(G)\setminus X$ and in the construction of $Y$ we excluded all vertices with neighbors in $V(G)\setminus X$, it follows that $y$ and $z$ are non-adjacent.

For the sake of contradiction, suppose there are some $m_y\in (N(y)\setminus N(z))\cap D_2$ and $m_z\in (N(z)\setminus N(y))\cap D_2$.
Since $z\in Z\subseteq \tricky(D_2,p_2,q_2)$, the neighborhood of $z$ within $D_2$ is the union of some collection of modules from $\Mod(D_2)$.
Since $m_z\in N(z)$ and $m_y\notin N(z)$, it follows that $m_z$ must be within some module from this collection, while the module in which $m_y$ resides is outside of this collection.
In particular, $m_y$ and $m_z$ reside in different modules of $\Mod(D_2)$, so since $D_2$ is a mesh, we conclude that $m_y$ and $m_z$ are adjacent.

Since $z\in \tricky(D_2,p_2,q_2)$, we have that $z$ has no neighbors in $M^p_2$, so in particular no neighbors in $I\cap D_2$.
Since $\Om$ is $I$-free, by the maximality of $I$ we infer that $z$ has to have a neighbor in $I$, which must reside in some component $D'\in \cc(G-\Om)$ other than $D_2$.
In particular, $D'\neq D_1$ since $z\notin N(D_1)$.
Let $u$ be any neighbor of $z$ in $D'$. Since $X\subseteq D_1\cup D_2\cup \Om$, we have that $D'\subseteq V(G)\setminus X$.
As $y\in Y$ and we excluded neighbors of $V(G)\setminus X$ in the construction of $Y$, we conclude that $u$ is not adjacent to $y$.
Let now $v$ be any neighbor of $y$ in $D_1$. Then from all the assertions presented above it follows that $(u,z,m_z,m_y,y,v)$ is an induced $P_6$ in $G$, a contradiction.
\cqed\end{proof}

Assume for a moment that $Y$ is empty.
Denote $$W=N(V(G)\setminus X)\cup N[p_2, q_2, r_2].$$
Then $Y=\emptyset$ implies that $W$ contains $N(D_1)$.
On the other hand, since $X\supseteq D_1\cup N(D_1)$, we have that $W$ is disjoint from $D_1$, so $D_1$ is among the connected components
of $G-W$.
Therefore, we construct a family $\Gg_2$ as follows: for each choice of $p_1, q_1, r_1, p_2, q_2, r_2\in V(G)$, 
define $X=N[p_1,q_1,r_1,p_2,q_2,r_2]$ and $W=N(V(G)\setminus X)\cup N[p_2,q_2,r_2]$, and include $\cc(G-W)$
in $\Gg_2$. Observe that thus $|\Gg_2|\leq n^7$ and $\Gg_2$ contains $D_1$ in case $Y$ is empty.
Hence, from now on we may proceed with the assumption that $Y$ is non-empty.

Let us pick $y_0\in Y$ that has an inclusion-wise minimal neighborhood within $D_2$, i.e., $N(y_0)\cap D_2$ is inclusion-wise minimal among elements of $Y$.
Further, let $s$ be neighbor of $y_0$ in $D_2$.
Let us define 
\begin{equation}\label{eq:defXpZp}
X'=N[p_1, q_1, r_1, p_2, q_2, r_2, s]\qquad\textrm{and}\qquad Z'=\Om\setminus X'.
\end{equation}
By definition we have $X'\supseteq X$ and $Z'\subseteq Z$, so in particular $Z'\subseteq N(D_2)\setminus N(D_1)$ and $Z'\subseteq \tricky(D_2,p_2,q_2)$, by Claim~\ref{cl:Ztricky}.

Observe that if we have $Z'=\emptyset$ then $X'=D_1\cup D_2\cup \Om$, similarly as for $Z$.
Therefore, by including into $\Ff_8$ the family $\Gg_3$ constructed by taking $N[p_1,q_1,r_1,p_2,q_2,r_2,s]$
for each choice of $p_1, q_1, r_1, p_2, q_2, r_2, s\in V(G)$, we are certain that $\Ff_8$ includes $D_1\cup D_2\cup \Om$ in case $Z'$ is empty.
Note that $|\Gg_3|\leq n^7$. Hence, we may proceed with the assumption that $Z'$ is non-empty.

Pick any $z_0\in Z'$ and let $t$ be any neighbor of $z_0$ in $D_2$.
Since $z_0\in Z\subseteq \tricky(D_2,p,q)$, the neighborhood of $z_0$ within $D_2$ is a collection of proper strong modules from $\Mod(D_2)$.
Clearly, $t$ is within one of these modules, whereas $s$ is not, since $s$ is not a neighbor of $z_0$ by the definition of $Z'$.
We infer that $s$ and $t$ are contained in different modules of $\Mod(D_2)$, so as $D_2$ is a mesh, $s$ and $t$ are adjacent.

Define
\begin{equation}\label{eq:defWp}
W'=N(V(G)\setminus X)\cup N[p_2,q_2,r_2,s,t].
\end{equation}
The next claim shows that $W'$ is sufficient to recognize $D_1$.

\begin{claim}\label{cl:Wp}
It holds that $W'\supseteq N(D_1)$ and $W'\cap D_1=\emptyset$.
\end{claim}
\begin{proof}
Observe that $W'=W\cup N[s,t]$.
Since $s,t\in D_2$, from $W\cap D_1=\emptyset$ (which follows from $D_1\cup N(D_1)\subseteq X$) we infer that $W'\cap D_1=\emptyset$ as well.
Moreover, since $Y=N(D_1)\setminus W$, to prove that $W'\supseteq N(D_1)$ it suffices to show that $Y\subseteq W'$.

Recall that by Claim~\ref{cl:comparable-nei}, for each $z'\in Z'$ we have either $N(z')\cap D_2\subseteq N(y_0)\cap D_2$ or $N(z')\cap D_2\supseteq N(y_0)\cap D_2$.
However, we have that $s$ is a neighbor of $y_0$, whereas no vertex of $Z'$ is adjacent to $s$.
We conclude that for each $z'\in Z'$ we have $N(z')\cap D_2\nsupseteq N(y_0)\cap D_2$, which implies that 
\begin{equation}\label{eq:zpy0}
N(z')\cap D_2\subsetneq N(y_0)\cap D_2.
\end{equation}
Take now any $y\in Y$. We claim that $y\in N[t]$, which would imply that $y\in W'$.
Assume otherwise: $y$ is non-adjacent to $t$.
Since $z_0$ is adjacent to $t$ by definition, we infer that $N(y)\cap D_2\nsupseteq N(z_0)\cap D_2$.
By Claim~\ref{cl:comparable-nei} again, we infer that
\begin{equation}\label{eq:yz0}
N(y)\cap D_2\subsetneq N(z_0)\cap D_2.
\end{equation}
By combining~\eqref{eq:zpy0} for $z'=z_0$ and~\eqref{eq:yz0}, we conclude that
\begin{equation*}
N(y)\cap D_2\subsetneq N(y_0)\cap D_2.
\end{equation*}
This is a contradiction with the choice of $y_0$.
\cqed\end{proof}

Thus, Claim~\ref{cl:Wp} suggests the following procedure for constructing a family $\Gg_4$ encompassing the remaining case.
For each choice of $p_1,q_1,r_1,p_2,q_2,r_2,s,t\in V(G)$,
define $X$ and $W'$ as above, and include the whole family $\cc(G-W')$ in $\Gg_4$.
Observe that if the situation conforms to the current case (i.e., $Z'\neq \emptyset$), then $D_1$ is included in $\Gg_4$ constructed in this manner, by Claim~\ref{cl:Wp}.
Moreover, $|\Gg_4|\leq n^9$.

It now remains to output the family
$$\Ff_8=\Ff_1\cup \Gg_1\cup \Gg_2\cup \Gg_3\cup \Gg_4,$$
and observe that the upper bounds presented above imply that $|\Ff_8|\leq 5n^9$.
\end{proof}

We are ready to finish the proof of Lemma~\ref{lem:summary}.

\begin{proof}[Proof of Lemma~\ref{lem:summary}]
First, let $\Ff_2$ be the family given by Lemma~\ref{lem:one-in-three-see}. 
Recall that $|\Ff_2|\leq 2n^9$ and $\Ff_2$ contains all but at most two connected components of $G-\Om$.
In particular, since $\Om$ has at least three neighbor-maximal components of $G-\Om$, we have that $\Ff_2$ contains at least one such component.

Construct the family
$$\Gg=\Ff_2\cup \Ff_3(\Ff_2)\cup \Ff_8,$$
where $\Ff_3(\Ff_2)$ is the family given by Lemma~\ref{lem:non-clique-see} for $\Ff_2$, whereas $\Ff_8$ is the family given by Lemma~\ref{lem:Omplus}.
Recall that $|\Ff_3(\Ff_2)|\leq n+2n^4\cdot |\Ff_2|\leq 3n^{13}$ and $|\Ff_8|\leq 5n^9$, thus $|\Gg|\leq 10n^{13}$.
Observe that if we have that some $D\in \cc(G-\Om)$ is not a mesh, then either $\Ff_2$ contains $D$,
or $\Ff_2$ contains some neighbor-maximal component $D_0$ different from $D$.
Consequently, in the latter case we conclude by Lemma~\ref{lem:non-clique-see} that $D\in \Ff_3(\Ff_2)$.
In summary, $\Gg$ contains every component $D\in \cc(G-\Om)$ that is not a mesh.

Since $\Gg\supseteq \Ff_2$, we also have that $\Gg$ contains all but at most two connected components of $G-\Om$.
Suppose first that $\Gg$ contains all connected components of $G-\Om$. 
Then the family $\Ff_{\rec,1}(\Gg)$, given by Lemma~\ref{lem:recover1} for $\Gg$, contains $\Om$.
Hence, by including $\Ff_{\rec,1}(\Gg)$ in $\Ff^1_9$ we cover this case.
Note that 
$$|\Ff_{\rec,1}(\Gg)|\leq 3n^6\cdot|\Gg|^3\leq 3\cdot 10^3\cdot n^{45}.$$

Next, suppose that $\Gg$ contains all but one connected component of $G-\Om$, say $D$.
Then the family $\Ff_7(\Gg)$ given by Lemma~\ref{lem:monster} for $\Gg$ contains all the components of $G-\Om$.
Consequently, the family $\Ff_{\rec,1}(\Ff_7(\Gg))$, given by Lemma~\ref{lem:recover1} for $\Ff_7(\Gg)$, contains $\Om$.
Hence, by including $\Ff_{\rec,1}(\Ff_7(\Gg))$ in $\Ff^1_9$ we cover this case as well.
Note that 
\begin{align*}
|\Ff_{\rec,1}(\Ff_7(\Gg))|\leq& 3n^6\cdot |\Ff_7(\Gg)|^3\\
\leq& 3n^6\cdot 11^3n^{36}\cdot |\Gg|^9\\
\leq& 5 \cdot 10^{3+9}\cdot n^{6+36+9\cdot 13}\\
\leq& 5 \cdot 10^{12}\cdot n^{159}.
\end{align*}
%TODO-maybe $$|\Ff_{\rec,1}(\Ff_7(\Gg))|\leq 3n^6\cdot |\Ff_7(\Gg)|^3\leq 3n^6\cdot 10^3n^{33}\cdot |\Gg|^9\leq 3\cdot 10^{3+9}\cdot n^{6+33+9\cdot 13}\leq 3\cdot 10^{12}\cdot n^{156}.$$

Finally, suppose that $\Gg$ contains all but two connected components $D_1,D_2$ of $G-\Om$.
Observe that $D_1$ and $D_2$ have to be meshes, since otherwise they would be already included in $\Gg$.
Since $\Gg$ contains some neighbor-maximal component $D_0$ other then $D_1$ or $D_2$,
the family $\Ff_6(\Gg)$ given by Lemma~\ref{lem:mesh-see-sticking} contains some sets $D_1^+$ and $D_2^+$ that are fuzzy versions of $D_1$ and $D_2$, respectively.
Note here that $|\Ff_6(\Gg)|\leq n^3\cdot |\Gg|\leq 10n^{16}$.

Since we included $\Ff_8$ in family $\Gg$, and both $D_1$ and $D_2$ are meshes, by Lemma~\ref{lem:Omplus} we infer that $\Gg$ contains either $D_1$, or $D_2$, or $\Om\cup D_1\cup D_2$.
However, we assumed that neither $D_1$ nor $D_2$ is included in $\Gg$, so $\Om\cup D_1\cup D_2\in \Gg$.
It then follows that we may take $\Ff^2_9=\Gg\times \Ff_6(\Gg)\times \Ff_6(\Gg)$; note here that $|\Ff^2_9|\leq 10^3\cdot n^{42}\leq 10^{13}\cdot n^{159}$. %TODO-maybe 10^{13}\cdot n^{156}
On the other hand, to cover the previous cases we take $\Ff^1_9=\Ff_{\rec,1}(\Gg)\cup \Ff_{\rec,1}(\Ff_7(\Gg))$, which is a set of size at most $10^{13}\cdot n^{159}$.%TODO-maybe 10^{13}\cdot n^{156}
\end{proof}

\section{Capturing by modifying}\label{sec:modifying}
In the previous section we have shown how to enumerate a polynomial-size family of PMCs in the graph which encompasses many cases of how a PMC used in an $I$-free minimal chordal completion can look like.
However, not all PMCs that are potentially necessary have been enumerated so far.
In this section our goal is to augment this enumeration by some additional PMCs with the following property: there {\em{exists}} some $I$-free minimal chordal completion whose maximal cliques
are among the enumerated PMCs. We stress that now we do not aim at recognizing all PMCs, but rather we would like to recognize a rich enough family of PMCs so that any chordal completion
can be modified to a one where maximal cliques are among the enumerated ones. We first need to understand formally how we are going to modify a chordal completion.

\newcommand{\Gm}{\Gamma}
\newcommand{\inte}{\mathsf{int}}
\newcommand{\torso}{\mathsf{torso}}
\newcommand{\cl}{\mathsf{cl}}

\subsection{Potential segments and modifying completions}

We start with the definition of a potential segment, which is essentially a portion of the graph that may correspond to a subtree in a clique tree of a minimally completed chordal graph, and related notions.
Every PMC is a (trivial) potential segment, but we will only be able to recognize larger potential segments and complete them arbitrarily. See also Figure~\ref{fig:segment1}.

\begin{figure}[tb]
\begin{center}
\includegraphics{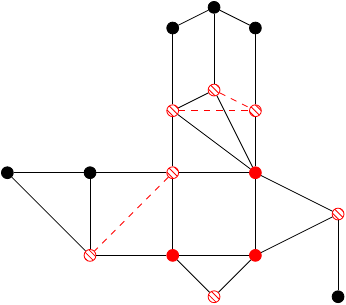}
\end{center}
\caption{An example of a potential segment of the form $N[A]$ for connected set $A$.
The vertices of $A$ are solid red, the vertices of $N(A)$ are stripped red.
The dashed red lines are the closure edges of the torso of $G[A]$.
Note that the interior of this segment consists not only of $A$, but also of the 
vertex of $N(A)$ at the bottom.}\label{fig:segment1}
\end{figure}

\begin{definition}
Let $G$ be a graph. A vertex subset $\Gm\subseteq V(G)$ is called a {\em{potential segment}} if for every $D\in \cc(G-\Gm)$ we have that $N(D)$ is a minimal separator in $G$.
The {\em{interior}} of $\Gm$, denoted $\inte(\Gm)$ is the set $V(G)\setminus N[V(G)\setminus \Gm]$, that is, $\inte(\Gm)$ comprises vertices of $\Gm$ that do not have neighbors outside of $\Gm$.
The {\em{torso}} of $\Gm$, denoted $\torso(\Gm)$ is the subgraph obtained from $G[\Gm]$ by turning $N(D)$ into a clique for each $D\in \cc(G-\Gm)$.
The set of {\em{closure edges}} of $\Gm$, denoted $\cl(\Gm)$, comprises the edges added in the construction of the torso, i.e., $\cl(\Gm)=E(\torso(\Gm))\setminus E(G[\Gm])$.
A chordal completion $F$ of $G$ {\em{respects}} $\Gm$ if for each $D\in \cc(G-\Gm)$, every edge of $F$ that has one endpoint in $D$, has the second endpoint in $N_G[D]$.
\end{definition}

\begin{lemma}\label{lem:segment-NA}
If $A \subseteq V(G)$ is a nonempty set of vertices such that $G[A]$ is connected, then $N[A]$
is a potential segment in $G$ that contains $A$ in its interior.
\end{lemma}
\begin{proof}
Consider a connected component $D \in \cc(G-N[A])$ and let $S = N(D)$.
Clearly, $D \in \cc(G-S)$ and $D$ is full to $S$. Furthermore, since $G[A]$ is connected
and $N(D) \cap A = \emptyset$, there exists a connected component $D_A \in \cc(G-S)$ that contains
$A$. Since $S \subseteq N[A]$, we have that $D_A$ is full to $S$, concluding the proof that
$S$ is a minimal separator in $G$.
As the choice of $D$ is arbitrary, we have that $N[A]$ is a potential segment; the fact that it contains $A$ in its interior is straightforward.
\end{proof}
For a set $A$ as in Lemma~\ref{lem:segment-NA} (i.e., nonempty set of vertices inducing
a connected subgraph of $G$), we denote $\partial A = N(V(G) \setminus N[A])$.

We now explore the connection between minimal chordal completions of $G$ that respect a segment $\Gm$ and minimal chordal completions of the torso of $\Gm$.
Intuitively, $\Gm$ serves as a separate piece where the chordal completion can be chosen independently of the rest of the graph, provided we look only at completions that respect $\Gm$.
The first check is very easy: any minimal completion that respects $\Gm$ has to turn the subgraph induced by $\Gm$ into a supergraph of its torso, and neighborhoods of components stay minimal separators.

\begin{lemma}\label{lem:segments-basics}
Suppose $G$ is a graph, $\Gm$ is a potential segment in $G$, and $F$ is a chordal completion of $G$ that respects $\Gm$. 
Denote $H=G+F$. Then the following assertions hold:
\begin{itemize}
\item The family $\cc(G-\Gm)$, treated as a family of vertex subsets, is equal to $\cc(H-\Gm)$.
\item For each $D\in \cc(G-\Gm)$, $N(D)$ is still a minimal separator in $H$ with $D$ being still a component of $H-N(D)$ full to $N(D)$.
\item It holds that $F\supseteq \cl(\Gm)$, that is, every closure edge of $\Gm$ has to be included in $F$.
\end{itemize}
\end{lemma}
\begin{proof}
The first assertion follows immediately from the assumption that $F$ respects $\Gm$.
For the second assertion, observe that $D$ remains a connected component of $H-N(D)$ due to $F$ respecting $\Gm$, and of course it remains full to $N(D)$.
Let $D'$ be some other component of $G-N(D)$ that is full to $N(D)$ in $G$, which exists due to $N(D)$ being a minimal separator in $G$.
Then $D'$ is contained in some component $D''$ of $H-N(D)$, which must be different from $D$.
Then $D''$ is also full to $N(D)$ in $H$, so $N(D)$ remains a minimal separator in $H$.
For the third assertion, as $H$ is chordal and $N(D)$ is a minimal separator in $H$, 
for each $D\in \cc(G-\Gm)$ we have that $N(D)$ is a clique in $H$ from Lemma~\ref{lem:chordal-minsep}.
It follows that $H$ is a supergraph of $\torso(\Gm)$ and $F\supseteq \cl(\Gm)$.
\end{proof}

We now formalize the intuition that a potential segment behaves independently w.r.t. completions respecting it.

\begin{figure}[tb]
\begin{center}
\includegraphics{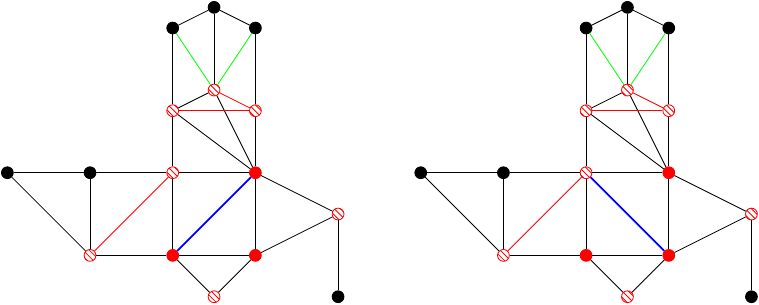}
\end{center}
\caption{Minimal completions inside segments can be chosen independently of the rest of the graph,
  as long as the closure edges of the torso are kept in the completion. 
  Here, the minimal completion consists of colorful edges (green, red, blue), and the blue completion
  of the segment can be chosen independently of the green completion outside the segment, 
  while keeping the red closure edges of the segment.}\label{fig:segment2}
\end{figure}

\begin{definition}\label{def:subst}
Let $G$ be a graph, $\Gm$ be a potential segment in $G$, and $F$ be a chordal completion of $G$ that respects $\Gm$.
Further, let $F_\Gm$ be a chordal completion of $\torso(\Gm)$.
Define {\em{$F$ with $\Gm$ filled using $F_\Gm$}}, denoted $F[\Gm\to F_\Gm]$, to be the set constructed from $F$ as follows: remove from $F$ all the edges of $(F\cap \binom{\Gm}{2})\setminus \cl(\Gm)$
and add all the edges of $F_\Gm$.
\end{definition}

\begin{lemma}\label{lem:subst}
In the setting of Definition~\ref{def:subst}, $F[\Gm\to F_\Gm]$ is also a chordal completion of $G$. 
Moreover, if $F$ is a minimal chordal completion of $G$ and $F_\Gm$ is a minimal chordal completion of $\torso(\Gm)$, then $F[\Gm\to F_\Gm]$ is a minimal chordal completion of $G$.
\end{lemma}
\begin{proof}
Denote $H=G+F$, $F'=F[\Gm\to F_\Gm]$, and $H'=G+F'$. We first show that $F'$ is a chordal completion, that is, $H'$ is chordal.
Suppose, for the sake of contradiction, that $H'$ contains some induced cycle $C$ of length at least $4$.
Observe that $C$ cannot be entirely contained in $H'[\Gm]$, since $H'[\Gm]$ is equal to $\torso(\Gm)+F_\Gm$, which is chordal.
Hence, $C$ contains a vertex $u$ of some connected component $D\in \cc(G-\Om)$. Observe that the subgraphs induced by $N[D]$ in $H$ and in $H'$ are equal, and $H[N[D]]$ is chordal because $H$ is chordal.
Hence, $C$ cannot be entirely contained in $N[D]$, so it contains some vertex $v$ from $V(G)\setminus N[D]$. From the construction of $H'$ it follows that the two paths contained in $C$ that connect $u$ and $v$
have to pass through $N(D)$, and hence $C$ has a pair of non-consecutive vertices contained in $N(D)$. However, $N(D)$ is a clique in $H'$, a contradiction to the assumption that $C$ is an induced cycle in $H'$.

Next, we prove that $F'$ is minimal provided $F$ and $F_\Gm$ are minimal as in the lemma statement.
Take any $F''\subseteq F'$ that is also a chordal completion.
By the construction we have that $F'$ respects $\Gm$, hence so does $F''$ as well.
By Lemma~\ref{lem:segments-basics}, we have that both $F'$ and $F''$ contain the set of closure edges $\cl(\Gm)$.

Since $F''\subseteq F'$, we have that $\cl(\Gm)\subseteq F''\cap \binom{\Gm}{2}\subseteq \cl(\Gm)\cup F_\Gm$. 
Hence $F''\setminus \cl(\Gm)$ is a chordal completion of $\torso(\Gm)$ that is a subset of $F_\Gm$. 
By the minimality of $F_\Gm$ we conclude that $F''\cap \binom{\Gm}{2}=\cl(\Gm)\cup F_\Gm=F'\cap \binom{\Gm}{2}$.

Denote $F^0_\Gm=(F\cap \binom{\Gm}{2})\setminus \cl(\Gm)$. Since $F$ is a chordal completion of $G$, we have that $F^0_\Gm$ is a chordal completion of $\torso(\Gm)$.
From the first point of the lemma, which we have already proved, it follows that $F'''=F''[\Gm\to F^0_\Gm]$ is a chordal completion of $G$.
By the construction, it can be easily seen that $F'''\subseteq F$, hence $F'''=F$ by the minimality of $F$.
This means that $F'''\setminus \binom{\Gm}{2}$, which is equal to $F''\setminus \binom{\Gm}{2}$ by the construction, is also equal to $F\setminus \binom{\Gm}{2}$, which in turn is equal to $F'\setminus \binom{\Gm}{2}$
by the construction. Concluding, we have argued that
\begin{equation*}
F''\cap \binom{\Gm}{2}=F'\cap \binom{\Gm}{2}\qquad\textrm{and}\qquad F''\setminus \binom{\Gm}{2}=F'\setminus \binom{\Gm}{2},
\end{equation*}
so $F''=F'$.
\end{proof}

In a number of arguments, we would like to replace a minimal completion of a segment
so that it uses some particular minimal separator. 
\begin{lemma}\label{lem:force-separator}
Let $G$ be a graph, $F$ be a minimal chordal completion of $G$, and let $\Gm$
be a potential segment in $G$ that is respected by $F$.
Furthermore, let $\Ss$ be a family of pairwise noncrossing
minimal separators in $G$ such that $S \subseteq \Gm$
for every $S \in \Ss$.
Then there exists a minimal completion $F_\Gm$ of $\torso(\Gm)$ such that
every $S \in \Ss$ is a clique in $F_\Gm$. Consequently, every $S \in \Ss$
is a minimal separator in the chordal graph $G+F[\Gm \to F_\Gm]$.
\end{lemma}
\begin{proof}
First, we observe that for every $D \in \cc(G-\Gm)$ and $S \in \Ss$,
  the minimal separators $N(D)$ and $S$
do not cross. Indeed, since $S \subseteq \Gm$, $S$ is disjoint from $D$, and hence
there exists a unique connected component of $G-S$ that contains $D \cup (N(D) \setminus S)$. 
Consequently, a family $\Ss' = \Ss \cup \{N(D) : D \in \cc(G-\Gm)\}$ is a family of pairwise
noncrossing minimal separators in $G$.
By~\cite[Theorem~2.9]{BouchitteT01}, there exists a minimal chordal completion $F'$ of $G$
such that $\Ss'$ are minimal separators in $G+F'$ as well. By defining $F_\Gm \subseteq F'$
to be the set of the nonedges of $\torso(\Gm)$ that are present in $F'$ we obtain the thesis.
\end{proof}

Next, we describe how potential segments arise naturally from clique trees of completed graphs.
Recall for a tree decomposition $(T,\bag)$ and a vertex $u$, by $T_u=T[\{x\in V(T)\colon u\in\bag(x)\}]$ we denote the subtree of $T$ induced by nodes whose bags contain $u$.

\begin{lemma}\label{lem:ctree-segment}
Let $G$ be a graph, $F$ be a minimal chordal completion of $G$, and $H=G+F$.
Further, let $(T,\bag)$ be any clique tree of $H$. Then for every subtree $S$ of $T$, the set
$$\Gm(S):=\bigcup_{x\in V(S)} \bag(x)$$
is a potential segment in $G$. Moreover, $F$ respects $\Gm(S)$. 
\end{lemma}
\begin{proof}
Let $D$ be any connected component of $G-\Gm(S)$.
By properties of tree decompositions it follows that there exists a connected component $T_D$ of $T-S$ such that 
the subtree $T_u$ is entirely contained in $T_D$ for each $u\in D$.
Observe that if $e$ is the unique edge of $T$ that connects a node of $S$ with a node of $T_D$, then 
$$\Gm(S)\cap \bigcup_{x\in V(T_D)} \bag(x) \subseteq \adh(e),$$
where $\adh(vw)=\bag(v)\cap \bag(w)$ is the adhesion of an edge $vw$ of $T$.

Since $N(D)\subseteq \Gm(S)$ and every vertex of $N(D)$ has to be together in some bag with some vertex of $D$, it follows that
$$N(D)\subseteq \Gm(S)\cap \bigcup_{x\in V(T_D)} \bag(x).$$
Consequently $N(D)\subseteq \adh(e)$.
By Proposition~\ref{prop:adh}, $\adh(e)$ is a minimal separator in $G$, so let $D'$ be a component of $G-\adh(e)$ that is full to $\adh(e)$ and is disjoint with $D$.
Then the component $D''$ of $G-N(D)$ that contains $D'$ is full to $N(D)$. We conclude that there are two components of $G-N(D)$ that are full to $N(D)$, namely $D$ and $D''$, so $N(D)$ is a minimal separator in $G$. Hence, $\Gm(S)$ is a potential segment.

To see that $F$ respects $\Gm(S)$, recall the following standard claim about minimal chordal completions.
It follows, e.g., from~\cite[Lemma 3.5]{BouchitteT01}.
\begin{claim}\label{cl:clique-separates}
Suppose $F$ is a minimal chordal completion of a graph $G$ and $\Om$ is a clique in $G+F$. Then $F$ does not contains any edge with endpoints in different connected components of $G-\Om$.
\end{claim}
Take now any $D\in \cc(G-\Gm(S))$, then $D$ is also a connected component of $G-N(D)$. In the previous paragraph we have argued that $N(D)\subseteq \adh(e)$ for some edge $e$ of the clique tree $T$, 
so in particular $N(D)$ is a clique in $H$. By Claim~\ref{cl:clique-separates}, it follows that there is no edge in $F$ with one endpoint in $D$ and the other outside of $N[D]$. 
\end{proof}

In the sequel we adopt the notation $\Gm(S)$ introduced in the statement of Lemma~\ref{lem:ctree-segment}.
The following simple statement will be useful.

\begin{lemma}\label{lem:adh-torsoed}
Let $G$ be a graph and $F$ be a minimal chordal completion of $G$.
Further, let $(T,\bag)$ be any clique tree of $G+F$ and $S$ be a subtree of $T$.
Then for each edge $e\in E(T)$ that connects a node of $S$ with a node outside of $S$, there exists a component $D\in \cc(G-\Gamma(S))$ such that $\adh(e)=N(D)$.
Consequently, $\adh(e)$ is a clique in $\torso(\Gamma(S))$.
\end{lemma}
\begin{proof}
Suppose the removal of $e$ splits $T$ into two subtrees, where $T'$ is the one that does not contain $S$.
By Proposition~\ref{prop:adh}, there is a component $D$ of $G-\adh(e)$ such that $D$ is full to $\adh(e)$ and $T_u\subseteq T'$ for each $u\in D$.
Then $T_u$ is disjoint with $S$ for each $u\in D$, so in particular $D$ and $\Gamma(S)$ are disjoint.
Since $\adh(e)\subseteq \Gamma(S)$, we infer that $D$ is also a component of $G-\Gamma(S)$, and $\adh(e)=N(D)$ since $D$ is full to $\adh(e)$.
\end{proof}

\newcommand{\comp}{\mathrm{complete}}

The following lemma is the crucial replacement argument. Intuitively, it shows that we can focus on recognizing larger potential segments that have low interaction with $I$, rather than individual potential maximal
cliques.

\begin{lemma}\label{lem:replacement}
Let $G$ be a graph and $I$ be a maximal independent set in $G$. Suppose $\Xx$ is a family of subsets of $V(G)$ with the following property.
There exists an $I$-free minimal chordal completion $F_0$ of $G$, a clique tree $(T,\bag)$ of $G+F_0$, and a partition $\Ss$ of $T$ into vertex-disjoint subtrees, such that
for each $S\in \Ss$ we have that $\Gm(S)\in \Xx$, $|I\cap \Gm(S)|\leq 1$, and $I\cap \Gm(S)\subseteq \inte(\Gm(S))$.
Then, given $\Xx$, one can in polynomial time compute a family $\Ff_\comp(\Xx)$ with $|\Ff_\comp(\Xx)|\leq n^2\cdot |\Xx|$ 
such that there exists an $I$-free chordal completion $F$ such that every maximal clique of $G+F$ is contained in $\Ff_\comp(\Xx)$.
\end{lemma}

\begin{proof}
We construct $\Ff_\comp(\Xx)$ as follows.
For every $\Gm\in \Xx$, verify whether $\Gm$ is a potential segment in $G$. If this is not the case, ignore $\Gm$.
Otherwise, check whether $\inte(\Gm)$ is empty.
If this is the case, compute any minimal chordal completion of $\torso(\Gm)$ and add all its maximal cliques to $\Ff_\comp(\Xx)$.
Otherwise, for each $u\in \inte(\Gm)$, compute any $\{u\}$-free minimal chordal completion of $\torso(\Gm)$ and add all its maximal cliques to $\Ff_\comp(\Xx)$.
Since a chordal graph on $n$ vertices has at most $n$ maximal cliques, it follows that $|\Ff_\comp(\Xx)|\leq n^2\cdot |\Xx|$.
We are left with proving that $\Ff_\comp(\Xx)$ has the required properties.

Starting from $F:=F_0$, we modify $F$ and $(T,\bag)$ gradually, keeping the invariant that $(T,\bag)$ is always a clique tree of $G+F$.
Each step of the modification procedure replaces a part of the completion for one $S\in \Ss$. 
Precisely, let us fix one $S\in \Ss$ for which we apply the replacement, and denote for brevity $\Gm=\Gm(S)$.
By assumption we have $\Gm\in \Xx$, $|I\cap \Gm|\leq 1$, and $I\cap \Gm\subseteq \inte(\Gm)$.
Pick $F_\Gm$ to be a minimal chordal completion of $\torso(\Gm)$ whose bags were included in $\Ff_\comp(\Xx)$, 
where we choose $F_\Gm$ to be the $(I\cap \Gm)$-free one in case $I\cap \Gm\neq \emptyset$,
and we choose $F_\Gm$ arbitrarily otherwise.
We modify $F$ and $(T,\bag)$ as follows:
\begin{itemize}
\item Remove all edges of $(F\cap \binom{\Gm}{2})\setminus \cl(\Gm)$ from $F$ and introduce the edges of $F_\Gm$ instead.
\item Let $(S',\bag')$ be any clique tree of $\torso(\Gm)+F_\Gm$. Remove $S$ from $T$ and replace it with $S'$. For each edge $e$ of $T$ that connected a node of $S$, say $x$, with a node outside of $S$, say $y$,
find any node $x'$ of $S'$ whose bag $\bag'(x')$ contains $\adh(e)$, and connect $x$ with $y'$. Such a node $y'$ exists since $\adh(e)$ is a clique in $\torso(\Gm)$, by Lemma~\ref{lem:adh-torsoed}.
\end{itemize}
It is straightforward to see that in this manner, the (new) tree decomposition, denoted henceforth $(T',\beta')$, is still a clique tree of the (new) graph $G+F$.
First, $T'$ is a tree as we replaced a subtree $S$ of $T$ with another tree $S'$ and connect $S'$ to every connected component of $T-V(S)$ with exactly one edge.
Second, every maximal clique $\Om$ in the new $G+F$ is either entirely contained in $\Gm$, in which case $\Om$ is a maximal clique of the chordal graph $\torso(\Gm)+F_\Gm$ and is
visible as one of the bags of $S'$,
or is entirely contained in $N[D]$ for some component $D\in \cc(G-\Gm)$, in which case $\Om$ is among bags of $T\setminus S = T'-S'$.
Note here that all the maximal cliques of $\torso(\Gm)+F_\Gm$ were included in the construction of $\Ff_\comp(\Xx)$.
Observe also that from the choice of $F_\Gm$ it directly follows that the new completion $F$ is still $I$-free.

By applying this replacement procedure to each $S\in \Ss$ one by one, we eventually obtain a new $I$-free chordal completion $F$ such that every maximal clique of $G+F$ is included in $\Ff_\comp(\Xx)$.
\end{proof}

Armed with the replacement tools, we are now ready to analyze $I$-free completions.
The goal is to show that the PMCs and separators discovered in families in the previous section, together with some new arguments, 
are rich enough in the following sense: any $I$-free completion can be locally modified so that the separators belonging to the so-far discovered
families are sufficiently dense in the completion, that is, the parts of the graph between separators are simple.

\subsection{Neighborhood of an element of an independent set}\label{ss:Nv}

Our first segment of interest is a neighborhood of a vertex from the independent set $I$ in question.
The argumentation here is essentially the same as the corresponding part of the algorithm for $P_5$-free graphs of~\cite{LokshtanovVV14},
but for completeness we recall it in our notation.

\begin{lemma}\label{lem:Nv}
Let $G$ be a graph, $I$ be an independent set in $G$, and $u\in I$.
Suppose $F$ is an $I$-free minimal chordal completion of $G$ and $(T,\bag)$ is a clique tree of $G+F$.
Then $\Gm(T_u)$ is equal to $N[u]$. Consequently, $N[u]$ is a potential segment in $G$ that respects $F$ and such that $|N[u]\cap I|=1$ and $N[u]\cap I\subseteq \inte(N[u])$.
\end{lemma}
\begin{proof}
The last assertion follows from the equality $\Gm(T_u)=N[u]$ due to Lemma~\ref{lem:ctree-segment}, so it suffices to prove that $\Gm(T_u)=N[u]$.
Recall that by definition, $\Gm(T_u)$ comprises all vertices that are contained in some bag of $(T,\bag)$ together with $u$.
Since $(T,\bag)$ is a clique tree of the chordal graph $G+F$, each of its bags is a clique in $G+F$, and hence $\Gm(T_u)$ is equal to $N_{G+F}[u]$.
However $F$ is $I$-free, so $N[u]=N_{G+F}[u]$. The claim follows.
\end{proof}

For an independent set $I$ in a graph $G$, consider a family $\Ff_\mathrm{ind}$ defined as follows: for every $u \in V(G)$, we take any $\{u\}$-free minimal completion $F_u$ of $\torso(N[u])$ and we insert all maximal cliques of $\torso(N[u]) + F_u$ into $\Ff_\mathrm{ind}$.
Clearly $|\Ff_\mathrm{ind}| \leq n^2$.
To streamline further arguments, we prefer to cope with replacements of segments $N[u]$ for $u \in I$ separately, without the full strength of Lemma~\ref{lem:replacement}.

\begin{lemma}\label{lem:Nv-replace}
Let $G$ be a graph and $I$ an independent set in $G$. There exists an $I$-free minimal completion $F$ of $G$ such that every maximal clique of $G+F$ that contains an element of $I$ belongs to $\Ff_\mathrm{ind}$.
\end{lemma}
\begin{proof}
Let $F$ be an $I$-free minimal completion of $G$ that minimizes the number of maximal cliques $\Om$ that satisfy $\Om \cap I \neq \emptyset$ and $\Om \notin \Ff_\mathrm{ind}$.
We claim that there is no such $\Om$. Assume the contrary, let $\Om$ be such a maximal clique, and let $\{u\} = \Om \cap I$.

By Lemma~\ref{lem:Nv}, $N[u]$ is a potential segment respected by $F$. Consequently, $F' := F[N[u] \to F_u]$ is an $I$-free minimal completion of $G$.
Observe that, for every $u' \in I \setminus \{u\}$, $F'$ and $F$ do not differ on the set of edges with both endpoints in $N[u']$:
for any two distinct $v_1,v_2 \in N[u'] \cap N[u]$, $v_1$ and $v_2$ are contained in $N(D')$ for $D'$ being the component of $G-N[u]$ that contains $u'$, and hence
$v_1v_2 \in E(\torso(N[u]))$. 
Consequently, every maximal clique $\Om'$ of $G+F'$ that contains  $u'$ 
is also a maximal clique of $G+F$. Thus, $G+F'$ contains strictly less maximal cliques $\Om'$ with $\Om' \cap I \neq \emptyset$ and $\Om' \notin \Ff_\mathrm{ind}$: every such maximal
clique with an element of $I \setminus \{u\}$ is also present in $G+F$, while there are no such cliques $\Om'$ with $u \in \Om'$ (in particular, $\Om$ is not a maximal clique of $G+F'$). 
This contradicts the minimality of $F$.
\end{proof}

Henceforth we will focus only on $I$-free minimal completions satisfying the properties of Lemma~\ref{lem:Nv-replace} for a fixed family $\Ff_\mathrm{ind}$, and call them \emph{$I$-clean}.

\newcommand{\meshsep}{\mathcal{S}^{\mathrm{mesh}}}

\subsection{Separators in the direction of a mesh component}\label{ss:modif-mesh}

We now focus on the following setting. 
Let $F$ be an $I$-free minimal chordal completion of $G$
and let $S$ be a minimal separator in $G+F$ such that there exists a component 
$D$ of $G-S$ that is full to $S$ and that is a mesh.
We now study the structure of minimal separators in $G+F$ ``between'' $S$ and the unique
module of $G[D]$ that contains vertices of $I \cap D$, denoted later
in this section as $M_{F,D}$.

The goal of our study is to show that these separators in $G+F$ are linearly ordered, in particular, there is a well-defined separator ``closest'' to $M_{F,D}$.
Furthermore, we show that there exists a good notion of
a canonical separator between
$M_{F,D}$ and $S$ that can be greedily used in the considered
$I$-free completion. This canonical separator can be inferred 
if we know $M_{F,D}$; in particular, a small family containing a
``fuzzy version of $D$'' gives rise to a small family containing candidates
for the aforementioned canonical separator. See also Figure~\ref{fig:modif-mesh1}.

\begin{figure}[tb]
\begin{center}
\includegraphics{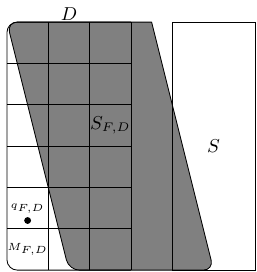}
\end{center}
\caption{The setting of Section~\ref{ss:modif-mesh}.
The separators between $S$ and $M_{F,D}$ are linearly ordered, 
with $S_{F,D}$ being the last separator. 
Given $M_{F,D}$ and a vertex $q_{F,D} \notin M_{F,D}$ that lives
in the same connected component of $G-S_{F,D}$ as (parts of) $M_{F,D}$,
one can in a canonical way define a separator $S_{F,D}'$ that lives
  between $S_{F,D}$ and $S$, and can therefore be greedily used
  in any $I$-free completion that respects the segment $S \cup S_{F,D}$.}\label{fig:modif-mesh1}
\end{figure}

Let us now proceed with a formal argument.
Recall that the set of minimal separators of $G+F$ is a subset of the family of minimal 
separators of $G$, and that the minimal separators of $G+F$ are noncrossing (cf. Corollary~\ref{cor:chordal-crossing} and Lemma~\ref{lem:fillin-minsep}).
That is, every minimal separator $S'$ of $G+F$ is either contained in $N[D]$, or disjoint with $D$.

Let $\widehat{D}$ be a component of $G-S$ that is full to $S$ and different than $D$.

Given $F$, $S$, and $D$, let $\meshsep_{F, D}$ be the family of those minimal separators $S'$
of $G+F$ that are contained in $N[D]$, are not proper subsets of $S$,
and such that $D-S'$ is not contained in any
maximal proper strong module of $G[D]$. Note that $S \in \meshsep_{F, D}$.
We claim the following:
\begin{lemma}\label{lem:meshsep}
For every $S' \in \meshsep_{F, D}$, there exists one component $D(S')$ of $G-S'$ that equals
$D-S'$ and a different one $\widehat{D}(S')$ that contains $\widehat{D} \cup (S \setminus S')$.
Furthermore, $D(S')$ and $\widehat{D}(S')$ are full to $S'$ and, if $S' \neq S$,
  these are the only two components of $G-S'$ that are full to $S'$.
\end{lemma}
\begin{proof}
Since $S' \in \meshsep_{F,D}$, by the definition of $\meshsep_{F,D}$, there exist
two vertices $p,q \in D \setminus S'$ that belong to two different maximal proper strong modules of $G[D]$.
Since $\Quo(D)$ is a clique, $pq \in E(G)$ and $D \subseteq N[p,q]$.
Consequently, if we define $D(S')$ to be the connected component of $G-S'$ that contains the
edge $pq$, then $D(S') \supseteq D-S'$.
Since $S' \subseteq N[D]$, there exists a connected component $\widehat{D}(S')$ of $G-S'$
that contains $\widehat{D}$. Since $\widehat{D}$ is full to $S$, $\widehat{D}(S')$ contains
all vertices of $S \setminus S'$, too.

Assume now that $S'\neq S$.
Since $S' \subseteq N[D]$ but $S'$ is not a proper subset of $S$, there 
exists $x \in S' \setminus S$; note that $x \in D$.
So we have $N[x] \subseteq N[D]$. Hence, the only components of $G-S'$ that may
potentially contain a neighbor of $x$ are $D(S')$ and $\widehat{D}(S')$.
Since the minimal separator $S'$ needs to contain at least two full components,
we infer that $D(S') \neq \widehat{D}(S')$ and these two components are the only
components of $G-S'$ that are full to $S'$.
Furthermore, $D(S')$ is disjoint with $S$ as $S \setminus S' \subseteq \widehat{D}(S')$,
  and, hence, $D(S') = D-S'$.

It remains to observe that in case $S=S'$, the claims are straightforward
with $D(S) = D$ and $\widehat{D}(S) = \widehat{D}$.
Note that in this case we do not claim that $D$ and $\widehat{D}$ are the only full components of $S=S'$.
\end{proof}

\begin{corollary}\label{cor:meshsep}
For every $S',S'' \in \meshsep_{F, D}$, we have $D(S') \subseteq D(S'')$ or $D(S') \supseteq D(S'')$.
In particular, there exists a unique separator $S_{F,D} \in \meshsep_{F,D}$ with inclusion-wise
minimal $D(S_{F,D})$.

Consequently, $\Gamma_{F,D} := S \cup S_{F,D}$ is a potential segment in $G$ respected by $F$.
\end{corollary}
\begin{proof}
Consider two separators $S',S'' \in \meshsep_{F,D}$. 
Since they are minimal separators of $G+F$, they are noncrossing, so in particular,
$S''$ is either contained in $N[D(S')]$ or disjoint with $D(S')$.
Clearly, in the latter case, we have $D-S'' = D(S'') \supseteq D(S')$.
If this is not the case, we have $S'' \subseteq N[D(S')]$ and $S'' \cap D(S') \neq \emptyset$.
Since $S'' \subseteq N[D(S')] = D(S') \cup S'$, we have $\widehat{D}(S'') \supseteq \widehat{D}(S')$
and thus any vertex of $S' \setminus S''$ belongs to $\widehat{D}(S'')$. 
Hence, no vertex of $S'$ can be an element of $D(S'')$, and hence $D(S'') \subseteq D(S')$.

To see that $\Gamma_{F,D} = S \cup S_{F,D}$ is a potential segment in $G$, note that $G-\Gamma_{F,D}$
has only components $(\cc(G-S) \setminus \{D\}) \cup \{D(S_{F,D})\}$. 
Then, $N(D(S_{F,D})) = S_{F,D}$, $N(\widehat{D}) = S$, and for every other component $D' \in \cc(G-S) \setminus \{D,\widehat{D}\}$
we have that $N(D') \subseteq S$ is a minimal separator in $G$.
Since $S_{F,D}$ and $S$ are both minimal separators of $G+F$, $F$ respects $\Gamma_{F,D}$.
\end{proof}

Let us now focus more closely on the separator $S_{F,D}$ and the component $D(S_{F,D})$. 
By Lemma~\ref{lem:chordal-minsep-clique}, there exists a (not necessarily unique) maximal clique
$\Om_{F,D}$ of
$G+F$ that contains $S_{F,D}$ and is contained in $N[D(S_{F,D})]$. 
The minimality of $D(S_{F,D})$ implies the following.
\begin{lemma}\label{lem:MFD}
There exists a maximal proper strong module $M_{F,D}$ of $G[D]$ that contains all components $D'$
of $G-\Om_{F,D}$ with $N(D') \not\subseteq S_{F,D}$.
\end{lemma}
\begin{proof}
Let $D'$ be a component of $G-\Om_{F,D}$ with $N(D') \not\subseteq S_{F,D}$.
Then, as $\Om_{F,D} \setminus S_{F,D} \subseteq D(S_{F,D})$, we have $D' \subsetneq D(S_{F,D})$.
By the minimality of $S_{F,D}$ in $\meshsep_{F,D}$, we have $N(D') \notin \meshsep_{F,D}$.
This can only happen due to $D'$ being contained in a single maximal proper strong module
of $G[D]$. Furthermore, since $\Quo(D)$ is a clique, all such components $D'$ are contained
in the same module, denoted henceforth $M_{F,D}$.
\end{proof}
If there are no components $D'$ as in Lemma~\ref{lem:MFD}, we take $M_{F,D}$ to be any 
maximal proper strong module of $G[D]$ that has a nonempty intersection with $D(S_{F,D})$.
Since $S_{F,D} \in \meshsep_{F,D}$, we can pick a vertex $q_{F,D} \in D(S_{F,D}) \setminus M_{F,D}$.
Note that since all components $D'$ of $G-\Om_{F,D}$ with $N(D') \not\subseteq S_{F,D}$ are contained in $M_{F,D}$, by Lemma~\ref{lem:MFD},
and $q_{F,D}\notin M_{F,D}$, it follows that $q_{F,D}\in \Om_{F,D}\setminus S_{F,D}$.

The tuple $(S_{F,D}, \Om_{F,D}, M_{F,D}, q_{F,D})$ as defined above is called a \emph{footprint} of the component $D$.
Note that, although $S_{F,D}$ is defined uniquely, the remaining components of the footprint may not be defined uniquely.

We now show that $q_{F,D}$ and $M_{F,D} \cap D(S_{F,D})$ define the separator $S_{F,D}$.
To this end, we first show the following generic lemma; recall that for a connected nonempty
set of vertices $A \subseteq V(G)$, Lemma~\ref{lem:segment-NA} asserts that $N[A]$ 
is a potential segment and we defined $\partial A = N(V(G) \setminus N[A])$.
\begin{lemma}\label{lem:SFD-make-S}
For a set $A \subseteq D$ that is not contained in a single maximal proper strong module of $G[D]$,
the set $\partial A$
is a minimal separator in $G$ with one full component containing $A$
and one full component containing $\widehat{D}$.
\end{lemma}
\begin{proof}
Clearly, $\partial A \subseteq N(A)$.

Let $D(A) = N[A] \setminus \partial A = N[A] \setminus N(V(G) \setminus N[A])$.
Clearly, $A \subseteq D(A)$. 
Since $G[A]$ is connected and $D(A) \subseteq N[A]$, we have that $G[D(A)]$ is connected
and therefore is a connected component of $G-\partial A$.
Since $A$ contains vertices of at least two maximal proper strong modules of $G[D]$,
we have $D \subseteq N[A]$ and $D \setminus \partial A \subseteq D(A)$.
Also, note that, as $\partial A \subseteq N(A)$ and $A \subseteq D(A)$, the component $D(A)$
is full to $\partial A$.

Since $A \subseteq D$, we have $N[A] \subseteq N[D]$, and hence there exists a unique
component $\widehat{D}(A)$ of $G-\partial A$ that contains $\widehat{D}$
and $S\setminus N[A]$. We claim that $\widehat{D}(A)$ is full to $\partial A$.
To this end, consider an arbitrary $x \in \partial A$. If $x \in S$, we have
$x \in N(\widehat{D})$ and $\widehat{D} \subseteq \widehat{D}(A)$. 
Otherwise, we have $x \in D$. Since $x \in \partial A$, there exists a neighbor
$y \in N(x) \setminus N[A]$. As $D \subseteq N[A]$, we need to have $y \in S$ and, consequently,
$y \in \widehat{D}(A)$. This proves that $\widehat{D}(A)$ is full to $\partial A$,
completing the proof of the lemma.
\end{proof}

\begin{lemma}\label{lem:SFD-def}
$S_{F,D}$ is equal to $\partial A_{F,D}$ for
$A_{F,D} = \{q_{F,D}\} \cup (M_{F,D} \cap D(S_{F,D}))$.
\end{lemma}
\begin{proof}
By definition, $A_{F,D} \subseteq D(S_{F,D})$ and, hence, $N[A_{F,D}] \subseteq N[D(S_{F,D})]$. 
Observe that it suffices to show that $N[A_{F,D}] = N[D(S_{F,D})]$.
Assume otherwise, and let $x \in N[D(S_{F,D})] \setminus N[A_{F,D}]$.

Since $A_{F,D}$ contains vertices from two different maximal proper strong modules of $G[D]$,
we have $D \subseteq N[A_{F,D}]$. Since $D(S_{F,D}) \subseteq D$, we have $x \in N[D]$. Hence,
$x \in S$ and $x \in S_{F,D}$.

Recall that $q_{F,D} \in \Om_{F,D} \setminus S_{F,D}$. Since $q_{F,D} \in A_{F,D}$ but $x \notin N[A]$, we have $q_{F,D}x \notin E(G)$.
As $\Om_{F,D}$ is a potential maximal clique in $G$, there exists a component $D'$ of $G-\Om_{F,D}$
with $q_{F,D}, x \in N(D')$. 

Since $q_{F,D} \notin S_{F,D}$, $D' \subseteq D(S_{F,D})$ and we have $D' \subseteq M_{F,D}$ from Lemma~\ref{lem:MFD}.
However, then $D' \subseteq A_{F,D}$, a contradiction to the fact that $x \notin N[A_{F,D}]$. 
This contradiction finishes the proof of the lemma.
\end{proof}

\begin{corollary}\label{cor:SFD}
Define $A_{F,D}' := M_{F,D} \cup \{q_{F,D}\}$. Then the following holds:
\begin{itemize}
\item $G[A_{F,D}']$ is connected,
\item $S'_{F,D} := \partial A'_{F,D}$ is a minimal separator in $G$,
\item $N[A_{F,D}'] \setminus S'_{F,D}$ contains $D(S_{F,D})$,
\item $S'_{F,D} \subseteq \Gamma_{F,D}$, in particular, $S'_{F,D}$ is disjoint with $D(S_{F,D})$.
\end{itemize}
\end{corollary}
\begin{proof}
$A_{F,D}'$ is connected, because $D$ is a mesh.
By Lemma~\ref{lem:SFD-make-S}, the separator $S'_{F,D} := \partial A'_{F,D}$ is a minimal separator
in $G$.
Clearly, $S'_{F,D} \subseteq N[D]$.
Since $A_{F,D} \subseteq A'_{F,D}$ and $S_{F,D} = \partial A_{F,D}$ by Lemma~\ref{lem:SFD-def} we have that
$S'_{F,D}$ is disjoint from $D(S_{F,D})$ and $D(S_{F,D}) \subseteq N[A_{F,D}'] \setminus S'_{F,D}$.
Consequently, we have $S'_{F,D} \subseteq \Gamma_{F,D}$.
\end{proof}

Lemma~\ref{lem:SFD-def} and Corrolary~\ref{cor:SFD} suggest the following modifying operation, which we henceforth call a \emph{footprint replacement}.
Fix a footprint $(S_{F,D}, \Om_{F,D}, M_{F,D}, q_{F,D})$ of $D$. 
Define $A'_{F,D}$ and $S'_{F,D}$ as in Corollary~\ref{cor:SFD}.
By Lemma~\ref{lem:force-separator} there exists a minimal completion $F'_{F,D}$ of $\torso(\Gamma_{F,D})$ that keeps
$S'_{F,D}$ as a minimal separator, and we can consider $F[\Gamma_{F,D} \to F'_{F,D}]$ instead of $F$ as $F$ respects $\Gamma_{F,D}$.

The footprint replacement operation will be useful later in this section, as in some cases there is only a polynomial number of choices for 
the module $M_{F,D}$ and the vertex $q_{F,D}$,
giving us polynomial number of choices for the set $A'_{F,D}$ and the separator $S'_{F,D}$. This in particular happens if, due to one of the lemmata from the previous section, there is only a polynomial number of choices for 
a fuzzy version of $D$.

\subsection{A separator with two full mesh components}\label{ss:twomesh}

In the previous section we have shown that from the knowledge of the module $M_{F,D}$ and the vertex $q_{F,D}$ in a footprint $(S_{F,D}, \Om_{F,D}, M_{F,D}, q_{F,D})$ one can
deduce a canonical separator between $S_{F,D}$ and $S$. 
A small number of choices for $M_{F,D}$ and $q_{F,D}$ follows from
a small number of choices for a fuzzy version of $D$.
Unfortunately, sometimes we do not even have the above;
this happens in case of a separator with two full components being meshes. 
In this section we focus on this situation.

Let $S$ be a minimal separator
with two full components $D_1$ and $D_2$ being meshes, and 
for $i=1,2$ let $(S_{F,D_i}, \Om_{F,D_i}, M_{F,D_i}, q_{F,D_i})$ be a footprint
of $D_i$. 
Furthermore, let $R = V(G) \setminus (D_1 \cup D_2 \cup S)$ be the \emph{rest}
of the graph.
We first observe that 
the Separator Covering Lemma (Lemma~\ref{lem:covering-general}),
   together with the assumption that $D_1$ and $D_2$ are meshes,
   asserts
existence of vertices $p_1,q_1,r_1 \in D_1$ and $p_2,q_2,r_2 \in D_2$
such that $N[p_1,q_1,r_1,p_2,r_2,q_2] = D_1 \cup D_2 \cup S = V(G) \setminus R$.
In particular, there are only $n^6$ reasonable choices for $R$.

\begin{figure}[tb]
\begin{center}
\includegraphics{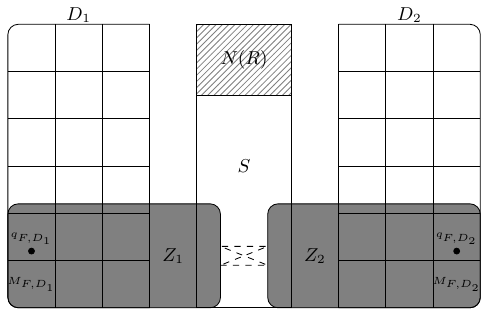}
\end{center}
\caption{The setting of Section~\ref{ss:twomesh}.
  We show that within a polynomially sized family of candidates
   there are supersets $Z_i$ of $M_{F,D_i} \cup \{q_{F,D_i}\}$
   that are disjoint with $N(R)$ and anti-adjacent to each other.}\label{fig:twomesh1}
\end{figure}

Intuitively, we would like to split the graph in a canonical way
between $S_{F,D_1}$, $S_{F,D_2}$ and $N[R]$. To this end, we show
that one can identify connected sets $Z_1$ and $Z_2$ with the following properties:
$q_{F,D_i} \in Z_i$, $M_{F,D_i} \subseteq Z_i$ for $i=1,2$, but there are no edges between
sets $R$, $Z_1$, and $Z_2$; see also Figure~\ref{fig:twomesh1}. More formally, we prove the following lemma:

\begin{lemma}\label{lem:rozrywanie}
One can in polynomial time compute a family $\Ff_9$ of size at most $n^6$
such that the following holds.
Let $S$ be a minimal separator in $G$ and let $D_1$ and $D_2$ be two components
of $G-S$ that are full to $S$ and are meshes.
For $i=1,2$, let $M_i^p$ be an arbitrary proper strong module of $D_i$ 
and let $q_i \in D_i \setminus M_i^p$ be arbitrary.
Then there exists an element $(Z_1,Z_2) \in \Ff_9$ such that
$Z_1,Z_2 \subseteq V(G)$ induce connected subgraphs of $G$,
there are no edges between $Z_1$ and $Z_2$ and
 for $i=1,2$, we have
$$M_i^p \cup \{q_i\} \subseteq Z_i \subseteq (D_i \cup S) \setminus N(V(G) \setminus (S \cup D_1 \cup D_2)).$$
\end{lemma}
\begin{proof}
Let $M_i^p$ and $q_i$ be as in the lemma statement.
Let $p_i \in M_i^p$ be arbitrary and let $M_i^q$ be the maximal proper strong module of $D_i$
that contains $q_i$.

By the Separator Covering Lemma (Lemma~\ref{lem:covering-general}), there exist
$r_1 \in D_1$ and $r_2 \in D_2$ such that
$$N[p_1,q_1,r_1,p_2,q_2,r_2] \supseteq S.$$
Note that, as both $D_1$ and $D_2$ are meshes and $p_i,q_i$ belong to different maximal
proper strong modules of $D_i$, we actually have
$$N[p_1,q_1,r_1,p_2,q_2,r_2] = S \cup D_1 \cup D_2.$$
Consequently, $R := V(G) \setminus N[p_1,q_1,r_1,p_2,q_2,r_2]$ is exactly the vertex set
of all connected components of $G-S$ except for $D_1$ and $D_2$, that is,
$R = V(G) \setminus (S \cup D_1 \cup D_2)$.
Our goal is to uniquely construct a pair $(Z_1,Z_2)$ given the tuple
$(p_1,q_1,r_1,p_2,q_2,r_2)$; by inserting the pair $(Z_1,Z_2)$ for every choice of this tuple
we obtain the desired family $\Ff_9$.

Define
$$O_S := (N[p_1,q_1,r_1] \cap N[p_2,q_2,r_2]) \cup N(R).$$
Clearly, $O_S \subseteq S$, and observe that $O_S$ is a function of the six vertices
$p_1,q_1,r_1,p_2,q_2,r_2$ only. Furthermore, every vertex of $S \setminus O_S$ has neighbors
only in $S \cup D_1 \cup D_2$; the notation $O_S$ can be read as vertices \emph{obviously in $S$}.

Let us partition $S \setminus O_S$ further. For $i=1,2$, let $S_i^{P_4} \subseteq S \setminus O_S$
be the set of these vertices $u \in S \setminus O_S$ for which there exists an induced $P_4$ of the form
$uD_{3-i}D_{3-i}D_{3-i}$. Note that every $u \in S_i^{P_4}$ is complete to $D_i$, as otherwise
an induced $P_3$ of the form $uD_iD_i$ and a $P_4$ of the form $uD_{3-i}D_{3-i}D_{3-i}$ would yield together an induced $P_6$.
In particular, $S_1^{P_4}$ and $S_2^{P_4}$ are disjoint.

For $i=1,2$, let
$$S_i^{pqr} = (S \setminus (O_S \cup S_i^{P_4})) \cap N[p_i,q_i,r_i].$$
Since $N[p_1,q_1,r_1] \cap N[p_2,q_2,r_2] \subseteq O_S$ and $S_{3-i}^{P_4}$ is complete to $D_{3-i}$, the set $S_i^{pqr}$ is disjoint
from both $S_{3-i}^{pqr}$ and $S_{3-i}^{P_4}$. 
Furthermore, the Neighborhood Decomposition Lemma (Lemma~\ref{lem:nei-decomp})
ensures that $S_1^{P_4}, S_2^{P_4}, S_1^{pqr}, S_2^{pqr}$ is a partition of $S \setminus O_S$
and that every $u \in S_i^{pqr}$ is tricky towards $(D_{3-i}, p_{3-i}, q_{3-i})$.
In particular, the latter ensures that every $u \in S_i^{pqr}$ does not have any neighbor
in $M_{3-i}^p \cup M_{3-i}^q$.

\begin{figure}[tb]
\begin{center}
\includegraphics{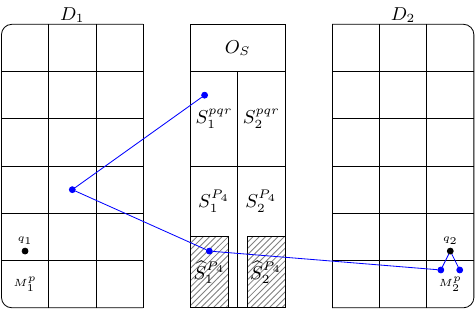}
\end{center}
\caption{Partition of $S$ analyzed in the proof of Lemma~\ref{lem:rozrywanie},
  and the sought $P_6$ depicted blue.}\label{fig:twomesh2}
\end{figure}

For $i=1,2$, let $\widehat{S}_i^{P_4} \subseteq S_i^{P_4}$ be the set of those vertices
$u \in S_i^{P_4}$ for which $N(u) \cap M_{3-i}^p \neq \emptyset$.
The crux of the proof is the following observation (cf. Figure~\ref{fig:twomesh2}):
\begin{claim}\label{cl:capture-mesh:daszek}
$\widehat{S}_i^{P_4}$
is complete to $N[p_i,q_i,r_i] \setminus (O_S \cup \widehat{S}_i^{P_4})$.
\end{claim}
\begin{proof}
Let $u \in \widehat{S}_i^{P_4}$ and $v \in N[p_i,q_i,r_i] \setminus (O_S \cup \widehat{S}_i^{P_4})$;
we want to show that $u$ and $v$ are adjacent.

If $v \in D_i$, the statement is straightforward as $S_i^{P_4}$ is complete to $D_i$, so assume
$v \in S$. Since $v \notin O_S$ nor $v \notin \widehat{S}_i^{P_4}$, we have $v \in S_i^{pqr}$
or $v \in S_i^{P_4} \setminus \widehat{S}_i^{P_4}$.
In both cases, we have that $v$ is not adjacent to $q_{3-i}$ nor to any vertex in $M_{3-i}^p$:
for $v \in S_i^{pqr}$ it follows from the fact that the vertices of $S_i^{pqr}$ are tricky
towards $(D_{3-i},p_{3-i},q_{3-i})$, and for $v \in S_i^{P_4} \setminus \widehat{S}_i^{P_4}$
it follows from the definition of $\widehat{S}_i^{P_4}$ and the fact that $S_i^{P_4} \subseteq
N[p_i,q_i,r_i]$.

Consider $A = N(u) \cap M_{3-i}^p$.
Clearly, $A \neq \emptyset$ by the definition of $\widehat{S}_i^{P_4}$.
Also, $p_{3-i} \in M_{3-i}^p \setminus A$, as otherwise
$u \in N[p_i,q_i,r_i] \cap N[p_{3-i}] \subseteq O_S$.
Since $M_{3-i}^p$ is a maximal proper strong module of a mesh $D_{3-i}$, $A$ cannot be complete
to $M_{3-i}^p\setminus A$,
that is, there exist $x \in A$ and $y \in M_{3-i}^p \setminus A$ with $xy \notin E(G)$.
Observe that $(u, x, q_{3-i}, y)$ is an induced $P_4$, denoted henceforth $P$, of the form $uD_{3-i}D_{3-i}D_{3-i}$, as
$u \notin N[p_{3-i}, q_{3-i}, r_{3-i}]$. 

Furthermore, we have that $v$ is not adjacent to any of $x$, $y$, or $q_{3-i}$.
Hence, if $vu \notin E(G)$, then $P$ together with a shortest path from $u$ to $v$
with all internal vertices lying in $D_i$ would yield an induced $P_6$ in $G$, a contradiction.
This finishes the proof of the claim.
\cqed\end{proof}

For $i=1,2$, define $N_i$ as follows.
If the quotient graph of $G[N[p_i,q_i,r_i] \setminus O_S]$ is a clique, let $N_i$ be the maximal
proper strong module of $G[N[p_i,q_i,r_i] \setminus O_S]$ that contains $p_i$; otherwise,
let $N_i = N[p_i,q_i,r_i] \setminus O_S$.
A consequence of Claim~\ref{cl:capture-mesh:daszek} is the following:
\begin{claim}\label{cl:capture-mesh:Ni}
We have $M_i^p \subseteq N_i$ and $N_i \cap N[M_{3-i}^p] = \emptyset$.
\end{claim}
\begin{proof}
For the first claim, note that $p_i \in N_i$ and $M_i^p \subseteq N[p_i,q_i,r_i] \setminus O_S$.
Thus, if $M_i^p \not\subseteq N_i$, then it must be the case that the quotient graph
of $G[N[p_i,q_i,r_i] \setminus O_S]$ is a clique and $N_i$ is one of the maximal proper strong
modules of this graph. However, then $N_i$ is complete to $N[p_i,q_i,r_i] \setminus (O_S \cup N_i)$,
in particular, $N_i \cap M_i^p$ is complete to $M_i^p \setminus N_i$. This is a contradiction
as $M_i^p$ is a maximal proper strong module of a mesh $D_i$ and both $N_i \cap M_i^p$
and $M_i^p \setminus N_i$ are nonempty.

We now move to the second claim. Note that Claim~\ref{cl:capture-mesh:daszek} ensures that if $\widehat{S}_i^{P_4}$ is nonempty, then
$\widehat{S}_i^{P_4}$ is a module of $G[N[p_i,q_i,r_i] \setminus O_S]$ complete to the rest of this graph, so in particular then $G[N[p_i,q_i,r_i] \setminus O_S]$ is a mesh.
Then, by definition of $N_i$, we have $N_i \cap \widehat{S}_i^{P_4} = \emptyset$;
note that this is also trivially true if $\widehat{S}_i^{P_4}=\emptyset$.
Consequently, $N_i \subseteq D_i \cup S_i^{pqr} \cup (S_i^{P_4} \setminus \widehat{S}_i^{P_4})$, regardless of whether $\widehat{S}_i^{P_4}$ is empty or not.
However, neither of these three sets have any neighbors in $M_{3-i}^p$:
$D_i$ is anti-complete to $D_{3-i} \supseteq M_{3-i}^p$,
 $S_i^{pqr}$ is tricky towards $(D_{3-i},p_{3-i},q_{3-i})$, and 
$\widehat{S}_i^{P_4}$ is defined as exactly these vertices of $S_i^{P_4}$ that have neighbors
in $M_{3-i}^p$.
\cqed\end{proof}

Claim~\ref{cl:capture-mesh:Ni} allows us to define $Z_i$ to be the connected component
of $G[(N_i \cup \{q_i\}) \setminus N[N_{3-i}]]$ that contains $p_i$ and $q_i$.
Note that this is a proper definition: by Claim~\ref{cl:capture-mesh:Ni}, no vertex of $M_i^p$
is adjacent to any vertex of $N_{3-i}$, and no neighbor of $q_i$ belongs to $N_{3-i}$
because $N_{3-i} \subseteq N[p_{3-i},q_{3-i},r_{3-i}] \setminus O_S$.
We summarize the properties of the sets $Z_i$ in the following claim.

\begin{claim}\label{cl:capture-mesh:Zi}
For $i=1,2$, the set $Z_i$ induces a connected subgraph of $G$, contains $q_i$ and $M_i^p$, and is contained
in $N[p_i,q_i,r_i] \setminus N(R)$ where $R = V(G) \setminus (D_1 \cup D_2 \cup S)$. 
Furthermore, $Z_1$ and $Z_2$ are disjoint and there is no edge connecting a vertex of $Z_1$ and a vertex of $Z_2$.
\end{claim}
\begin{proof}
The fact that $G[Z_i]$ is connected, the disjointness of $Z_1$ and $Z_2$, and the fact that there is no edge between $Z_1$ and $Z_2$ follows directly
from the definition of $Z_i$, as $Z_i \subseteq N_i$.
Furthermore, $Z_i \subseteq N_i$ implies $Z_i \subseteq N[p_i,q_i,r_i] \setminus N(R)$.
Finally, the fact that $q_i$ and $M_i^p$ are contained in $Z_i$ follows from the fact that both $q_i$ and $M_i^p$ are disjoint from $N[N_{3-i}]$ by Claim~\ref{cl:capture-mesh:Ni},
$N_{3-i} \subseteq N[p_{3-i},q_{3-i},r_{3-i}] \setminus O_S$, and that $q_i$ is complete to $M_i^p$.
\cqed\end{proof}
Claim~\ref{cl:capture-mesh:Zi} concludes the proof of the lemma.
\end{proof}

\subsection{Chopping into recognizable segments}\label{ss:chopping}

Let us now gather a family of candidate minimal separators that were recognized in lemmas
so far.

Before we list the candidate minimal separators, we note that it is straightforward to
verify in polynomial time if a given set $A \subseteq V(G)$ is a potential segment. 
In the list below, whenever we specify that we perform some action for a potential segment
from some set, we implicitly mean that we discard elements from the said set that
are not potential segments.
Note that potential maximal cliques, promised to be listed in
the family $\Ff_\mathrm{ind}$ of Section~\ref{ss:Nv} or the sets 
$\Omega \cup D_1 \cup D_2$, promised to be listed in Lemma~\ref{lem:Omplus},
are potential segments.

Into a family $\Ss$ we insert every set that is a minimal
separator in $G$ and is defined in one of the following fashions:
\begin{enumerate}
\item $N(D)$ for $D$ from one of the following families:
\begin{enumerate}
\item $\Ff_1$ from Lemma~\ref{lem:two-not-whole};
\item $\Ff_2$ from Lemma~\ref{lem:one-in-three-see};
\item $\Ff_4$ from Lemma~\ref{lem:capture-nonmesh};
\item $\Ff_8$ from Lemma~\ref{lem:Omplus};
\end{enumerate}
\item $N(D)$ for every $D \in \cc(G-\Gm)$ for potential segments $\Gm$ from one of the
following families:
\begin{enumerate}
%\item $\Ff^1_9$ from Lemma~\ref{lem:summary};
\item $\Ff_8$ from Lemma~\ref{lem:Omplus};
\item $\Ff_\mathrm{ind}$ from Section~\ref{ss:Nv};
\end{enumerate}
\item $N(D)$ for every $D \in \cc(G-N[A])$ for potential segments $N[A]$ (cf. Lemma~\ref{lem:segment-NA}) for connected $A$ defined as follows:
\begin{enumerate}
\item for every $D^+ \in \Ff_5$ from Lemma~\ref{lem:capture-mesh-nomesh}, if $D^+$ is a mesh,
  we iterate over every $A = M \cup \{q\}$ for a maximal proper strong module $M$ of $G[D^+]$
  and $q \in D^+ \setminus M$;
%\item similarly as above, for every $D_1^+$ and $D_2^+$ in every element $(\Om\cup D_1\cup D_2, D_1^+, D_2^+) \in \Ff^2_9$ from Lemma~\ref{lem:summary};
\item $A = Z_1$ and $A= Z_2$ for every element $(Z_1,Z_2) \in \Ff_9$ from Lemma~\ref{lem:rozrywanie}.
\end{enumerate}
\end{enumerate}
Let $\Ff = \bigcup \{\cc(G-S) : S \in \Ss\}$.
Our goal is to show that one can choose an $I$-clean minimal chordal completion $F$ of $G$
such that the separators of $\Ss$ chop $G+F$ into relatively simple pieces.

\begin{lemma}\label{lem:FfX}
One can compute in polynomial time a family $\Ff_X$ of polynomial size, such that the following holds.
Let $G+F$ be an $I$-clean minimal chordal completion with the maximum possible number of minimal separators
that are elements of $\Ss$.
Let $(T,\beta)$ be a clique tree of $G+F$ and let $E_\Ss \subseteq E(T)$ be those edges of $T$
whose adhesions (that are minimal separators of $G+F$ due to Lemma~\ref{lem:chordal-minsep-adh})
belong to $\Ss$. 
Then there exists $E'_\Ss \supseteq E_\Ss$ such that
for every connected component $T'$ of $T-E_\Ss'$, the potential segment $\Gm(T')$%
\footnote{$\Gm(T')$ is a potential segement respected by $F$ due to Lemma~\ref{lem:ctree-segment}.}
belongs to $\Ff_X$.
\end{lemma}
\begin{proof}
Consider a connected component $T'$ of $T-E_\Ss$; our first goal is to show that
$T'$ is quite simple since $\Ss$ is rich enough. 
First, note that the inclusion of all separators of components from $\Ff_2$ from
Lemma~\ref{lem:one-in-three-see} imply that the maximum degree of a node in $T'$ is at most $2$,
that is, $T'$ is an isolated vertex or a path. 

If $T'$ is an isolated vertex, then $\Gm(T')$ is a PMC in $G$.
We claim that for every $D \in \cc(G-\Gm(T'))$ the set $N(D)$ is included in $\Ss$.
To this end, consider arbitrary $D \in \cc(G-\Gm(T'))$.
By the properties of a tree decomposition, there exists a component $T_D$
of $T-V(T')$ whose bags contain all vertices of $D$. 
Let $e_D$ be the edge connecting $T_D$ with $T'$.
Then either $N(D) = \adh(e_D)$ and then $N(D) \in \Ss$ due to the definition of $E_\Ss$
or $N(D) \subsetneq \adh(e_D)$ and then $N(D) \in \Ss$ due to the inclusion of
the sets $N(D)$ for $D \in \Ff_1$ of Lemma~\ref{lem:two-not-whole}.

Consequently, we infer that $\Gm(T')$
is contained in the family $\Ff_{\rec,1}(\Ff)$ given by Lemma~\ref{lem:recover1}.
We include this family in $\Ff_X$,
and continue with the case when $T'$ is a path.

\begin{claim}\label{cl:FfX:noI}
$\Gm(T') \cap I = \emptyset$.
\end{claim}
\begin{proof}
Assume the contrary. Then there exists $t \in V(T')$ with some $u \in I \cap \beta(t)$.
However, since $F$ is $I$-clean, we have that the maximal clique $\beta(t) \in \Ff_\mathrm{ind}$,
where $\Ff_\mathrm{ind}$ is defined in Section~\ref{ss:Nv}.
Then the minimal separators associated with all edges incident to $t$ belong to $\Ss$,
a contradiction with the assumption that $T'$ is not an isolated vertex.
\cqed\end{proof}

Let $t_1t_2 \in E(T')$, let $T_i$ be the component of $T-t_1t_2$ that
contains $t_i$, and similarly let $T_i'$ be the component of $T'-t_1t_2$ that contains $t_i$.
Let $S = \beta(t_1) \cap \beta(t_2)$ be the minimal separator of $G+F$ corresponding to the
edge $t_1t_2$.
\begin{claim}\label{cl:FfX:D1D2}
For $i=1,2$, there exists a unique component $D_i$ of $G-S$ that contains $\bigcup_{t' \in V(T_i')}
\beta(t') \setminus S$. Furthermore, $D_1$ and $D_2$ are full to $S$.
\end{claim}
\begin{proof}
Let $D_i$ be the component of $G-S$ that contains $\beta(t_i) \setminus S$; such a component
exists and is full to $S$ due to Proposition~\ref{prop:adh}.

Let $t_i = t_i^1, t_i^2, \ldots, t_i^{\ell_i}$ be consecutive vertices on the path $T_i'$.
Denote $t_i^0 = t_{3-i}$.
The crucial observation is that for every $1 \leq j < \ell_i$, 
we have $\adh(t_i^{j-1} t_i^j) \cup \adh(t_i^j t_i^{j+1}) = \beta(t_i^j)$,
   $\adh(t_i^{j-1} t_i^j) \not\subseteq  \adh(t_i^j t_i^{j+1})$, and
we have $\adh(t_i^{j-1} t_i^j) \not\supseteq \adh(t_i^j t_i^{j+1})$,
as otherwise $\adh(t_i^{j-1} t_i^j)$ or $\adh(t_i^j t_i^{j+1})$
would be contained in $\Ss$ due to the inclusion of 
$\{N(D) : D \in \Ff_1\}$ for $\Ff_1$ from Lemma~\ref{lem:two-not-whole}.

We show by induction that for every $1 \leq j \leq \ell_i$, the component $D_i$
\begin{itemize}
\item contains the whole set $\beta(t_i^j) \setminus S$, and
\item contains at least one vertex of $\adh(t_i^jt_i^{j+1})$ whenever $j < \ell_i$.
\end{itemize}
For the base case, $\beta(t_i^1) \setminus S = \beta(t_i) \setminus S$ is contained in $D_i$
by the definition of $D_i$ and, if $\ell_i > 1$, then
$\beta(t_i^1) \setminus S = \adh(t_i^1t_i^2) \setminus S$ by the above observation.

By the properties of a tree decomposition, $S \cap \beta(t_i^j) \subseteq \adh(t_i^jt_i^{j-1})$
for every $1 \leq j \leq \ell_i$. 
Consequently, if $\beta(t_i^j) \setminus S \subseteq D_i$ for some $1 \leq j < \ell_i$, 
then $\adh(t_i^jt_i^{j+1})$ contains a vertex of $D_i$ as it contains $\beta(t_i^j) \setminus \adh(t_i^jt_i^{j-1})$.

Consider now an index $1 < j \leq \ell_i$ and assume $D_i$ contains a vertex $v \in \adh(t_i^jt_i^{j-1})$. Let $D_i^j$ be the component of $G-\adh(t_i^jt_i^{j-1})$ that contains $\beta(t_i^j) \setminus \adh(t_i^jt_i^{j-1})$ (cf. Proposition~\ref{prop:adh}). 
Note that $D_i^j \cap S = \emptyset$, thus $D_i^j$ is contained in a single connected component
of $G-S$. Since $D_i^j$ is full to $\adh(t_i^jt_i^{j-1})$, we have $v \in N(D_i^j)$, and thus
$D_i^j$ is contained in $D_i$. 
This finishes the proof of the claim.
\cqed\end{proof}

Note that $S \subseteq \Gm(T')$, and thus each component of $G-\Gm(T')$
is a subset of some component of $G-S$.
\begin{claim}\label{cl:FfX:S-Di-sep}
Let $D_1$ and $D_2$ be as in Claim~\ref{cl:FfX:D1D2}.
Let $i \in \{1,2\}$ be such that $D_i$ is a mesh and at least two connected components
of $G-\Gm(T')$ are contained in $D_i$.
Then, there exists a footprint $(S_{F,D_i}, \Om_{F,D_i}, M_{F,D_i}, q_{F,D_i})$ of $D_i$ such that the following holds:
$D_i \setminus \Gm(T') \subseteq M_{F,D_i}$
and either $S_{F,D_i} = S$ or $S_{F,D_i}=\adh(e_i)$ for some $e_i \in E(T_i')$.
\end{claim}
\begin{proof}
For ease of presentation, let $i=1$, that is, $D_1$ is a mesh and $D_1$ contains
at least two connected components of $G-\Gm(T')$.

First note that if $D_1 \setminus \Gm(T')$ consists of at least
two connected components of $G-\Gm(T')$, then $D_1 \setminus \Gm(T')$ is contained
in a unique maximal proper strong module of $D_1$. Let us denote this module by $M_1$.

Recall that the separator $S_{F,D_1}$ is defined uniquely. 
By Lemma~\ref{lem:chordal-minsep-adh}, there exists an edge $s^1s^2 \in E(T)$
with $\adh(s^1s^2) = \beta(s^1) \cap \beta(s^2) = S_{F,D_1}$ and
$\beta(s^1) \subseteq N[D(S_{F,D_1})] \subseteq N[D_1]$.
By applying Proposition~\ref{prop:adh} to the edge $s^1s^2$, there exist components $D^1$ and $D^2$, that are full to $S_{F,D_1}$.

If $S = S_{F,D_1}$, then we can take $(s^1, s^2) = (t_1, t_2)$ and have $(D^1, D^2) = (D_1,D_2)$.
Otherwise, Lemma~\ref{lem:meshsep} asserts that $D(S_{F,D_1})$ and $\widehat{D}(S_{F,D_1})$
are the only two full components of $S_{F,D_1}$, and we have $D(S_{F,D_1}) = D^1$ and
$\widehat{D}(S_{F,D_1}) = D^2$.
This implies that in both cases, $D^1$ contains vertices from at least two maximal proper
strong modules of $D_1$ and, consequently, contains a vertex of $D_1 \cap \Gm(T')$
as $D_1 \setminus \Gm(T') \subseteq M_1$.
Also, $D^2 \supseteq D_2$ and, since $\beta(t_2) \setminus S \subseteq D_2$ and $\beta(t_2) \subseteq \Gm(T')$, $D_2 \cap \Gm(T') \neq \emptyset$ and 
we have that $D^2$ contains a vertex of $D_2 \cap \Gm(T')$. 
We infer that $s^1s^2$ is an edge of $T'$. Furthermore, if $S_{F,D_1} \neq S$ (i.e., $s^1s^2 \neq t_1t_2$), then
$S_{F,D_1}$ contains a vertex of $D_1$ and, consequently, $s^1s^2$ is an edge of $T_1'$.

We proceed to constructing a footprint $(S_{F,D_i}, \Om_{F,D_i}, M_{F,D_i}, q_{F,D_i})$
of $D_i$ as in the claim statement.
Note that we can take $\Om_{F,D_i} = \beta(s^1)$.
It remains to show that we can take $M_{F,D_1} = M_1$ in Lemma~\ref{lem:MFD}
for a footprint of the mesh component $D_1$ of $G-S$,
as the latter choice of $q_{F,D_i}$ is arbitrary from $D(S_{F,D_1}) \setminus M_{F,D_1}$.
That is, we need to show that the maximal proper strong module $M_1$ of $D_1$
that contains $D_1 \setminus \Gm(T')$ also contains all components $D'$ of $G-\beta(s^1)$
with $N(D')\not\subseteq S_{F,D_1}$. 

Recall that $D^1 = D(S_{F,D_1})$.
Since $\beta(s^1) \setminus S_{F,D_i} \subseteq D^1$, all such connected components $D'$
are subsets of $D^1 \setminus \beta(s^1)$. 
By the definition of $S_{F,D_1}$, every connected component $D'$ of $G-\beta(s^1)$ 
with $D' \subseteq D^1$ lies in a single maximal proper strong module of $D_1$. 
Hence, as $\Quo(D_1)$ is a clique, 
there exists a maximal proper strong module of $D_1$ that contains $D^1 \setminus \beta(s^1)$.

As $s^1 \in V(T')$, we have $\beta(s^1) \cup S \subseteq \Gm(T')$, that is,
$D_1 \setminus \Gm(T') \subseteq D_1 \setminus \beta(s^1)$.
Note that $D^1 = D_1 \setminus S_{F,D_1}$, and hence $D^1 \setminus \beta(s^1) = D_1 \setminus \beta(s^1)$. Hence, $D_1 \setminus \Gm(T') \subseteq D^1 \setminus \beta(s^1)$.
Consequently, $M_1$ is the maximal proper strong module of $D_1$ that contains $D^1 \setminus \beta(s^1)$. This finishes the proof that we can take $M_{F,D_1} = M_1$ for a footprint of $D_1$
in $G-S$, finishing the proof of the claim.
\cqed\end{proof}

\begin{claim}\label{cl:FfX:S-cut}
At most one of the components $D_1$ and $D_2$
from Claim~\ref{cl:FfX:D1D2} contains more than one component of $G-\Gm(T')$.
Furthermore, if the component $D_i$ is not a mesh for $i=1,2$, then
the component $D_{3-i}$ contains at most one component of $G-\Gm(T')$.
\end{claim}
\begin{proof}
Assume one of the two claims of the lemma is not true for separator $S$ with components $D_1$ and $D_2$.
Since the family $\Ff_4$ from Lemma~\ref{lem:capture-nonmesh} is taken into account
in the construction of $\Ss$, we have that either $D_1$ or $D_2$ is a mesh.
Note that if there exists $i \in \{1,2\}$ such that $D_i$ is a mesh but $D_i$ contains at most one component of $G-\Gm(T')$, then both of the claims are automatically satisfied.
Hence, we are interested only in the following case: at least one of $D_1$ or $D_2$ is a mesh,
  and, for every $i=1,2$ such that $D_i$ is a mesh, $D_i$
  contains at least two connected components of $G-\Gm(T')$.

For $i=1,2$ such that $D_i$ is a mesh, we invoke Claim~\ref{cl:FfX:S-Di-sep}, yielding
and edge $e_i \in E(T_i') \cup \{t_1t_2\}$ with $\adh(e_i) = S_{F,D_i}$.
Furthermore, if $D_i$ is not a mesh, we define $e_i = t_1t_2$.
We define a potential segment $\Gm$ to be equal $S$ if
$e_1=e_2=t_1t_2$ and otherwise $\Gm = \Gm(T'')$ where $T''$ is the subpath of $T'$
between $e_1$ and $e_2$.
We would like to replace $F$ with $F[\Gm \to F_\Gm]$ for some minimal completion $F_\Gm$
of $\torso(\Gm)$ that includes at least one separator of $\Ss$; such a replacement would
contradict the maximality of $F$.
Furthermore, note that as $T''$ does not contain any edge that corresponds to a separator
from $\Ss$, every separator $S' \in \Ss$ that is a minimal separator of $G+F$
is either not contained in $\torso(\Gm)$ or contained in $N(D)$ for some $D \in \cc(G-\Gm)$.
Consequently, for any such replacement $F[\Gm \to F_\Gm]$, every minimal separator
$S' \in \Ss$ that is a minimal separator of $G+F$ will also be a minimal separator
of $G+F[\Gm \to F_\Gm]$. Our goal is to introduce at least one new minimal separator
from $\Ss$ into $G+F[\Gm \to F_\Gm]$.

Without loss of generality, let $D_1$ be a mesh. 
We distinguish cases depending on whether $D_2$ is a mesh or not.

\medskip

\noindent\textbf{Case 1: $D_2$ is not a mesh.}
Here, we have $\Gm = S_{F,D_1} \cup S$, and we perform
the footprint replacement for the component $D_1$ and the footprint $(S_{F,D_1}, \Om_{F,D_1}, M_{F,D_1}, q_{F,D_1})$
whose existence is asserted by Claim~\ref{cl:FfX:S-Di-sep}.
That is, we pick $F_\Gm$ to be any minimal completion of $\torso(\Gm)$
that keeps the separator $S'_{F,D} = \partial (M_{F,D_1} \cup \{q_{F,D_1}\})$ as a minimal separator
of $G+F[\Gm \to F_\Gm]$.
Since a fuzzy version of $D_1$ is included in the family $\Ff_5$
output by Lemma~\ref{lem:capture-mesh-nomesh}, we have $S'_{F,D} \in \Ss$.

Note that this case completely covers the second statement of the lemma.

\medskip

\noindent\textbf{Case 2: $D_2$ is a mesh.}
Let $(S_{F,D_i}, \Om_{F,D_i}, M_{F,D_i}, q_{F,D_i})$ be a footprint of $D_i$ for $i=1,2$, whose existence is asserted by Claim~\ref{cl:FfX:S-Di-sep}.

We have $\Gm = S_{F,D_1} \cup S \cup S_{F,D_2}$. Consequently,
we have that the connected components of $G-\Gm$ are $D(S_{F,D_1})$, 
$D(S_{F,D_2})$ and all connected components of $G-(S \cup D_1 \cup D_2)$.
Henceforth, we denote $R = V(G) \setminus (S \cup D_1 \cup D_2)$.

By Lemma~\ref{lem:rozrywanie}, the family $\Ff_9$ contains 
a pair $(Z_1,Z_2)$ for
the separator $S$ with components $D_i$, modules $M_{F,D_i}$, and elements $q_{F,D_i}$. 

The first observation is that $\partial Z_i \subseteq \Gm$ for $i=1,2$.
Indeed, first note that Corollary~\ref{lem:SFD-def} 
applied to the footprint $(S_{F,D_i}, \Om_{F,D_i}, M_{F,D_i}, q_{F,D_i})$,
together with the assumption
$M_{F,D_i} \subseteq Z_i$, $q_{F,D_i} \in Z_i$ implies that
$$D(S_{F,D_i}) \subseteq N[A'_{F,D_i}] \setminus \partial A'_{F,D_i}
\subseteq N[Z_i] \setminus \partial Z_i.$$
Together with the fact that there are no edges between $Z_1$ and $Z_2$, we obtain
that $D(S_{F,D_{3-i}})$ is disjoint from $N[Z_i]$. 
Finally, the property that $Z_i$ is disjoint from $N(R)$ implies that $N[Z_i]$ is disjoint
from $R$. This finishes the proof that $\partial Z_i \subseteq \Gm$.

Let 
$$\Ss' = \{N(D) : D \in \cc(G-N[Z_1]) \cup \cc(G-N[Z_2])\}.$$
We claim that the separators of $\Ss'$ are pairwise noncrossing. 
Indeed, consider two separators $S^1 = N(D^1)$ and $S^2 = N(D^2)$ from $\Ss'$.
By symmetry, we can assume $D^1 \in \cc(G-N[Z_1])$ and $D^2 \in \cc(G-N[Z_j])$
for some $j=1,2$. Then, the assumption that there are no edges from $Z_1$ to $Z_2$ implies
that $S^1 \cap D^2 = \emptyset$.
Consequently, there exists a unique connected component of $G-S^1$
that contains $D^2$ and $N(D^2) \setminus S^1$, witnessing that $S^2$ and $S^2$ do not cross.

Since $\Ss'$ are pairwise noncrossing family of minimal separators
contained in $\Gm$, Lemma~\ref{lem:force-separator}
asserts an existence of a minimal chordal completion $F_\Gm$ of $\torso(\Gm)$ such that
$\Ss'$ is a subset of the family of all minimal separators of $G+F[\Gm \to F_\Gm]$. 
This finishes the proof, as $\Ss' \subseteq \Ss$.
\cqed\end{proof}

Claim~\ref{cl:FfX:S-cut} implies the following:
\begin{claim}\label{cl:FfX:t0}
There exists a node $t_0 \in V(T')$ such that every connected component of $G-\beta(t_0)$
contains at most one component of $G-\Gm(T')$.
\end{claim}
\begin{proof}
For an edge $t_1t_2 \in E(T')$ and components $D_1$, $D_2$ defined as in Claim~\ref{cl:FfX:S-cut},
we say that $t_i$ is a \emph{big} side of $t_1t_2$ if $D_i$ contains more than one connected component of $G-\Gm(T')$,
and \emph{small} otherwise.
Claim~\ref{cl:FfX:S-cut} asserts that every edge of $T'$ has at most one big side.
Furthermore, Claim~\ref{cl:FfX:D1D2} implies that if $t_1,t_2,t_3$ are three consecutive nodes of $T'$, then
if $t_2$ is a small side of $t_2t_3$, then $t_1$ is a small side of $t_1t_2$.
Consequently, there exists a node $t_0 \in V(T')$ such that for every $t' \in N_{T'}(t_0)$ we have that
$t'$ is a small side of the edge $t't_0$; such a node satisfies the conditions of the claim.
\cqed\end{proof}

We need also one more simple observation:
\begin{claim}\label{cl:FfX:DGamma}
For every $D \in \cc(G-\Gamma(T'))$, the component $D$ belongs to $\Ff$.
\end{claim}
\begin{proof}
There exists a unique component $T^\Gm$
of $T-V(T')$ such that the vertices of $D^\Gm$ appear only in bags of $T^\Gm$.
Let $e^\Gm$ be the edge of $T$ connecting $T'$ with $T^\Gm$.
Then $e^\Gm \in E_\Ss$ by the definition of $T'$.
Furthermore, $N(D^\Gm) \subseteq \adh(e^\Gm)$
and $D^\Gm$ is a connected component of $G-\adh(e^\Gm)$.
We infer that $D^\Gm \in \Ff $ by the construction of $\Ff$.
\cqed\end{proof}

We now distniguish a number of cases, depending on the structure around the node $t_0$ and the maximal clique $\beta(t_0)$.
For brevity, we denote $\Gamma := \Gamma(T')$.

\medskip
\noindent\textbf{Case 1: $t_0$ is an endpoint of $T'$.}
Informally speaking, in this case
we have all but one separators of $\{N(D) : D \in \cc(G-\beta(t_0))\}$ in $\Ss$, so Lemma~\ref{lem:monster}
does the job for us.

More formally, let $t_1$ be the unique neighbor of $t_0$ in $T'$, let $S = \beta(t_0) \cap \beta(t_1)$
be the minimal separator corresponding to the edge $t_0t_1$, and let $D_1$
be the unique component of $G-S$ that contains $\Gm \setminus \beta(t_0)$ (its existence is asserted by Claim~\ref{cl:FfX:D1D2}).
As $D_1$ is a component of $G-\beta(t_0)$ as well, there exists at most one component $D_1^\Gm$ of $G-\Gm$
that is contained in $D_1$; we take $D_1^\Gm = \emptyset$ if no such component exists.

Let $D \in \cc(G-\beta(t_0)) \setminus \{D_1\}$; we claim that $N(D) \in \Ss$.
If $N(D) \subseteq N(D_1)$, then $N(D) \in \Ss$ as we have taken into account the family 
$\Ff_1$ from Lemma~\ref{lem:two-not-whole};
Otherwise, $N(D)$ contains an element of $\beta(t_0) \setminus S$, $D$ is not a connected component of $G-\beta(t_1)$, 
and, consequently, $N(D) \subseteq S'$ for some minimal separator $S'$ associated with an edge $e$ of $T$ incident 
with $t_0$, but $e \neq t_0t_1$. Hence, if $N(D) \subsetneq S'$, we have $N(D) \in \Ss$
again due to the family $\Ff_1$ from Lemma~\ref{lem:two-not-whole},
and otherwise $S' = N(D) \in \Ss$ as $t_0$ is an endpoint of the path $T'$.

Let $\Ff = \bigcup \{\cc(G-S) : S \in \Ss\}$.
We infer that the family $\Ff_7(\Ff)$ from Lemma~\ref{lem:monster} contains $D_1$, while
$\Ff$ contains $D_1^\Gm$ if it is not empty. 
Hence, if we include $N[D_1] \setminus D_1^\Gm$
for every $D_1 \in \Ff_7(\Ff)$ and $D_1^\Gm \in \Ff \cup \{\emptyset\}$ in the desired family $\Ff_X$
as well as $\Ff_{\rec,1}(\Ff_7(\Ff))$ from Lemma~\ref{lem:recover1},
the family $\Ff_X$ contains both $\beta(t_0)$ and $N[D_1] \setminus D_1^\Gm$.
This finishes the proof, as the latter is the segment $\Gm(T'-t_0)$.

\medskip
In the remaining cases, $t_0$ is not an endpoint of $T'$, and thus has
two neighbors $t_1,t_2 \in N_{T'}(t_0)$. 
For $i=1,2$, let $S_i = \beta(t_0) \cap \beta(t_i)$ be the minimal separator associated with the edge $t_0t_i$,
let $T_i'$ be the connected component of $T'-t_0$ that contains $t_i$,
let $D_i$ be the connected component of $G-\beta(t_0)$ that contains $\bigcup_{t' \in T_i'} \beta(t') \setminus \beta(t_0)$
(cf. Claim~\ref{cl:FfX:D1D2}), 
  and let $D_i^\Gm$ be the connected component of $G-\Gm$ that is contained in $D_i$
(or $D_i^\Gm = \emptyset$ if no such component exists).
By Claim~\ref{cl:FfX:DGamma}, $D_i^\Gm \in \Ff \cup \{\emptyset\}$ for $i=1,2$.

\medskip
\noindent\textbf{Case 2: both $D_1$ and $D_2$ are meshes.}
Here, we rely on Lemma~\ref{lem:Omplus} to guess $D_1 \cup D_2 \cup \beta(t_0)$.

More formally, since $S_1,S_2 \notin \Ss$, neither $D_1$ nor $D_2$ belongs
to the family $\Ff_8$ from Lemma~\ref{lem:Omplus}.
However, then $\beta(t_0) \cup D_1 \cup D_2 \in \Ff_8$.
As $D_i^\Gm \in \Ff \cup \{\emptyset\}$, we can resolve this case by inserting into $\Ff_X$ the set
$A \setminus (D_1^\Gm  \cup D_2^\Gm)$ for every choice of $A \in \Ff_8$
and $D_1^\Gm, D_2^\Gm \in \Ff \cup \{\emptyset\}$, as this set is equal to $\Gm$ provided $A=\beta(t_0) \cup D_1 \cup D_2$.

\medskip
\noindent\textbf{Case 3: $D_1$ or $D_2$ is not a mesh.}
Without loss of generality, assume that $D_1$ is not a mesh.
Let $\overline{D}_2$ be the connected component of $G-S_1$ that contains $D_2$.
Note that $\overline{D}_2$ is full to $S_1$.
By Claim~\ref{cl:FfX:S-cut}, $\overline{D}_2$ contains at most one component of $G-\Gm$;
denote it by $\overline{D}_2^\Gm$ if it exists, and take $\overline{D}_2^\Gm = \emptyset$ otherwise.

We use the Separator Covering Lemma (Lemma~\ref{lem:covering-general}) to pick
two vertices $p_1,q_1 \in D_1$ and $p_2,q_2 \in \overline{D}_2$ such that
$S_1 \subseteq N[p_1,q_1,p_2,q_2] \subseteq S_1 \cup D_1 \cup \overline{D}_2$.

Let $R = V(G) \setminus (D_1 \cup S_1 \cup \overline{D}_2)$.
We claim that
\begin{equation}\label{eq:case3-1}
V(G) = R \cup N[p_1,q_1,p_2,q_2] \cup N[D_1^\Gm] \cup N[\overline{D}_2^\Gm].
\end{equation}
Assume the contrary, let $x$ be a vertex not included in any of the four sets of the right hand side.
Recall that $I \cap \Gm = \emptyset$, which implies that $I \cap D_1 = I \cap D_1^\Gm$ and $I \cap \overline{D}_2 = I \cap \overline{D}_2^\Gm$.
Since $N(R) \subseteq S_1 \subseteq N[p_1,q_1,p_2,q_2]$, we have that $N[x] \cap I = \emptyset$, contradicting the maximality of $I$.
This proves~\eqref{eq:case3-1}.

Hence, as the sets $N[p_1,q_1,p_2,q_2]$, $N[D_1^\Gm]$, and $N[\overline{D}_2^\Gm]$ are disjoint with $R$, we have
\begin{equation}\label{eq:case3-2}
R = V(G) \setminus (N[p_1,q_1,p_2,q_2] \cup N[D_1^\Gm] \cup N[\overline{D}_2^\Gm]).
\end{equation}
That is,~\eqref{eq:case3-2} asserts that the set $R$ is fully determined by the vertices $p_1,q_1,p_2,q_2$ and components
$D_1^\Gm$ and $\overline{D}_2^\Gm$.
Recall that $D_1^\Gm, \overline{D}_2^\Gm \in \Ff \cup \{\emptyset\}$ due to Claim~\ref{cl:FfX:DGamma}. Consequently, we have $\Gm \in \Ff_X$ 
if we include into $\Ff_X$ a set
$$(N[p_1,q_1,p_2,q_2] \cup N[D_1^\Gm] \cup N[\overline{D}_2^\Gm]) \setminus (D_1^\Gm \cup \overline{D}_2^\Gm)$$
for every choice $p_1,q_1,p_2,q_2 \in V(G)$ and $D_1^\Gm, \overline{D}_2^\Gm \in \Ff \cup \{\emptyset\}$.

\medskip
This finishes the proof of the lemma.
\end{proof}

We conclude the proof 
of Theorem~\ref{thm:main} as follows.
We apply Lemma~\ref{lem:replacement} with the set $\Xx = \Ff_X$,
where $\Ff_X$ is given by Lemma~\ref{lem:FfX}, obtaining a family
$\Ff_\comp(\Ff_X)$ of polynomial size. 
Lemma~\ref{lem:FfX} asserts that $\Ff_X$ satisfies the prerequisites for 
Lemma~\ref{lem:replacement}, and thus $\Ff_\comp(\Ff_X)$ satisfies
the properties promised by Theorem~\ref{thm:main}.

\section{Conclusions}\label{sec:conc}
In this paper we have shown that the {\sc{Maximum Weight Independent Set}} problem is polynomial-time solvable on the class of $P_6$-free graphs.
The obvious open question is what the complexity of the problem is on $P_7$-free graphs and beyond.
Unfortunately, it seems that many of the most basic tools used in this work break in the $P_7$-free setting, with the most important case of the Separator Covering Lemma (Lemma~\ref{lem:covering-simple}).
Namely, it is simply not true that any minimal separator in a $P_7$-free graph can be covered by the union of neighborhoods of a constant number of vertices lying outside of this separator.
Curiously, it turns out that if we allow vertices covering the separator to lie within it, then the statement is again true in $P_7$-free graphs, and even can be generalized to covering PMCs.
However, such a relaxed statement seems of little use for our goal of capturing maximal cliques in an $I$-free minimal chordal completion.
Moreover, this relaxation turns out to be simply not true in $P_8$-free graphs.
All of the results claimed above will be covered in a separate note, in preparation.

Therefore, as far as we see some hope of lifting some of our techniques to $P_7$-free graphs, tackling $P_8$-free graphs seems to require a complete change of methodology.

\bibliographystyle{abbrv}

\bibliography{references}

\end{document}